\documentclass[11pt,english]{article}
\usepackage[T1]{fontenc}
\usepackage[latin9]{inputenc}
\usepackage[letterpaper]{geometry}
\usepackage{color}
\usepackage{float}
\usepackage{amsmath}
\usepackage{amssymb}
\usepackage{graphicx}

\makeatletter

\newcommand{\noun}[1]{\textsc{#1}}
\providecommand{\tabularnewline}{\\}

\usepackage{amsfonts}

\usepackage{epsfig}

\usepackage{mathrsfs}

\newtheorem{theorem}{Theorem}
\newtheorem{proposition}[theorem]{Proposition}
\newtheorem{lemma}[theorem]{Lemma}
\newtheorem{corollary}[theorem]{Corollary}
\newtheorem{definition}[theorem]{Definition}
\newtheorem{example}[theorem]{Example}
\newtheorem{remark}[theorem]{Remark}

\newenvironment{proof}[1][Proof]{\noindent \textbf{#1.} }{\qedsymbol}
\newcommand{\qedsymbol}{\hspace{\fill}\rule{1.5ex}{1.5ex}}
\input{tcilatex}

\newcommand{\wt}{\widetilde}

\newcommand{\sol}{\mathop{\rm SOL}}

\def\b0{\mbox{\boldmath $0$}}

\newcounter{algo}\newenvironment{algo}[2]{\refstepcounter{algo}\label{#2}   \begin{center}
\begin{minipage}{0.9\textwidth}   \hrule\smallskip
\textbf{Algorithm \thealgo: #1}
\par\smallskip\hrule\smallskip\ignorespaces}{\par\smallskip\hrule
\end{minipage}
\end{center}
}
\input{tcilatex}

\def\baselinestretch{1.15}

\usepackage{subfigure}

\makeatother

\usepackage{babel}
\begin{document}

\title{\vspace{-2cm}
 \textbf{\huge{Real and Complex }}\\
\textbf{\huge{ Monotone Communication Games}}}

\author{Gesualdo Scutari$^{1}$, Francisco Facchinei$^{2}$, Jong-Shi Pang$^{3}$,
and Daniel P. Palomar$^{4}${\small{\vspace{-0.4cm}
 }}\\
 \\
{\footnotesize{E-mail: gesualdo@buffalo.edu, facchinei@dis.uniroma1.it,
jspang@uiuc.edu,}} {\footnotesize{palomar@ust.hk.}} \\
 {\small{{} {} {} {} $^{1}$}}{\footnotesize{COMSENS Research
Center, Dpt. of Electrical Engineering, State University of New York
at Buffalo, Buffalo, USA.}} \\
 {\small{{} {} {} {} $^{2}$}}{\footnotesize{Dpt. of Computer,
Control, and Management Engineering Antonio Ruberti, Univ. of Rome
``La Sapienza\textquotedbl{}, Rome, Italy.}}\\
 {\small{{} {} {} {} $^{3}$}}{\footnotesize{Dpt. of Industrial
and Enterprise Systems Engineering, University of Illinois at Urbana-Champaign,
Urbana, USA.}}\\
 {\small{{} {} {} {} $^{4}$}}{\footnotesize{Dpt. of Electronic
and Computer Eng., Hong Kong Univ. of Science and Technology, Hong
Kong.}}%
\thanks{{\footnotesize{The work of G. Scutari is supported by the National
Science Foundation Grant No. CNS-1218717. The work of J.-S. Pang is
based on research supported by the U.S. National Science Foundation
grant No. CMMI 0969600. The work of D. P. Palomar was supported by
the Hong Kong RGC 617810 research grant. Part of this paper has been
presented at the 31st Annual IEEE International Conference on Computer
Communications (IEEE INFOCOM 2012).}}%
}\\
 }

\date{{\normalsize{Submitted to }}\emph{\normalsize{IEEE Transactions on
Information Theory}}{\normalsize{, October 8, 2012; Revised Dec. 11,
2013.}}{\small{{}\vspace{-1cm}
 }}}

\maketitle
 
\begin{abstract}
Noncooperative game-theoretic tools have been increasingly used to
study many important resource allocation problems in communications,
networking, smart grids, and portfolio optimization. In this paper,
we consider a general class of convex Nash Equilibrium Problems (NEPs),
where each player aims at solving an arbitrary smooth convex optimization
problem. Differently from most of current works, we do not assume
any specific structure for the players' problems, and we allow the
optimization variables of the players to be matrices in the complex
domain. Our main contribution is the design of a novel class of distributed
(asynchronous) best-response- algorithms suitable for solving the
proposed NEPs, even in the presence of \emph{multiple} solutions.
The new methods, whose convergence analysis is based on Variational
Inequality (VI) techniques, can select, among all the equilibria of
a game, those that optimize a given performance criterion, at the
cost of limited signaling among the players. This is a major departure
from existing best-response algorithms, whose convergence conditions
imply the uniqueness of the NE. Some of our results hinge on the use
of VI problems directly in the complex domain; the study of these
new kind of VIs also represents a noteworthy innovative contribution.
We then apply the developed methods to solve some new generalizations
of SISO and MIMO games in cognitive radio systems, showing a considerable
performance improvement over classical pure noncooperative schemes. 
\end{abstract}

\section{Introduction and Motivation{\small{\label{sec:Introduction-and-Motivation}}}}

In recent years, there has been a growing interest in the use of noncooperative
games to model and solve resource allocation problems in communications
and networking, wherein the interaction among several agents is by
no means negligible and centralized approaches are not suitable. Examples
are power control and resource sharing in wireless/wired peer-to-peer
networks, cognitive radio systems (e.g., \cite{Yu-Ginis-Cioffi_jsac02,Luo-Pang_IWFA-Eurasip,ShumandLeungSung_JSAC07,Scutari-Palomar-Barbarossa_SP08_PI,Scutari-Palomar-Barbarossa_AIWFA_IT08,Pang-Scutari-Facchinei-Wang_IT_08,Scutari-Palomar-Barbarossa_JSAC08,Scutari-Palomar-Barbarossa_SP08_IWFA,Scutari-Palomar-Barbarossa_SPMag08,Scutari-Palomar-Facchinei-Pang_SPMag09,Scutari-Palomar_SP09_CR_GTMIMO,RazaviyaynLuoTsengPang_MP11}),
distributed routing, flow and congestion control, and load balancing
in communication networks (e.g., \cite{Altman-Wynter_survey_GT-TLC-Net04,Altman-Boulogne-ElAzouzi-Jimenez-Wynter_survey_GT-TLC06,Huang-Cendrillon-Chiang-Moonen_SP07}
and references therein), and smart grids (see \cite{Saad-Han-Poor-Basar_SPMag12,Atzeni-Ordonez-Scutari-Palomar_TSG12}
and references therein). Two recent special issues on the subject
are \cite{HuangPalomarMandayamWalrandWickerBasarGTJSAC08,JorswieckLarssonLuisePoor_GTSPMag09}.

Among the variety of models and solution concepts proposed in the
literature, the Nash Equilibrium Problem (NEP) plays a central role
and has been used mostly to model interactions among individuals competing
selfishly for scarce resources. In a NEP there is a finite number
$I$ of players; each player $i$ makes decisions on a set of variables
$\mathbf{x}_{i}$ belonging to a given feasible set $\mathbf{x}_{i}\in\mathcal{Q}_{i}$.
The goal of each player $i$ is to minimize his own objective function
$f_{i}(\mathbf{x}_{i},\mathbf{x}_{-i})$ over $\mathcal{Q}_{i}$ while
anticipating the reactions $\mathbf{x}_{-i}\triangleq(\mathbf{x}_{j})_{j\neq i=1}^{I}$
from the rivals: 
\begin{equation}
\begin{array}{ll}
\limfunc{minimize}\limits _{\mathbf{x}_{i}} & f_{i}(\mathbf{x}_{i},\,\mathbf{x}_{-i})\\[5pt]
\text{subject to} & \mathbf{x}_{i}\in{\cal Q}_{i}.
\end{array}\label{eq:NEP_intro}
\end{equation}
The NEP is the problem of finding a vector $\mathbf{x}^{\star}\triangleq(\mathbf{x}_{i}^{\star})_{i=1}^{I}$
such that each $\mathbf{x}_{i}^{\star}$ belongs to $\mathcal{Q}_{i}$
and solves the player's problem (given $\mathbf{x}_{-i}^{\star}$):
\begin{equation}
f_{i}(\mathbf{x}_{i}^{\star},\,\mathbf{x}_{-i}^{\star})\leq f_{i}(\mathbf{x}_{i},\,\mathbf{x}_{-i}^{\star}),\quad\forall\mathbf{x}_{i}\in\mathcal{Q}_{i}.\label{eq:NE}
\end{equation}
Such a point $\mathbf{x}^{\star}$ is called a Nash Equilibrium (NE)
or, more simply, a solution of the NEP. In words, a NE is a feasible
strategy profile $\mathbf{x}^{\star}$ such that no \emph{single}
player can benefit from a \emph{unilateral} deviation from $\mathbf{x}_{i}^{\star}$.

In this paper we focus on NEPs in the general form (\ref{eq:NEP_intro}),
in the following setting: i) the optimization variables of each player
can be either real vectors or complex matrices; ii) each optimization
problem in (\ref{eq:NEP_intro}) is convex for any given feasible
$\mathbf{x}_{-i}$; and iii) players' objective functions are continuously
differentiable in all the variables (more precisely, functions of
complex variables are assumed to be $\mathbb{R}$-differentiable,
see Sec. \ref{sec:VI_complex_domain}). We will term such a game (\emph{real}
or\emph{ complex}) \emph{player-convex} NEP. Note that assumptions
ii) and iii) are mild and quite standard in the literature, see for
example \cite{Yu-Ginis-Cioffi_jsac02,Luo-Pang_IWFA-Eurasip,ShumandLeungSung_JSAC07,Scutari-Palomar-Barbarossa_SP08_PI,Scutari-Palomar-Barbarossa_AIWFA_IT08,Pang-Scutari-Facchinei-Wang_IT_08,Scutari-Palomar-Barbarossa_JSAC08,Scutari-Palomar-Barbarossa_SP08_IWFA,Scutari-Palomar-Barbarossa_SPMag08,Scutari-Palomar-Facchinei-Pang_SPMag09,Scutari-Palomar_SP09_CR_GTMIMO,HuangPalomarMandayamWalrandWickerBasarGTJSAC08,JorswieckLarssonLuisePoor_GTSPMag09}
where special instances of the player-convex NEP (\ref{eq:NEP_intro})
are studied. The convexity assumption ii) makes the NEP numerically
tractable (a NE may not even exist otherwise) while, to date, the
differentiability of the players' functions seems indispensable to
analyze \emph{distributed} solution methods \cite{Facchinei_Pang_VI-NE_bookCh_09,FacchineiPangScutariLampariello_MP11},
unless the game has a very specific structure, like in potential or
supermodular games; see, e.g., \cite{FacchineiPiccialliSciandrone_IT_08,Huang-Berry-Honig_JSAC06,Facchinei-Kanzow_GNEP}
and references therein. Motivated by recent applications of noncooperative
models in MIMO communications \cite{Scutari-Palomar-Barbarossa_JSAC08,Scutari-Palomar-Barbarossa_SP08_IWFA,Scutari-Palomar-Facchinei-Pang_SPMag09,Scutari-Palomar_SP09_CR_GTMIMO,Arslan-Fatih-Demirkol-Song_WC07},
we also allow, according to i), players' optimization variables to
be complex matrices, which significantly enlarges the range of applicability
of model (\ref{eq:NEP_intro}). To the best of our knowledge, this
is the first work where a \emph{complex} NEP in the general form (\ref{eq:NEP_intro})
is considered.

While the solution analysis (e.g., solution existence) of a \emph{real}
player-convex NEP relies on standard results in game theory (see,
e.g., the seminal work \cite{Rosen_econ65}, or \cite{Facchinei_Pang_VI-NE_bookCh_09}
for more recent results), the development of \emph{distributed }solution
algorithms is much more involved. The goal of this paper is to address
this difficult task in the broad setting described above. We are interested
in the design and analysis of (possibly) asynchronous iterative \emph{best-response
}algorithms, suitable for solving real and complex player-convex NEPs,
even in the presence of multiple NEs. By ``best-response'' algorithms
we mean iterative schemes where the players iteratively choose the
(feasible) strategy that minimizes their cost functions, given the
actions of the other players; the reason for our emphasis on best-response
schemes will be described shorty.\vspace{-0.2cm}

\subsection{Literature review\vspace{-0.2cm}
}

The study of iterative algorithms for (special cases of) player-convex
NEPs has been addressed in a number of papers, under different settings
and assumptions; the main features and limitations of current state-of-the-art
approaches are discussed next.

A first class of papers is composed of works motivated by specific
applications, some examples are \cite{Yu-Ginis-Cioffi_jsac02,Luo-Pang_IWFA-Eurasip,ShumandLeungSung_JSAC07,Scutari-Palomar-Barbarossa_SP08_PI,Scutari-Palomar-Barbarossa_AIWFA_IT08,Scutari-Palomar-Barbarossa_JSAC08,Scutari-Palomar-Barbarossa_SP08_IWFA,Scutari-Palomar-Barbarossa_SPMag08,Scutari-Palomar-Facchinei-Pang_SPMag09,Scutari-Palomar_SP09_CR_GTMIMO},
where different resource allocation problems in communications are
modelled as noncooperative games and solved via iterative algorithms;
all these formulations are special cases of the NEP (\ref{eq:NEP_intro}).
A key feature of all these models is that the best-response of each
player (i.e., the optimal solution of each player's optimization problem)
is unique and can be expressed in closed form; this simplifies enormously
the application of standard fixed-point arguments to the study of
the convergence of best-response algorithms. A monotonicity-based
approach is instead used in \cite{Uryasev-Rubinstein,Yates_JSAC95}.
Even though algorithms in \cite{Basar-Olsder_DGT,Li-Basar_Automatica87,Uryasev-Rubinstein}
do not require a closed form solution of players' optimization problems,
they can be computationally very demanding and the convergence conditions
are based on assumption whose verification for games arising from
realistic applications remains elusive. Last but not least, convergence
conditions of the algorithms proposed in all the aforementioned papers
imply the uniqueness of the NE.

A more general and powerful methodology suitable for studying noncooperative
games is offered by the theory of finite-dimensional Variational Inequalities
(VIs) \cite{Facchinei-Pang_FVI03}. VI and complementarity problems
have a long history and have been well documented in the literature
of operation research \cite{Facchinei-Pang_FVI03}, but only recently
they have been brought to the attention of the signal processing,
communications, and networking communities \cite{Luo-Pang_IWFA-Eurasip,Scutari-Palomar-Barbarossa_SP08_PI,Pang-Scutari-Facchinei-Wang_IT_08,Scutari-Palomar-Facchinei-Pang_SPMag09,Pang-Scutari-Palomar-Facchinei_SP_10,Scutari-Palomar-Facchinei-Pang_SPMag10}.
Given a subset $\mathcal{K}$ of $\mathbb{R}^{n}$ and a vector-valued
function $\mathbf{F}:\mathcal{K}\rightarrow\mathbb{R}^{n}$, the VI
problem, denoted by VI$(\mathcal{K},\mathbf{F})$, consists in finding
a point $\mathbf{x}^{\star}\in\mathcal{K}$ such that 
\begin{equation}
\left(\mathbf{x}-\mathbf{x}^{\star}\right)^{T}\mathbf{F}(\mathbf{x}^{\star})\geq0\quad\forall\mathbf{x}\in\mathcal{K}.\label{eq:VI_intro}
\end{equation}
The VI approach to \emph{real} player-convex NEPs as in (\ref{eq:NEP_intro})
hinges on an easy equivalence with the (partitioned) VI problem VI$(\mathcal{K},\mathbf{F})$
in (\ref{eq:VI_intro}), with $\mathcal{K}=\prod_{i=1}^{I}{\cal Q}_{i}$
and $\mathbf{F}=(\nabla_{\mathbf{x}_{i}}f_{i}(\mathbf{x}))_{i=1}^{I}$
(intended to be a column vector), where $\nabla_{\mathbf{x}_{i}}f_{i}(\mathbf{x})$
denotes the gradient of $f_{i}$ with respect to $\mathbf{x}_{i}$.
Based on this equivalence, one can solve a real player-convex NEP
by focusing on the associated VI problem and taking advantage of the
many (centralized and distributed) solution methods available in the
literature for partitioned VIs \cite[Vol. II]{Facchinei-Pang_FVI03}.

In the effort of obtaining distributed schemes for NEPs, researchers
have focused on so called\emph{ projection algorithms} \cite[Ch. 12]{Facchinei-Pang_FVI03}
for partitioned VIs; see, e.g., \cite{Konnov_VI_book,Yin-Shanbhag-Mehta_TAC10,KANNAN-SHANBHAG_sub10}
(and also \cite{Rosen_econ65,Shamma-Arslan_TAC05} for related approaches).
However, these solution methods suffer from some drawbacks, which
strongly limit their applicability in practice, especially in the
design of wireless systems. First, they are not ``incentive compatible'',
meaning that selfish users may deviate from them, unless they are
imposed by some authority as a protocol to follow. Second, and most
importantly, they generally converge very slowly; this has been observed
in a number of different applications (see, e.g., numerical results
in \cite{Yin-Shanbhag-Mehta_TAC10,KANNAN-SHANBHAG_sub10,Shamma-Arslan_TAC05}
and Fig. \ref{fig:fig_intro} in Sec. \ref{sub:Numerical-Results}).

A different approach to the design of algorithms for partitioned VIs
has been followed in \cite{Pang_MP85,Pang-Chan_MP82}, where the authors
investigated the local and global convergence of various iterative
synchronous methods that decompose the original VI problem into a
sequence of simpler lower-dimensional VI subproblems. Unfortunately
the convergence analysis in \cite{Pang_MP85,Pang-Chan_MP82}, based
on contraction arguments, leads to abstract convergence conditions,
whose verification in practice seems not possible. Easier conditions
to be checked have been obtained recently in \cite{Facchinei_Pang_VI-NE_bookCh_09}
for simultaneous best-response algorithms, still using the VI approach.
However, conditions in \cite{Facchinei_Pang_VI-NE_bookCh_09,Pang_MP85,Pang-Chan_MP82}
are applicable only to a restricted class of \emph{real} NEPs; they
indeed imply the (uniformly) strongly convexity of the players' cost
functions and the uniqueness of the NE. In the presence of multiple
solutions, the distributed computation of even a single NE of real/\emph{complex}
NEPs via best-response algorithms becomes a difficult and unsolved
task. 

The analysis of the current literature carried out so far leads to
the following conclusions: When it comes to distributed computation
of NE via best-response dynamics, the following issues arise: i) the
convergence analysis and algorithms apply only to a restricted class
of NEPs, whose players' cost functions and feasible sets have a very
specific structure, leaving outside player-convex NEPs in the general
form (\ref{eq:NEP_intro}); ii) the best-response mapping of each
player must be unique and/or is required to be computed in closed
form; iii) convergence is obtained only under conditions implying
the uniqueness of the NE; and iv) none of current results and VI-based
methodologies can be applied to study and solve \emph{complex} player-convex
NEPs, which arise naturally, e.g., from applications in MIMO communications.

\vspace{-0.2cm}

\subsection{Main contributions\vspace{-0.2cm}
}

In order to address the key issues listed at the end of the previous
subsection, in this paper we introduce several new developments that
are summarized next. 
\begin{enumerate}
\item Building on our recent contributions \cite{Facchinei_Pang_VI-NE_bookCh_09,FacchineiPangScutariLampariello_MP11,Pang-Scutari-Palomar-Facchinei_SP_10,Scutari-Facchinei-Pang-Lampariello_INFOCOM12},
we develop a VI-based unified theory for the study and design of distributed
best-response algorithms for the solution of \emph{real} player-convex
NEPs, having (possibly) multiple solutions. Our unified framework
has many desirable properties, such as: \smallskip{}
\\
$-$ It provides a systematic methodology for analyzing old and new
algorithms, simplifying greatly the application of game-theoretical
models to new problems.\smallskip{}
\\
$-$ It improves on traditional \emph{synchronous} methods studied
in the literature, see e.g. \cite{Facchinei_Pang_VI-NE_bookCh_09,Rosen_econ65,Konnov_VI_book,Yin-Shanbhag-Mehta_TAC10,KANNAN-SHANBHAG_sub10,Shamma-Arslan_TAC05},
by providing for the first time totally \emph{asynchronous} and \emph{distributed}
methods for general player-convex NEPs. In spite of their better features,
the proposed algorithms converge under weaker conditions than those
available in the literature for synchronous best-response schemes;
nevertheless, the convergence conditions still imply the uniqueness
of the NE.\smallskip{}
\\
$-$ It provides convergent best-response schemes also for NEPs having
\emph{multiple} solutions. Although no centralized control is required,
these schemes need some (limited) signaling among the players. Nevertheless,
our algorithms are still applicable to a variety of resource allocation
problems in wireless systems, such as \cite{Yu-Ginis-Cioffi_jsac02,Luo-Pang_IWFA-Eurasip,ShumandLeungSung_JSAC07,Scutari-Palomar-Barbarossa_SP08_PI,Scutari-Palomar-Barbarossa_AIWFA_IT08,Scutari-Palomar-Barbarossa_JSAC08,Scutari-Palomar-Barbarossa_SP08_IWFA,Scutari-Palomar-Barbarossa_SPMag08,Scutari-Palomar-Facchinei-Pang_SPMag09,Scutari-Palomar_SP09_CR_GTMIMO}
and constitute the fist class of provable convergent distributed best-response
schemes for NEPs with multiple solutions. Moreover, an additional
new feature of our methods is that one can also control the quality
of the achievable solution by forcing convergence to a NE that optimizes
a further performance criterion (thus performing an equilibrium selection).
This feature is very appealing in the design of practical wireless
systems, where algorithms with unpredictable performance are not acceptable.\smallskip{}
\\
 $-$ It does not require the players' best-response to be unique
or given in closed form. \smallskip{}
\\
 $-$ It allows us to gauge the trade-off between signaling and characteristic
of the resulting algorithms. 
\item We develop an entirely new theory for the study of VIs in the complex
domain along with new several instrumental technical tools (of independent
interest). Once this new theory has been established, one can (almost)
effortlessly extend all the aforementioned results to player-convex
NEPs whose players' optimization variables are complex matrices. The
resulting algorithms are new to the literature.
\end{enumerate}
\indent 

To the best of our knowledge the above features constitute a substantial
advancement in the distributed solution methods of noncooperative
games, which enlarges considerably scope and flexibility of game-theoretical
models in wireless distributed (MIMO) networks. In order to illustrate
our techniques we consider some new MIMO games over vector Gaussian
Interference Channels (ICs), modeling some distributed resource allocation
problems in SISO and MIMO CR systems. These games are examples of
NEPs that cannot be handled by current methodologies. Numerical results
show the superiority of our approach with respect to plain noncooperative
solutions as well as good performance with respect to centralized
solutions, in spite of very limited signaling among the players.

The paper is organized as follows. Sec. \ref{sub:Some-motivating-examples}
introduces the just mentioned new resource allocations problems. Building
on the connection between VIs and NEPs, Sec. \ref{sec:Nash-Equilibrium-Problems}
focuses on the solution analysis of real convex-player NEPs; special
emphasis is given to some classes of vector functions $\mathbf{F}$
and its properties that play a key role also in the convergence analysis
of distributed algorithms for NEPs. Sec. \ref{sub:Distributed-algorithms-for_NEP}
and Sec. \ref{sec:VI_complex_domain} constitute the core theoretical
part of the paper; in Sec. \ref{sub:Distributed-algorithms-for_NEP}
we provide various distributed algorithms for solving real player-convex
NEPs in several significant settings along with their convergence
properties; Sec. \ref{sec:VI_complex_domain} generalizes the main
results obtained for real convex-player NEPs (VIs) to the complex
case. Sec. \ref{sec:Applications} shows how to apply the developed
machinery to the resource allocation problems introduced in Sec. \ref{sub:Some-motivating-examples},
whereas Sec. \ref{sub:Numerical-Results} provides some numerical
results corroborating our theoretical findings. Finally, Sec. \ref{sec:Conclusions}
draws some conclusions. \vspace{-0.2cm}

\section{Motivating Examples: Noncooperative Games Over Gaussian ICs\label{sub:Some-motivating-examples}}

To motivate and illustrate our new results more in detail, we start
introducing some novel resource allocation problems over SISO frequency-selective
and MIMO Gaussian ICs, widely extending formulations that have already
been studied in the literature. We will show that these problems cannot
be analyzed and solved using current results and algorithms, but call
for a more general theory. 

The IC is suitable to model many practical multiuser systems, such
as digital subscriber lines, wireless ad-hoc and Cognitive Radio (CR)
networks, peer-to-peer systems, multicell OFDM/TDMA cellular systems,
and Femtocell-based networks. We will focus on CR systems; however
the proposed techniques can be readily applied also to the other aforementioned
network models.\vspace{-0.2cm}

\subsection{The SISO case\label{sub:The-SISO-case}}

We consider an $I$-user $N$-parallel Gaussian interference channel,
modeling a CR system composed of $I$ secondary users (SUs) and $P$
primary users (PUs). In this model, there are $I$ transmitter-receiver
pairs$-$the SUs$-$where each transmitter wants to communicate with
its corresponding receiver over a set of $N$ parallel Gaussian subchannels
which may represent time or frequency bins (here we consider transmissions
over the frequency-selective IC without loss of generality). We denote
by $H_{ij}(k)$ the (cross-) channel transfer function over the $k$-th
frequency bin between the secondary transmitter $j$ and the receiver
$i$, while the channel transfer function of secondary link $i$ is
$H_{ii}(k)$. The transmission strategy of each user (pair) $i$ is
the power allocation vector $\mathbf{p}_{i}=\{p_{i}(k)\}_{k=1}^{N}$
over the $N$ subcarriers; the power budget of each transmitter $i$
is $\sum_{k=1}^{N}p_{i}(k)\leq P_{i}$. In a CR system, additional
power constraints limiting the interference radiated by the SUs need
to be imposed. Here we envisage the use of the following general interference
constraints: for each SU $i$,\vspace{-0.2cm} 
\begin{equation}
\sum_{k=1}^{N}\mathbf{w}_{i}(k)\, p_{i}(k)\leq\boldsymbol{\alpha}_{i},\quad i=1,\ldots,I,\vspace{-0.2cm}\label{eq:Interference_Constraint}
\end{equation}
where $\mathbf{w}_{i}(k)\in\mathbb{R}_{+}^{m}$ and $\boldsymbol{\alpha}_{i}\in\mathbb{R}_{+}^{m}$
are nonnegative $m$-length vectors. Note that constraints in the
form of (\ref{eq:Interference_Constraint}) are general enough to
include, as special cases, for example: i) spectral mask constraints
$\mathbf{p}_{i}\leq\mathbf{p}_{i}^{\max}$, where $\mathbf{p}_{i}^{\max}=({p}_{i}^{\max}(k))_{k=1}^{N}$
is the vector of spectral masks over licensed bands; and ii) interference
temperature limit-like constraints ${\sum_{k=1}^{N}}|H_{pi}^{(P,S)}(k)|^{2}p_{i}(k)\leq{\displaystyle I_{pi}}$
for $p=1,\ldots,P,$ where $H_{pi}^{(P,S)}(k)$ is the cross-channel
transfer function over carrier $k$ between the secondary transmitter
$i$ and the primary receiver $p$, and $I_{pi}$ is the maximum level
of interference that SU $i$ is allowed to generate. Let us define
by\vspace{-0.2cm} 
\begin{equation}
\tilde{{\mathcal{P}}}{}_{i}^{\,\texttt{{siso}}}\triangleq\left\{ \mathbf{p}_{i}\in\mathbb{R}^{N}\,:\,{\displaystyle {\sum_{k=1}^{N}}p_{i}(k)\leq P_{i},\quad}\mathbf{0}\leq\mathbf{p}_{i}\leq\mathbf{p}_{i}^{\max}\right\} ,\label{set_P_q}
\end{equation}
the set of power budget constraint of SU $i$ including explicitly
the power budget and spectral mask constraints.

Under basic information theoretical assumptions (see, e.g., \cite{Yu-Ginis-Cioffi_jsac02,Scutari-Palomar-Barbarossa_SP08_PI}),
the maximum achievable rate on link $i$ for a specific power allocation
profile $\mathbf{p}_{1},\ldots,\mathbf{p}_{I}$ is 
\begin{equation}
r_{i}(\mathbf{p}_{i},\mathbf{p}_{-i})=\sum_{k=1}^{N}\log\left(1+\dfrac{|H_{ii}(k)|^{2}p_{i}(k)}{\sigma_{i}^{2}(k)+\sum_{j\neq i}|H_{ij}(k)|^{2}p_{j}(k)}\right)\label{eq:rate_FSIC}
\end{equation}
where $\mathbf{p}_{-i}\triangleq\left(\mathbf{p}_{1},\ldots,\mathbf{p}_{i-1},\mathbf{p}_{i+1},\ldots,\mathbf{p}_{I}\right)$
is the set of all the users power allocation vectors, except the $i$-th
one, and $\sigma_{i}^{2}(k)+\sum_{j\neq i}|H_{ij}(k)|^{2}p_{j}(k)$
is the variance of the noise plus the multiuser interference (MUI)
over subcarrier $k$ measured by the receiver $i$, with $\sigma_{i}^{2}(k)$
denoting the power of the thermal noise (possibly including the interference
generated by the PUs). 

In this setting, the system design is formulated as a NEP: the aim
of each player (link) $i$, given the strategy profile $\mathbf{p}_{-i}$
of the others, is to choose a feasible power allocation $\mathbf{p}_{i}$
that maximizes the rate $r_{i}(\mathbf{p}_{i},\mathbf{p}_{-i})$,
i.e., 
\begin{equation}
\begin{array}{lll}
\underset{\mathbf{p}_{i}}{\textnormal{maximize}} & r_{i}(\mathbf{p}_{i},\mathbf{p}_{-i})\\
\textnormal{subject\,\ to}\\
\quad\begin{array}{l}
\mathbf{(a)}:\vspace{0.3cm}\\
\mathbf{(b)}:\vspace{0.25cm}
\end{array} & \hspace{-0.2cm}\left.\begin{array}{l}
\mathbf{p}_{i}\in\tilde{{\mathcal{P}}}{}_{i}^{\,\texttt{{siso}}},\\
{\displaystyle {\sum_{k=1}^{N}}}\mathbf{w}_{i}(k)\, p_{i}(k)\leq\boldsymbol{\alpha}_{i},
\end{array}\right\} \triangleq\mathcal{P}_{i}^{\,\texttt{{siso}}}
\end{array}\label{eq:max_rate_game}
\end{equation}
for all $i=1,\ldots,I$, where $\tilde{{\mathcal{P}}}{}_{i}^{\,\texttt{{siso}}}$
and $r_{i}(\mathbf{p}_{i},\mathbf{p}_{-i})$ are defined in (\ref{set_P_q})
and (\ref{eq:rate_FSIC}), respectively. We denote the NEP based on
(\ref{eq:max_rate_game}) by $\mathcal{G}_{\texttt{{siso}}}=\left\langle \mathcal{P}{}^{\,\texttt{{siso}}},\,(r_{i})_{i=1}^{I}\right\rangle $
, with $\mathcal{P}{}^{\,\texttt{{siso}}}\triangleq\prod_{i}\mathcal{P}{}_{i}^{\,\texttt{{siso}}}$
and $\mathcal{P}{}_{i}^{\,\texttt{{siso}}}$ being the feasible set
of the optimization problem (\ref{eq:max_rate_game}) of SU $i$.
Note that $\mathcal{G}_{\texttt{{siso}}}$ is an instance of the real
player-convex NEP in (\ref{eq:NEP_intro}). \smallskip{}

\noindent \textbf{Literature review}. Special cases of the NEP in
(\ref{eq:max_rate_game}) have been extensively studied in the literature
in the context of ad-hoc networks, namely when there are\emph{ only
power constraints }(a) \cite{Yu-Ginis-Cioffi_jsac02,Luo-Pang_IWFA-Eurasip,Scutari-Palomar-Barbarossa_SP08_PI,Scutari-Palomar-Barbarossa_AIWFA_IT08,CendrillonHuangChiangMoonen_TSP06}.
In such a simplified setting, given the strategy profile $\mathbf{p}_{-i}$,
the optimization problem of each player reduces to: 
\begin{equation}
\begin{array}{lll}
\underset{\mathbf{p}_{i}}{\textnormal{maximize}} &  & r_{i}(\mathbf{p}_{i},\mathbf{p}_{-i})\\
\textnormal{subject\,\ to} &  & \mathbf{p}_{i}\in\tilde{{\mathcal{P}}}_{i}^{\,\texttt{{siso}}}.
\end{array}\label{eq:game_g_tilde}
\end{equation}
We denote the game resulting from (\ref{eq:game_g_tilde}) by $\tilde{{\mathcal{G}}}_{\texttt{{siso}}}=\left\langle \tilde{{\mathcal{P}}}^{\,\texttt{{siso}}},\,(r_{i})_{i=1}^{I}\right\rangle $,
with $\tilde{{\mathcal{P}}}^{\,\texttt{{siso}}}\triangleq\prod_{i}\tilde{{\mathcal{P}}}_{i}^{\,\texttt{{siso}}}$.
Introducing the matrices $\mathbf{M}\triangleq(\mathbf{M}_{ij})_{i,\, j=1}^{I}\in\mathbb{R}^{N\, I\times N\, I}$
and $\boldsymbol{\Gamma}\in\mathbb{R}^{I\times I}$ defined respectively
as\vspace{-0.2cm}

\begin{equation}
\ensuremath{\mathbf{M}_{ij}\triangleq\text{diag}\left\{ \left(\dfrac{|H_{ij}(k)|^{2}}{|H_{ii}(k)|^{2}}\right)_{k=1}^{N}\right\} }\quad\mbox{and}\quad\left[\boldsymbol{\Gamma}\right]_{ij}\triangleq\left\{ \begin{tabular}{ll}
 \ensuremath{0},  &  if \ensuremath{i=j}; \\
\ensuremath{\max_{k}\dfrac{|H_{ij}(k)|^{2}}{|H_{ii}(k)|^{2}},} &  \ensuremath{\text{otherwise},}
\end{tabular}\right.\label{Upsilon_matrix}
\end{equation}
the state-of-the-art-results on $\tilde{{\mathcal{G}}}_{\texttt{{siso}}}$
can be collected together in the following theorem, where $\rho(\boldsymbol{A})$
denotes the spectral radius of $\mathbf{A}$.\vspace{-0.2cm}

\begin{theorem}\label{Theo_NEP}Given the NEP $\tilde{{\mathcal{G}}}_{\texttt{{siso}}}$
(with no interference constraints), the following hold.\emph{ } \vspace{-0.2cm}
\begin{description}
\item [{{(a)}}] $\tilde{{\mathcal{G}}}_{\texttt{{siso}}}$ has a nonempty
and compact solution set\emph{; } 
\item [{{(b)}}] If $\mathbf{M}\succ\mathbf{0}$ then $\tilde{{\mathcal{G}}}_{\texttt{{siso}}}$
has a unique NE\emph{ \cite{Luo-Pang_IWFA-Eurasip,Scutari-Palomar-Barbarossa_SP08_PI}};
\vspace{-0.1cm}
\item [{{(c)}}] If $\rho(\boldsymbol{\Gamma})<1$, then $\tilde{{\mathcal{G}}}_{\texttt{{siso}}}$
has a unique NE and the asynchronous Iterative Waterfilling Algorithm
(IWFA) based on the waterfilling best-response as proposed in \emph{\cite{Scutari-Palomar-Barbarossa_SP08_PI}}
converges to the equilibrium\emph{.} \vspace{-0.1cm}
\end{description}
\end{theorem}

Theorem \ref{Theo_NEP} provides a satisfactory characterization of
the NEP $\tilde{{\mathcal{G}}}_{\texttt{{siso}}}$ under $\rho(\boldsymbol{\Gamma})<1$
(or $\mathbf{M}$ positive definite). 
However, condition $\rho(\boldsymbol{\Gamma})<1$ may be too restrictive
in practice; indeed there are channel scenarios resulting in games
$\tilde{{\mathcal{G}}}_{\texttt{{siso}}}$ having multiple Nash equilibria,
resulting thus in $\rho(\boldsymbol{\Gamma})>1$. In such cases, the
IWFA is no longer guaranteed to converge and there are no algorithms
available in the literature solving the game $\tilde{{\mathcal{G}}}_{\texttt{{siso}}}$.
Moreover, the results in Theorem \ref{Theo_NEP} as well as the mathematical
tools used in \cite{Yu-Ginis-Cioffi_jsac02,Luo-Pang_IWFA-Eurasip,Scutari-Palomar-Barbarossa_SP08_PI,Scutari-Palomar-Barbarossa_AIWFA_IT08,CendrillonHuangChiangMoonen_TSP06}
to study $\tilde{{\mathcal{G}}}_{\texttt{{siso}}}$ cannot be applied
to the more general $\mathcal{G}_{\texttt{{siso}}}$, even in the
case of unique NE. The theoretical analysis of $\mathcal{G}_{\texttt{{siso}}}$
is then an open problem, which will be addressed in Sec. \ref{sec:Applications},
based on the general framework that we introduce in the forthcoming
sections.\vspace{-0.2cm}

\subsection{The MIMO case\label{sub:The-MIMO-case}}

In a MIMO setting, the secondary transceivers are equipped with multiple
antennas and are allowed to transmit over a multidimensional space,
whose coordinates may represent time slots, frequency bins, or angles.
In this setting, we envisage the use of the following very general
interference constraints:

\vspace{-0.2cm}

\begin{description}
\item [{\emph{-}}] \emph{Null constraints}: 
\[
\mathbf{U}_{i}^{H}\mathbf{Q}{}_{i}=\mathbf{0},
\]
where $\mathbf{Q}_{i}\in\mathbf{\mathbb{C}}^{n_{T_{i}}\times n_{T_{i}}}$
is the transmit covariance matrix of SU $i$ with $n_{T_{i}}$ being
the number of transmit antennas and $\mathbf{U}_{i}\in\mathbf{\mathbb{C}}^{n_{T_{i}}\times r_{U_{i}}}$
is a tall matrix whose columns represent the ``directions'' along
with user $i$ is not allowed to transmit. We assume, without loss
of generality (w.l.o.g.) that each matrix $\mathbf{U}_{i}$ is full-column
rank and, to avoid the trivial solution $\mathbf{Q}_{i}=\mathbf{0}$,
$r_{U_{i}}<n_{T_{i}}$. 
\item [{-}] \emph{Soft and peak power shaping constraints}: 
\[
\text{{tr}}\left(\mathbf{G}_{pi}^{H}\mathbf{Q}{}_{i}\mathbf{G}_{pi}\right)\leq I_{pi}^{\limfunc{ave}}\quad\mbox{and}\quad\lambda_{\max}\left(\mathbf{F}_{pi}^{H}\mathbf{Q}{}_{i}\mathbf{F}_{pi}\right)\leq I_{pi}^{\limfunc{peak}},\quad p=1,2,\ldots,
\]
which represent a relaxed version of the null constraints by limiting
the total average and peak average power radiated along the range
space of matrices $\mathbf{G}_{pi}\in\mathbf{\mathbb{C}}^{n_{T_{i}}\times n_{G_{p}}}$
and $\mathbf{F}_{pi}\in\mathbf{\mathbb{C}}^{n_{T_{i}}\times n_{F_{p}}}$,
where $I_{pi}^{\limfunc{ave}}$ and $I_{pi}^{\limfunc{peak}}$ are
the maximum average and average peak power respectively that can be
transmitted along the directions spanned by $\mathbf{G}_{pi}$ and
$\mathbf{F}_{pi}$. \vspace{-0.2cm}
\end{description}
\noindent \indent Null constraints are enforced to prevent SUs from
transmitting over prescribed subspaces (the range space of $\mathbf{U}_{i}$),
which for example can identify portion of licensed spectrum, time
slots used by the PUs, and/or angular directions identifying the primary
receivers as observed from the secondary transmitters. Soft shaping
constraints can be used instead to control the (average and peak average)
power radiated by the SUs along prescribed time/frequency/angular
``directions'' (those spanned by the columns of matrices $\mathbf{G}_{pi}$
and $\mathbf{F}_{pi}$); for instance, classical power constraints,
such as per-antenna power constraints $[\mathbf{Q}_{i}]_{kk}\leq\beta_{ik}$
with $k=1,,\ldots n_{T_{i}}$, or power budget constraints $\text{{tr}}$$(\mathbf{Q}_{i})\leq P_{i}$
are example of soft-shaping constraints.

\noindent \indent Under basic information theoretical assumptions
(see, e.g., \cite{Scutari-Palomar-Barbarossa_SP08_IWFA}), the maximum
information rate on secondary link $i$ for a given set of user covariance
matrices $\mathbf{Q}_{1},\ldots,\mathbf{Q}_{I}$, is 
\begin{equation}
R_{i}(\mathbf{Q}_{i},\mathbf{Q}_{-i})=\log\det\left(\mathbf{I}+\mathbf{H}_{ii}^{H}\mathbf{R}_{\mathbf{-}i}(\mathbf{Q}_{-i})^{-1}\mathbf{H}_{ii}\mathbf{Q}_{i}\right)\label{R_qq}
\end{equation}
where $\mathbf{R}_{-i}(\mathbf{Q}_{-i})\triangleq\mathbf{R}_{n_{i}}+\sum\limits _{j\neq i}\mathbf{H}_{ij}\mathbf{Q}_{j}\mathbf{H}_{ij}^{H}$
is the covariance matrix of the noise plus MUI, with $\mathbf{R}_{n_{i}}\in\mathbb{C}^{n_{R_{i}}\times n_{R_{i}}}$
denoting the covariance matrix of the thermal Gaussian zero mean noise
(possibly including the interference generated by the PUs), and assumed
to be positive definite; $\mathbf{Q}_{-i}\triangleq\left(\mathbf{Q}_{j}\right)_{j\neq i}$
is the set of all the users covariance matrices, except the $i$-th
one; $\mathbf{H}_{ii}\mathbf{\in}$ $\mathbf{\mathbb{C}}^{n_{R_{i}}\times n_{T_{i}}}$
is the channel matrix between the $i$-th secondary transmitter and
the intended receiver, whereas $\mathbf{H}_{ij}\mathbf{\in}$ $\mathbf{\mathbb{C}}^{n_{R_{i}}\times n_{T_{j}}}$
is the cross-channel matrix between secondary source $j$ and destination
$i$. Within the above setup, the game theoretical formulation is:
for each SU $i=1,\ldots,I$, 
\begin{equation}
\begin{array}{ll}
\underset{\mathbf{Q}_{i}\succeq\mathbf{0}}{\textnormal{maximize}} & R_{i}(\mathbf{Q}_{i},\mathbf{Q}_{-i})\\
\textnormal{subject\,\ to}\\
\begin{array}{c}
\qquad\mathbf{(a)}:\\
\qquad\mathbf{(b)}:\\
\qquad\mathbf{(c)}:\\
\qquad\mathbf{(d)}:
\end{array} & \hspace{-0.2cm}\left.\begin{array}{l}
\text{{tr}}\left(\mathbf{Q}{}_{i}\right)\leq P_{i},\\
\mathbf{U}_{i}^{H}\mathbf{Q}{}_{i}=\mathbf{0},\\
\text{{tr}}\left(\mathbf{G}_{pi}^{H}\mathbf{Q}{}_{i}\mathbf{G}_{pi}\right)\leq I_{pi}^{\limfunc{ave}},\quad\lambda_{\max}\left(\mathbf{F}_{pi}^{H}\mathbf{Q}{}_{i}\mathbf{F}_{pi}\right)\leq I_{pi}^{\limfunc{peak}},\quad\quad p=1,2,\ldots,\\
\mathbf{Q}_{i}\in\mathcal{Q}_{i},
\end{array}\right\} \triangleq\mathcal{P}_{i}^{\text{\,\texttt{{mimo}}}}
\end{array}\label{eq:MIMO_Game}
\end{equation}
where $\mathcal{Q}_{i}\subseteq\mathbb{C}^{n_{T_{i}}\times n_{T_{i}}}$
is an abstract set that can accommodate (possibly) additional constraints
on the covariance matrix $\mathbf{Q}_{i}$, on top of the power and
interference constraints; we only make the (blanket) assumption that
each $\mathcal{Q}_{i}$ is closed and convex. We refer to the NEP
based on (\ref{eq:MIMO_Game}) as $\mathcal{G}_{\texttt{{mimo}}}=\left\langle \mathcal{P}^{\texttt{{mimo}}},\,(R_{i})_{i=1}^{I}\right\rangle $,
with $ $${\mathcal{P}}^{\texttt{{mimo}}}\triangleq\prod_{i}{\mathcal{P}}_{i}^{\texttt{{mimo}}}$
and $\mathcal{P}_{i}^{\texttt{{mimo}}}$ defined in (\ref{eq:MIMO_Game}).
Note that $\mathcal{G}_{\texttt{{mimo}}}$ is an instance of the complex
NEP (\ref{eq:NEP_intro}).

\noindent \textbf{Literature review}. The design of MIMO CR systems
under different interference-power/interference-temperature constraints
has been addressed in a number of papers. Distributed algorithms (mostly)
for ad-hoc networks based on game theoretical formulations have been
proposed in \cite{Arslan-Fatih-Demirkol-Song_WC07,Scutari-Palomar-Barbarossa_JSAC08,Larsson-Jorswieck_SP08,Scutari-Palomar-Barbarossa_SP08_IWFA,Scutari-Palomar_SP09_CR_GTMIMO};
the state-of-the-art result is the asynchronous MIMO IWFA solving
the NEP in (\ref{eq:MIMO_Game}), in the presence of constraints (a)
\cite{Scutari-Palomar-Barbarossa_SP08_IWFA} and (b) \cite{Scutari-Palomar_SP09_CR_GTMIMO}
only. Results in these papers are strongly based on the specific structure
of the optimization problem and the resulting solution$-$the MIMO
waterfilling-like expression$-$and thus are not applicable to the
general NEP (\ref{eq:MIMO_Game}). $\mathcal{G}_{\texttt{{mimo}}}$
is thus an other example of a novel game whose solution analysis requires
new mathematical tools, which is the goal of this paper. The study
of $\mathcal{G}_{\texttt{{mimo}}}$ is addressed in Sec. \ref{sub:The-MIMO-case_revised}
and will result as a direct application of the framework developed
in the forthcoming sections for complex NEPs.

\section{Nash Equilibrium Problems\label{sec:Nash-Equilibrium-Problems}}

In a standard real NEP there are $I$ players each controlling a variable
$\mathbf{x}_{i}\in\mathbb{R}^{n_{i}}$ that must belong to the player's
feasible set $\mathcal{Q}_{i}$, which is assumed to be closed and
convex: $\mathbf{x}_{i}\in\mathcal{Q}_{i}$. In what follows we denote
by $\mathbf{x}\triangleq(\mathbf{x}_{1},\ldots,\mathbf{x}_{I})$,
the vector of all players' variables, while $\mathbf{x}_{-i}\triangleq\left(\mathbf{x}_{1},\ldots,\mathbf{x}_{i-1},\mathbf{x}_{i+1},\ldots,\right.$
$\left.\mathbf{x}_{I}\right)$ denote the vector of all players' strategies
variables except that of player $i$. The aim of player $i$, given
the other players' strategies $\mathbf{x}_{-i}$, is to choose an
$\mathbf{x}_{i}\in\mathcal{Q}_{i}$ that minimizes his cost function
$f_{i}(\mathbf{x}_{i},\mathbf{x}_{-i})$, i.e., 
\begin{equation}
\begin{array}{ll}
\limfunc{minimize}\limits _{\mathbf{x}_{i}} & f_{i}(\mathbf{x}_{i},\,\mathbf{x}_{-i})\\[5pt]
\text{subject to} & \mathbf{x}_{i}\in{\cal Q}_{i}.
\end{array}\label{eq:G_def}
\end{equation}
Note that the players' optimization problem are \emph{coupled} since
the players' objective function (may) depend on the other players'
choices. Define the joint strategy set of the NEP by $\mathcal{Q}=\prod_{i=1}^{I}\mathcal{Q}_{i}$,
whereas $\mathcal{Q}_{-i}\triangleq\prod_{j\neq i}\mathcal{Q}_{j}$,
and set $\mathbf{f\triangleq}(f_{i})_{i=1}^{I}$. The NEP is formally
defined by the tuple $\mathcal{G}=\left\langle \mathcal{Q},\mathbf{f}\right\rangle $.
A solution of the NEP is the well-known Nash Equilibrium (NE), which
is formally defined in (\ref{eq:NE}).

We recall that a solution of (\ref{eq:G_def}), given $\mathbf{x}_{-i}$,
is also called best-response of user $i$. A useful way to see a NE
is as a fixed-point of the best-response mapping for each player;
this suggests the use of (iterative) best-response-based algorithms
to solve the game. Given the limitations of classical fixed-point
results in the study of convergence of best-response based algorithms
(cf. Sec. \ref{sec:Introduction-and-Motivation}), we address this
issue by reducing the NEP to a VI problem. The main advantage of this
reformulation is algorithmic, since once it has been carried out,
we can build on the well-developed VI theory \cite{Facchinei-Pang_FVI03}
in order to design new solution methods for NEPs. In the rest of this
paper, we freely use some basic results from VI theory. Since this
theory is not widely known in the information theory, communications,
and signal processing communities, for the reader convenience we summarize
the VI results used in this paper in Appendix  \ref{sec:A-Theory-of_partitioned_VI}.\vspace{-0.2cm}

\subsection{Connection to variational inequalities}

At the basis of the VI approach to NEPs there is an easy equivalence
between a real NEP and a suitably defined partitioned VI. This equivalence
follows readily from the minimum principle for convex problems and
the Cartesian structure of the joint strategy set $\mathcal{Q}$ \cite[Prop. 1.4.2]{Facchinei-Pang_FVI03}.\vspace{-0.2cm}

\begin{proposition}\label{VI_reformulation} Given the real NEP $\ensuremath{\mathcal{G}=\left\langle \mathcal{Q},\mathbf{f}\right\rangle }$,
suppose that for each player $\ensuremath{i}$ the following hold:\vspace{-0.3cm}

\begin{description}
\item [{{i)}}] the (nonempty) strategy set $\ensuremath{\mathcal{Q}_{i}}$
is closed and convex;\vspace{-0.3cm}

\item [{{ii)}}] the payoff function $\ensuremath{f_{i}(\mathbf{x}_{i},\mathbf{x}_{-i})}$
is convex and continuously differentiable in $\ensuremath{\mathbf{x}_{i}}$
for every fixed $\mathbf{x}_{-i}$. \vspace{-0.3cm}

\end{description}
Then, the game $\ensuremath{\mathcal{G}}$ is equivalent to the $\ensuremath{\limfunc{VI}(\mathcal{Q},\mathbf{F})}$,
where $\ensuremath{\mathbf{F}(\mathbf{x})\triangleq\left(\nabla_{\mathbf{x}_{i}}f_{i}(\mathbf{x})\right)_{i=1}^{I}}$.\end{proposition}\vspace{-0.2cm}

\noindent \indent In the sequel we refer to the VI$(\mathcal{Q},\mathbf{F})$
defined in previous proposition as the VI associated to the NEP $\ensuremath{\mathcal{G}}$.
It is possible to relax the assumptions in Proposition \ref{VI_reformulation}
and still get useful connections between games and VIs \cite{Facchinei_Pang_VI-NE_bookCh_09};
but since our aims are mainly computational, we do not pursue this
topic further. 
Indeed, throughout the paper, we will make the following blanket convexity/smoothness
assumptions, unless stated otherwise. \medskip{}

\noindent \textbf{Assumption 1.} For each $i=1,\ldots,I$, the set
$\mathcal{Q}_{i}$ is a nonempty, closed, and convex subset of $\mathbb{R}^{n_{i}}$
and the function $\ensuremath{f_{i}(\mathbf{x}_{i},\mathbf{x}_{-i})}$
is continuously differentiable on $\mathcal{Q}=\prod_{i}\mathcal{Q}_{i}$
and convex in $\ensuremath{\mathbf{x}_{i}}$ for every fixed $\mathbf{x}_{-i}\in\mathcal{Q}_{-i}$.

\medskip{}

\noindent \textbf{Assumption 2. }For each $i=1,\ldots,I$, each function
$\ensuremath{f_{i}(\mathbf{x})}$ is twice continuously differentiable
with bounded derivatives on $\mathcal{Q}=\prod_{i}\mathcal{Q}_{i}$.

\subsection{Existence and uniqueness of a NE\label{sub:Existence-and-uniqueness_NE}}

Building on the VI reformulation in the previous section and the existence/uniqueness
results for VIs (see Theorem \ref{Theo_existence_uniqueness} in Appendix
\ref{sec:A-Theory-of_partitioned_VI}), we can easily state the following
theorem that needs no further proof.

\begin{theorem}\label{Theo_existence_uniquenessNE} Given the real
NEP $\mathcal{G}=\left\langle \mathcal{Q},\mathbf{f}\right\rangle $,
suppose that $\mathcal{G}$ satisfies Assumption 1 and let $\ensuremath{\mathbf{F}(\mathbf{x})\triangleq\left(\nabla_{\mathbf{x}_{i}}f_{i}(\mathbf{x})\right)_{i=1}^{I}}$.
Then, the following statements hold:

\vspace{-0.3cm}

\begin{description}
\item [{{(a)}}] Suppose that for every $i$ the strategy set $\ensuremath{\mathcal{Q}_{i}}$
is bounded. Then the NEP has a nonempty and compact solution set;\vspace{-0.3cm}

\item [{{(b)}}] Suppose that $\ensuremath{\mathbf{F}}$ is a monotone
function on $\mathcal{Q}$. Then the NEP has a convex (possibly empty)
solution set;\vspace{-0.3cm}
\item [{{(c)}}] Suppose that $\ensuremath{\mathbf{F}}$ is a P (or strictly
monotone) function on $\mathcal{Q}$. Then the NEP has at most one
solution;\vspace{-0.3cm}

\item [{{(d)}}] Suppose that $\ensuremath{\mathbf{F}}$ is a uniformly-P
(or strongly monotone) function on $\mathcal{Q}$. Then the NEP has
a unique solution. 
\end{description}
\end{theorem}

\noindent \indent The above theorem and many of the algorithmic developments
to follow hinge critically on the monotonicity or P properties of
the function $\mathbf{F}$. However, checking such properties by using
directly the definition (see Def. \ref{Def_monotonicity} in Appendix
\ref{sec:A-Theory-of_partitioned_VI}) is in general not possible.
It is then useful to derive more practical conditions to establish
whether the aforementioned properties hold. It is well known that
when $\mathcal{Q}$ is an \emph{open} set and $\mathbf{F}$ is continuously
differentiable on $\mathcal{Q}$, with Jacobian matrix denoted by
$\mbox{\textbf{J}}\mathbf{F}$, it holds that \cite[Prop. 2.3.2]{Facchinei-Pang_FVI03}:%
\footnote{Conditions in (\ref{monotonicity-convexity-connection}) can be generalized
also to the case in which $\mathcal{Q}$ is closed; this will be done
in Sec. \ref{sec:VI_complex_domain}, where we introduce the VI problem
in the complex domain; see Proposition \ref{VI_monotonicity_closed_sets}.%
}\vspace{-0.3cm}

\begin{equation}
\begin{array}{lll}
\mathbf{F}(\mathbf{x})\mbox{ is monotone on \ensuremath{\mathcal{Q}}} & \quad\Leftrightarrow\quad\  & \mathbf{JF(x)}\succeq\mathbf{0},\,\,\forall\mathbf{x}\in\mathcal{Q};\\
\mathbf{F}(\mathbf{x})\mbox{ is strictly monotone on \ensuremath{\mathcal{Q}}} & \quad\Leftarrow\quad\  & \mathbf{JF(x)}\succ\mathbf{0},\,\,\forall\mathbf{x}\in\mathcal{Q};\\
\mathbf{F}(\mathbf{x})\mbox{ is strongly monotone on \ensuremath{\mathcal{Q}}} & \quad\Leftrightarrow\quad\  & \mathbf{JF}-c_{\text{{sm}}}\,\mathbf{I}\succeq\mathbf{0},\,\,\forall\mathbf{x}\in\mathcal{Q};
\end{array}\label{monotonicity-convexity-connection}
\end{equation}
where $\mathbf{A}\succeq\mathbf{0}$ ($\mathbf{A}\succ\mathbf{0}$)
means that $\mathbf{A}$ is a positive semidefinite (definite) matrix.
The verification of these kind of conditions is often difficult and,
furthermore, in many practical instances their verification cannot
easily be linked to physical characteristics of the systems being
studied. Therefore, our aim in the remaining part of this subsection
is developing some (conceptually) simpler and new conditions that
permit to deduce the desired $\mathbf{F}$ properties and that, at
least in some instances, can give some further insight into the problem
at hand. The conditions we introduce here capture some kind of ``diagonal
dominance'' property of $\mathbf{JF}$, and will play a key role
in the convergence theory of the algorithms introduced in Sec. \ref{sub:Distributed-algorithms-for_NEP}.

Let us define the matrix $\mathbf{JF}_{\text{{low}}}$ having the
same dimension as $\mathbf{JF}(\mathbf{x})$: 
\begin{equation}
\left[\mathbf{JF}_{\text{{low}}}\right]_{rs}\,\triangleq\,\left\{ \begin{array}{lcl}
\underset{\mathbf{x}\in\mathcal{Q}}{\inf}\left[\mathbf{B}^{T}\,\mathbf{JF}(\mathbf{x})\,\mathbf{B}\right]_{rr}, &  & \mbox{if }r=s,\\
-\,\underset{\mathbf{x}\in\mathcal{Q}}{\sup}\left|\left[\mathbf{B}^{T}\,\mathbf{JF}(\mathbf{x})\,\mathbf{B}\right]_{rs}\right|, &  & \mbox{otherwise,}
\end{array}\right.\label{eq:def_lower_of_comparison_of_Jacobian}
\end{equation}
where $\mathbf{B}\in\mathbb{R}^{n\times n}$ is an arbitrary nonsingular
matrix. A case that is relevant in the analysis of NEPs is that of
partitioned VIs. This corresponds to the set ${\cal Q}$ being a Cartesian
product of lower-dimensional sets: $\mathcal{Q}\triangleq\prod_{i=1}^{I}\mathcal{Q}_{i}$,
with each $\mathcal{Q}_{i}\subseteq\mathbb{R}^{n_{i}}$ being nonempty,
closed, and convex and with $n\triangleq\sum_{i=1}^{I}n_{i}$. When
this structure arises it will be quite natural to partition both $\mathbf{F}$
and $\mathbf{x}$ accordingly and therefore write $\mathbf{F}(\mathbf{x})=(\mathbf{F}_{i}(\mathbf{x}))_{i=1}^{I}$
and $\mathbf{x}=(\mathbf{x})_{i=1}^{I}$, where $\mathbf{F}_{i}:\mathcal{Q}\to\mathbb{R}^{n_{i}}$
is the $i$th-component block function of $\mathbf{F}$ and $\mathbf{x}_{i}\in\mathbb{R}^{n_{i}}$
is the $i$th-component block of $\mathbf{x}$. In the case of partitioned
VIs, let us  introduce the ``condensed'' $I\times I$ real matrices
$\mathbf{\boldsymbol{\Upsilon}}_{\mathbf{F}}$ and $\boldsymbol{\Gamma}_{\mathbf{F}}$:

\begin{equation}
\left[\mathbf{\boldsymbol{\Upsilon}}_{\mathbf{F}}\right]_{ij}\,\triangleq\,\left\{ \begin{array}{lcl}
\alpha_{i}^{\min}, &  & \mbox{if }i=j,\\
-\beta_{ij}^{\max}, &  & \mbox{otherwise},
\end{array}\right.\quad\mbox{and}\quad\left[\mathbf{\boldsymbol{\Gamma}_{F}}\right]_{ij}\,\triangleq\,\left\{ \begin{array}{lcl}
0, &  & \mbox{if }i=j,\\
{\beta_{ij}^{\max}}/\alpha_{i}^{\min}, &  & \mbox{otherwise,}
\end{array}\right.\label{eq:Upsilon_matrix}
\end{equation}
with
\begin{equation}
\alpha_{i}^{\min}\,\triangleq\,\inf_{\mathbf{x}\in\mathcal{Q}}\,\lambda_{\text{{least}}}\left(\mathbf{C}_{i}^{T}\,\mathbf{J}_{i}\mathbf{F}_{i}(\mathbf{x})\,\mathbf{C}_{i}\right)\,\quad\mbox{and}\quad\beta_{ij}^{\max}\,\triangleq\,\sup_{\mathbf{x}\in\mathcal{Q}}\,\left\Vert \mathbf{C}_{i}^{T}\,\mathbf{J}_{j}\mathbf{F}_{i}(\mathbf{x})\,\mathbf{C}_{j}\right\Vert ,\label{eq:def_alpha_and_beta_Jac}
\end{equation}
where $\lambda_{\text{{least}}}\left(\mathbf{A}\right)$ denotes the
smallest eigenvalue of $\frac{{1}}{2}(\mathbf{A}+\mathbf{A}^{T})$
(the symmetric part of $\mathbf{A}$), $\mathbf{J}_{j}\mathbf{F}_{i}(\mathbf{x})$
is the Jacobian of $\mathbf{F}_{i}(\mathbf{x})$ with respect to $\mathbf{x}_{j}$,
and $\mathbf{C}_{i}\in\mathbb{R}^{n_{i}\times n_{i}}$ with $i=1,\ldots,I$,
is a set of arbitrary nonsingular matrices. Note that in the definition
of $\mathbf{\boldsymbol{\Gamma}_{F}}$ we tacitly assumed all $\alpha_{i}^{\min}\neq0$
and $\beta_{ij}^{\max}$ are finite; the latter condition is equivalent
to the boundedness of $\mathbf{J}_{j}\mathbf{F}_{i}(\mathbf{x})$
on $\mathcal{Q}$. Matrices $\mathbf{B}$ and $\mathbf{C}_{i}$'s
provide an additional degree of freedom in obtaining conditions for
monotonicity and P properties of $\mathbf{F}$ that can be linked
to physical characteristics of the systems being studied (see Sec.
\ref{sec:Applications} for some examples). In order to explore the
relationship between the two matrices $\mathbf{\boldsymbol{\Upsilon}}_{\mathbf{F}}$
and $\boldsymbol{\Gamma}_{F}$, we need the following definition (see,
e.g., \cite{Cottle-Pang-Stone_bookLCP92,Berman-Plemmons_bookNonNegMat87}).

\begin{definition}\label{Def_Z_P_K_matrix} A matrix $\mathbf{M}\in\mathbb{R}^{n\times n}$
is called P matrix if every principal minor of $\mathbf{M}$ is positive.\end{definition}\vspace{-0.2cm}

Any positive definite matrix is obviously a P-matrix, but the reverse
does not hold (unless the matrix is symmetric). Furthermore, building
on the properties of the P-matrices \cite[Lemma 13.14]{Cottle-Pang-Stone_bookLCP92}\emph{,}
one can show that $\boldsymbol{\Upsilon}_{\mathbf{F}}$ is a P-matrix
if and only if $\rho(\boldsymbol{\Gamma}_{\mathbf{F}})<1$\emph{,
}where $\rho(\mathbf{A})$ denotes the spectral radius of $\mathbf{A}$
(see, e.g., \cite{Scutari-Palomar-Barbarossa_AIWFA_IT08}).

Matrices \textbf{$\mathbf{J}\textbf{F}_{\text{{low}}}$} and $\mathbf{\boldsymbol{\Upsilon}}_{\mathbf{F}}$
are useful to obtain sufficient conditions for the monotonicity and
P property of the mapping $\mathbf{F}$, as given next.\vspace{-0.2cm}

\begin{proposition} \label{monotonicity} Let $\mathbf{F}:\mathcal{Q}\mathcal{\rightarrow\mathbb{R}}^{n}$
be continuously differentiable with bounded derivatives on the closed
and convex set $\mathcal{Q}$. The following statements hold: \vspace{-0.1cm}
\begin{description}
\item [{{{\rm}(a)}}] If \textbf{\emph{$\mathbf{J}\mathbf{F}_{\text{low}}$}}
is \emph{copositive},%
\footnote{A matrix $\mathbf{A}$ is copositive if $\mathbf{x}^{T}\mathbf{A}\mathbf{x}\geq0$
for all $\mathbf{x}\geq\mathbf{0}$; it is strictly copositive if
$\mathbf{x}^{T}\mathbf{A}\mathbf{x}>0$ for all $\mathbf{0}\neq\mathbf{x}\geq\mathbf{0}$.
A positive (semi)definite matrix is (strictly) copositive.%
} then $\mathbf{F}$ is \emph{monotone} on $\mathcal{Q}$; \vspace{-0.1cm}
\item [{{{\rm}(b)}}] If \textbf{\emph{$\mathbf{J}\mathbf{F}_{\text{{low}}}$}}\emph{
}is \emph{strictly} \emph{copositive,$^{2}$ }then\emph{ }$\mathbf{F}$
is \emph{strictly monotone} on $\mathcal{Q}$; \vspace{-0.1cm}
\item [{{{\rm}(c)}}] If \textbf{\emph{$\mathbf{J}\mathbf{F}_{\text{{low}}}$}}
is \emph{positive definite, }then $\mathbf{F}$ is \emph{strongly
monotone} on $\mathcal{Q}$ with strong monotonicity constant given
by \emph{$c_{\limfunc{sm}}=\lambda_{\text{{least}}}\left(\mathcal{\mathbf{J}\textbf{F}_{\text{{low}}}}\right)$}
\emph{{[}}or\emph{ $c_{\limfunc{sm}}=\lambda_{\mathcal{\text{{least}}}}\left(\boldsymbol{{\Upsilon}_{F}}\right)${]}}.
\vspace{-0.1cm}
\end{description}
If we assume a Cartesian product structure, i.e. $\mathbf{F}=(\mathbf{F}_{i}(\mathbf{x}))_{i=1}^{I}$
and $\mathcal{Q}=\prod_{i}\mathcal{Q}_{i}$, then: \vspace{-0.1cm}
\begin{description}
\item [{{(d)}}] If $\mathbf{\boldsymbol{\Upsilon}}_{\mathbf{F}}$ is
positive semidefinite/P$_{0}$-matrix, then $\mathbf{F}$ is a monotone/P$_{0}$
function on $\mathcal{Q}$; \vspace{-0.1cm}
\item [{{(e)}}] If $\mathbf{\boldsymbol{\Upsilon}}_{\mathbf{F}}$ is
a P-matrix \emph{{[}}which is equivalent to $\rho(\boldsymbol{\Gamma}_{\mathbf{F}})<1$\emph{{]}},
then $\mathbf{F}$ is a \emph{uniformly P-function} on $\mathcal{Q}$
with uniform P constant given by\vspace{-0.2cm} 
\begin{equation}
\hat{c}_{\limfunc{uP}}(\mathbf{F})\,=\,{\displaystyle \frac{\delta(\boldsymbol{\Upsilon}_{\mathbf{F}})}{I\cdot\left(1+\zeta(\boldsymbol{\Upsilon}_{\mathbf{F}})/\delta(\boldsymbol{\Upsilon}_{\mathbf{F}})\right)^{2(I-1)}\cdot\max_{i=1,\ldots,I}\lambda_{\max}(\mathbf{C}_{i}^{T}\mathbf{C}_{i})},}\label{eq:mod_F_unifP}
\end{equation}
where $\zeta(\boldsymbol{\Upsilon}_{\mathbf{F}})\triangleq\max_{r\neq q}|[\boldsymbol{\Upsilon}_{\mathbf{F}}]_{rq}|$,
and $\delta(\boldsymbol{\Upsilon}_{\mathbf{F}})\triangleq\min\{\sigma([\boldsymbol{\Upsilon}_{\mathbf{F}}]_{\alpha\,\alpha})\,:\,\alpha\subseteq\{1,\ldots,I\}\}$,
with $\sigma([\mathbf{M}]_{\alpha\,\alpha})$ denoting the smallest
of the real eigenvalues (if any exists) of the principal submatrix
of $\mathbf{M}$ of order $\alpha$. 
\end{description}
\end{proposition}\vspace{-0.2cm}\begin{proof} See Appendix \ref{sec:Proof-of-Proposition_monotonicity}.
\end{proof}

\begin{remark}[On the uniqueness conditions]\label{Remark_uniqueness}
\emph{Under the assumption that }$\ensuremath{\mathbf{F}(\mathbf{x})=\left(\nabla_{\mathbf{x}_{i}}f_{i}(\mathbf{x})\right)_{i=1}^{I}}$\emph{
is continuously differentiable with bounded derivatives on $\mathcal{Q}$
(Assumption 2), a sufficient condition for the uniqueness of the NE
is that the matrix }$\boldsymbol{\Upsilon}_{\mathbf{F}}$ \emph{defined
in (\ref{eq:Upsilon_matrix}) be a} \emph{P matrix {[}cf. Theorem
\ref{Theo_existence_uniquenessNE}(d) and Proposition \ref{monotonicity}(e){]}.
It turns out that this condition is sufficient also for global convergence
of best-response asynchronous distributed algorithms described in
Sec. \ref{sub:Distributed-algorithms-for_NEP}.} \emph{Note that if
}$\boldsymbol{\Upsilon}_{\mathbf{F}}$\emph{ is a P matrix, it must
be $\alpha_{i}^{\min}=\inf_{\mathbf{z}\in\mathcal{Q}}\left[\lambda_{\text{{min}}}(\nabla_{\mathbf{x}_{i}}^{2}f_{i}(\mathbf{z}))\right]>0$
for all $i$, where $\lambda_{\text{{min}}}(\nabla_{\mathbf{x}_{i}}^{2}f_{i}(\mathbf{z}))$
denotes the minimum eigenvalue of $\nabla_{\mathbf{x}_{i}}^{2}f_{i}(\mathbf{z})$.}%
\footnote{Note the difference between $\lambda_{\min}$ and $\lambda_{\text{{least}}}$;
the former is used for symmetric matrices, whereas the latter refers
to possibly non symmetric matrices. Of course if $\mathbf{A}$ is
symmetric, then $\lambda_{\min}(\mathbf{A})=\lambda_{\text{{least}}}(\mathbf{A})$.%
}\emph{ Thus an implicit consequence of the P assumption on the matrix
}$\boldsymbol{\Upsilon}_{\mathbf{F}}$\emph{ is the uniform positive
definiteness of the matrices $\nabla_{\mathbf{x}_{1}}^{2}f_{i}$ on
$\mathcal{Q}$, which implies the uniformly strong convexity of $f_{i}(\cdot,\,\mathbf{x}_{-i})$
for any given $\mathbf{x}_{-i}\in\mathcal{Q}_{-i}$ and thus the uniqueness
of the solution of the $i$-th player's optimization problem, for
any given $\mathbf{x}_{-i}\in\mathcal{Q}_{-i}$. }

\emph{The $\beta$'s in the definition of the matrix $\boldsymbol{\Upsilon}_{\mathbf{F}}$
measure the coupling of the players' optimization problems: the larger
the $\beta$'s, the more coupled the players' subproblems are. Indeed,
if all the $\beta$'s were 0, the game $\mathcal{G=<}\mathcal{Q},\,\mathbf{f}\mathcal{>}$
would decompose into $I$ uncoupled optimization problems; in such
a case, requiring the matrix $\boldsymbol{\Upsilon}_{\mathbf{F}}$
to be P simply amounts to requiring all $\alpha$'s to be positive,
which obviously implies uniqueness of the solution. It is reasonable
that if the $\beta$'s increase from zero but remain small enough
with respect to the $\alpha$'s, the game will still have a unique
solution. The P property quantifies how large the $\beta$'s can grow
while still preserving the uniqueness of the solution.}\hfill{}$\square$\end{remark}\vspace{-0.3cm}

We conclude this subsection providing a sufficient condition for the
matrix $\boldsymbol{\Upsilon}_{\mathbf{F}}$ in\emph{ }(\ref{eq:Upsilon_matrix})
to be a P (positive definite) matrix, which can be derived by elementary
diagonal dominance arguments.\vspace{-0.2cm}

\begin{proposition} \label{Cor_SF_for_Pmat} The matrix $\boldsymbol{\Upsilon}_{\mathbf{F}}$
in \emph{(\ref{eq:Upsilon_matrix})} is a P-matrix if \emph{one} of
the following two sets of conditions are satisfied: for some $\mathbf{w}=(w_{i})_{i=1}^{I}>\mathbf{0}$,
\emph{ 
\begin{equation}
\frac{1}{w_{i}}\dsum\limits _{j\neq i}w_{j}\dfrac{{\beta_{ij}^{\max}}}{\alpha_{i}^{\min}}<1,\,\,\forall i=1,\cdots,I,\qquad\frac{1}{w_{j}}\dsum\limits _{i\neq j}w_{i}\dfrac{{\beta_{ij}^{\max}}}{\alpha_{j}^{\min}}<1,\,\,\forall j=1,\cdots,I.\label{SF_aVI}
\end{equation}
} If actually both conditions in\emph{ \eqref{SF_aVI} }are satisfied,
then $\boldsymbol{\Upsilon}_{\mathbf{F}}$ is positive definite. \end{proposition}
The sufficient conditions in Proposition \ref{Cor_SF_for_Pmat} will
be shown in Sec. \ref{sec:Applications} to have an interesting physical
interpretation in the context of power control problems in CR systems.\vspace{-0.2cm}

\subsection{Problem classes\label{sub:Problem-Classes}}

Based on the previous results, it is natural to introduce the following
classes of real NEPs.\vspace{-0.2cm}

\begin{definition}\label{Def_monotone_NEP} A real NEP $\mathcal{G}=\left\langle \mathcal{Q},\mathbf{f}\right\rangle $
is: \vspace{-0.2cm}
\begin{description}
\item [{{i)}}] a\emph{ monotone NEP} if Assumption 1 holds and the associated
VI$(\mathcal{Q},\mathbf{F})$ is monotone; \vspace{-0.2cm}
\item [{{ii)}}] a \emph{uniformly P NEP} if Assumption 1 holds and the
associated VI$(\mathcal{Q},\mathbf{F})$ is uniformly P; \vspace{-0.2cm}
\item [{{iii)}}] a\emph{ $P{}_{\boldsymbol{\Upsilon}}$ NEP} if Assumptions
1 and 2 hold and the matrix $\boldsymbol{\Upsilon}_{\mathbf{F}}$
of the associated VI$(\mathcal{Q},\mathbf{F})$ is P. 
\end{description}
\end{definition}\vspace{-0.2cm}

\noindent 
\begin{figure}[H]
\vspace{-0.4cm}\center \includegraphics[height=3cm]{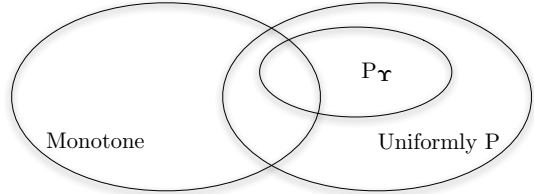}\vspace{-0.5cm}
 \caption{{\small{Relation among NEP classes.\vspace{-0.1cm}}}}

{\small{\label{fig_NE_class}}} 
\end{figure}
 Figure \ref{fig_NE_class} summarizes the relations between these
classes of problems. Note that a monotone NEP is not necessarily uniformly
P; it is enough to observe that monotone NEPs may have multiple NE
(see Example \#1 in Sec. \ref{sub:Numerical-Results}), whereas uniformly
P NEPs have only one solution {[}cf. Theorem \ref{Theo_existence_uniquenessNE}(d){]}\emph{.}
Similarly,\emph{ $P{}_{\boldsymbol{\Upsilon}}$} NEPs (and thus uniformly
P NEPs) are not a subclass of monotone NEPs, as shown by the following
example. \vspace{-0.3cm}

\begin{example}[A P$_{\boldsymbol{\Upsilon}}$ NEP which is not monotone]
\emph{Consider a real NEP with two players, each controlling one scalar
variable: $x_{1}$ and $x_{2}$. The players' problems are}\vspace{-0.1cm}
\[
\begin{array}{rl}
{\displaystyle {\hbox{\rm minimize}_{x_{1}}}} & \dfrac{1}{2}x_{1}^{2}+4x_{1}x_{2}\\[0.5em]
{\rm subject\, to\,\,\,} & x_{1}\,\in\,[0,10]
\end{array}\qquad\qquad\begin{array}{rl}
{\displaystyle {\hbox{\rm minimize}_{x_{2}}}} & \dfrac{1}{2}x_{2}^{2}-\dfrac{1}{8}x_{1}x_{2}\\[0.5em]
\hbox{\rm subject\, to\,} & x_{2}\,\in\,[-2,2]
\end{array}\vspace{-.1cm}
\]
\emph{The VI associated to this NEP is VI$([0,10]\times[-2,2],\mathbf{F})$,
with $\mathbf{F}=[x_{1}+4x_{2},x_{2}-(1/8)x_{1}]^{T}$. The symmetric
part of $\mathbf{J}\mathbf{F}$, $\mathbf{J}\mathbf{F}_{s}$, and
the matrix $\boldsymbol{\Upsilon}_{\mathbf{F}}$ are given by:}\vspace{-0.1cm}
\[
\mathbf{J}\mathbf{F}_{s}\,=\,\left[\begin{array}{cc}
1 & 31/16\\
31/16 & 1
\end{array}\right],\qquad\qquad\boldsymbol{\Upsilon}_{\mathbf{F}}\,=\,\left[\begin{array}{cc}
1 & -4\\
-1/8 & 1
\end{array}\right].
\]
\emph{Since $\mathbf{J}\mathbf{F}_{s}$ has a negative determinant,
$\mathbf{F}$ cannot be monotone; on the other hand it is easy to
check that the two principal minors of $\boldsymbol{\Upsilon}_{\mathbf{F}}$
are positive, implying that $\mathcal{G}$ is a P NEP.} \end{example}\vspace{-0.3cm}

\noindent \indent Centralized algorithms for monotone and uniformly
\emph{$P{}_{\boldsymbol{\Upsilon}}$} NEPs, based on VI theory, are
well-known \cite[vol II]{Facchinei-Pang_FVI03}; in this paper, we
focus on the more challenging issue of devising \emph{distributed}
(and possibly asynchronous) solution schemes for NEPs, which is the
topic of the next section.\vspace{-0.2cm}

\section{Distributed Algorithms for NEPs\label{sub:Distributed-algorithms-for_NEP}}

This section along with the next one constitute the core theoretical
part of the paper. We develop here a novel theory that allows devising
{\em distributed} algorithms for computing Nash equilibria in several
significant settings. More specifically, we will provide novel distributed
(asynchronous) algorithms for the solution of: (a) \emph{$P{}_{\boldsymbol{\Upsilon}}$}
NEPs; and (b) monotone NEPs. 

Since monotone NEPs may have multiple solutions, in case (a) we will
further consider both the situations in which one is interested in
computing \emph{any} \emph{one} solution, and the situations in which
one wants to select the \emph{best }solution, according to a given
criterion. In each of the settings above we will provide best-response-based
distributed algorithms along with their convergence properties; the
proposed algorithms differ in: i) the computational effort; ii) the
players' synchronization/signaling requirements; and iii) the convergence
speed. Note that while centralized solution methods are known for
uniformly P NEPs, the development of distributed algorithms for this
class of games is at the time of this writing an open problem. 

This section is organized in three parts. Sec. \ref{sub:Best-response-decomposition-algos}
and Sec. \ref{sub:Proximal-decomposition-algorithms_monotone_VI}
focus on algorithms for \emph{$P{}_{\boldsymbol{\Upsilon}}$} and
monotone NEPs, respectively; results in this sections will be the
building blocks for the more difficult issue of equilibrium selection
problem addressed in Sec. \ref{sub:Equilibrium-Selection-Monotone-NEP}.\vspace{-0.3cm}

\subsection{Best-response distributed algorithms for \emph{$P{}_{\boldsymbol{\Upsilon}}$}
NEPs\label{sub:Best-response-decomposition-algos}\vspace{-0.2cm}}

Since in a NEP every player is trying to minimize his own objective
function, a natural approach to compute a solution of a NEP is to
consider an iterative algorithm wherein all the players, given the
strategies of the others and according to a given scheduling (e.g.,
sequentially or simultaneously), update their own strategy by solving
their optimization problem (\ref{eq:G_def}). Here, we focus on a
very general class of best-response-based algorithms, namely the \emph{totally
asynchronous} best-response algorithms (in the sense specified in
\cite{Bertsekas_Book-Parallel-Comp}). In these schemes, some players
may update their strategies more frequently than others and they may
even use an outdated information about the strategy profile used by
the others; which is very appealing in many practical multiuser communication
systems, such as wireless ad-hoc networks or CR systems wherein synchronization
requirements are hard to enforce.

To provide a formal description of the algorithm, we need to introduce
some preliminary definitions. In an asynchronous scheme, the users
may not update their own strategies at each iteration; let denote
then by $\mathcal{T}_{i}\subseteq\mathcal{T}\subseteq\left\{ 0,1,2,\ldots\right\} $
the set of times at which player $i$ updates his own strategy $\mathbf{x}_{i}$,
denoted by $\mathbf{x}_{i}^{(n)}$ (thus, implying that, at time $n\notin\mathcal{T}_{i}$,
$\mathbf{x}_{i}^{(n)}$ is left unchanged). Moreover, in computing
their optimal strategy, the users can use an outdated version of the
others' strategies; let then $\tau_{j}^{i}(n)$ be the most recent
time at which the strategy profile of player $j$ is perceived by
player $i$ at the $n$-$th$ iteration (observe that $\tau_{j}^{i}(n)$
satisfies $0\leq\tau_{j}^{i}(n)\leq n$). Hence, if player $i$ updates
its strategy at the $n$-th iteration, then he minimizes his cost
function using the following (possibly) outdated strategy profile
of the other players: 
\begin{equation}
\mathbf{x}_{-i}^{(\mathbf{\boldsymbol{{\tau}}}^{i}(n))}\triangleq\left(\mathbf{x}_{1}^{(\tau_{1}^{i}(n))},\ldots,\mathbf{x}_{i-1}^{(\tau_{i-1}^{i}(n))},\mathbf{x}_{i+1}^{(\tau_{i+1}^{i}(n))},\ldots,\mathbf{x}_{I}^{(\tau_{I}^{i}(n))}\right).\label{eq:p_q_interference}
\end{equation}
Some standard conditions in asynchronous convergence theory, which
are fulfilled in any practical implementation, need to be satisfied
by the schedule $\mathcal{T}_{i}$'s and $\tau_{j}^{i}(n)$'s, namely
for each $i$: 
\begin{description}
\item [{{A1)}}] $0\leq\tau_{j}^{i}(n)\leq n$ (at any given iteration
$n$, each player $i$ can use only the strategy profile $\mathbf{x}_{-i}^{(\mathbf{\tau}^{i}(n))}$
adopted by the other players in the previous iterations ); 
\item [{{A2)}}] $\lim_{k\rightarrow\infty}\tau_{j}^{i}(n_{k})=+\infty$,
where $\{n_{k}\}$ is a sequence of elements in $\mathcal{T}_{i}$
that tends to infinity {[}for any given iteration index $n_{k}$,
the values of the components of $\mathbf{x}_{-i}^{(\mathbf{\tau}^{i}(n))}$
in (\ref{eq:p_q_interference}) generated prior to $n_{k}$ are not
used in the updates of $\mathbf{x}_{i}^{(n)}$, when $n$ becomes
sufficiently larger than $n_{k}${]}; 
\item [{{A3)}}] $\left\vert \mathcal{T}_{i}\right\vert =\infty$ (no
player fails to update his own strategy as time $n$ goes on). 
\end{description}
Using the above definitions, the totally asynchronous algorithm based
on the best-responses of the players is described in Algorithm \ref{async_best-response_algo}.
The convergence properties of the algorithm are given in Theorem \ref{Theo-async_best-response_NEP}.

\begin{algo}{Asynchronous Best-Response Algorithm} S\texttt{$\mbox{(\mbox{S.0})}:$}
Choose any feasible $\mathbf{x}^{(0)}\in\mathcal{Q}$ and set $n=0$.

\texttt{$\mbox{(S.1)}:$} \texttt{If} $\mathbf{x}^{(n)}$ satisfies
a suitable termination criterion: \texttt{STOP} \\
 \texttt{$\mbox{(S.2)}:$}\noun{ }\texttt{for} $i=1,\ldots,I$,$\,$compute
\begin{equation}
\mathbf{x}_{i}^{(n+1)}=\left\{ \begin{array}{lll}
\mathbf{x}_{i}^{\star}\in\underset{\mathbf{x}_{i}\in\mathcal{Q}_{i}}{\text{{argmin}}}f_{i}\left(\mathbf{x}_{i},\,\mathbf{x}_{-i}^{(\boldsymbol{{\tau}}^{i}(n))}\right), &  & \mbox{if }n\in\mathcal{T}_{i}\medskip\\
\mathbf{x}_{i}^{(n)}, &  & \mbox{otherwise}
\end{array}\right.\label{eq:Async_update}
\end{equation}
$\qquad\,$$\quad\,\,\,$\\
 \texttt{$\mbox{(S.3)}:$} $n\leftarrow n+1$; go to \texttt{$\mbox{(S.1)}$}.
\label{async_best-response_algo}\end{algo}

\begin{theorem} \label{Theo-async_best-response_NEP} Let $\mathcal{G}=\left\langle \mathcal{Q},\,\mathbf{f}\right\rangle $
be a \emph{$P{}_{\boldsymbol{\Upsilon}}$} NEP. Any sequence $\{\mathbf{x}^{(n)}\}_{n=0}^{\infty}$
generated by Algorithm \ref{async_best-response_algo} converges to
the unique NE of ${\mathcal{G}}$, for any given updating schedule
of the players satisfying assumptions A1-A3.\end{theorem}\vspace{-0.3cm}

\begin{proof} See Appendix \ref{Appendix:Proof-of-Theorem-convergence-best_response}.
\end{proof}\vspace{-0.2cm}

\begin{remark}[Flexibility of the algorithm]\label{Rmk_flexibility_bestREspAlg}
\emph{Algorithm \ref{async_best-response_algo} contains as special
cases a large number of algorithms, each one obtained by a possible
choice of the schedule of the users in the updating procedure (i.e.,
the parameters $\{\tau_{j}^{i}(n)\}$ and $\{\mathcal{T}_{i}\}$).
Examples are the }simultaneous\emph{ (Jacobi scheme) and }sequential\emph{
(Gauss-Seidel scheme) updates, where the players update their own
strategies }simultaneously\emph{ and }sequentially\emph{, respectively.
Indeed, the Jacobi update corresponds to the schedule $\tau_{j}^{i}(n)=n$
and $\mathcal{T}_{i}=\{1,2,\ldots\}$ for all $i$ and $j$, whereas
the Gauss-Seidel scheme is obtained by taking $\tau_{j}^{i}(n)=n$
and $\mathcal{T}_{i}=\{i,\, i+I,\, i+2I,\ldots\}$ for all $i$ and
$j$. Moreover, variations of such a totally asynchronous scheme,
e.g., including constraints on the maximum tolerable delay in the
updating and on the use of the outdated information (which leads to
the so-called }partially\emph{ asynchronous algorithms), can also
be considered \cite{Bertsekas_Book-Parallel-Comp}. An important result
stated in Theorem \ref{Theo-async_best-response_NEP} is that all
the algorithms resulting as special cases of Algorithm \ref{async_best-response_algo}
are guaranteed to reach the unique NE of the NEP, under the same set
of convergence conditions, since the matrix $\boldsymbol{{\Upsilon}}_{\mathbf{F}}$
does not depend on the particular choice of $\{\tau_{j}^{i}(n)\}$
and $\{\mathcal{T}_{i}\}$. Note that all the algorithms coming from
Algorithm \ref{async_best-response_algo} are robust against missing
or outdated updates of the players. This feature strongly relaxes
the constraints on the synchronization of the players' updates; which
makes this class of algorithms appealing in many practical distributed
systems.} 

\emph{Note that the (synchronous) projection-response algorithms for
monotone VIs (and thus NEPs) proposed in} \emph{\cite{Yin-Shanbhag-Mehta_TAC10}
and \cite{Facchinei-Pang_FVI03,Konnov_VI_book} are not guaranteed
to converge if applied to a P$_{\boldsymbol{{\Upsilon}}}$ NEP that
is not monotone.}\hfill{}$\square$\end{remark}\vspace{-0.5cm}

\begin{remark}[On the convergence conditions] \label{rem: strongly convex}
\emph{Global convergence of Algorithm \ref{async_best-response_algo}
is guaranteed under the P property of $\boldsymbol{\Upsilon}_{\mathbf{F}}$
(or equivalently $\rho(\boldsymbol{\Gamma}_{\mathbf{F}})<1$). However,
we have already pointed out in Remark \ref{Remark_uniqueness} that
such a condition cannot be satisfied if there is a player whose cost
function has a singular Hessian, even in just one point. In fact,
if this is the case, we have, $\alpha_{i}^{\min}=0$ for some $i$,
let us say $i=1$, which implies that $\boldsymbol{\Gamma}_{\mathbf{F}}$
has a 1 in the left-upper corner. Since $\boldsymbol{\Gamma}_{\mathbf{F}}$
is nonnegative, we have that this implies $\rho(\boldsymbol{\Gamma_{\mathbf{F}}})\geq1$
\cite[Th. 1.7..4]{Bapat-Raghavan_book}. Assuming that the element
1 is contained in an irreducible principal matrix, we will actually
have $\rho(\boldsymbol{\Gamma_{\mathbf{F}}})>1$. Note that the irreducibility
assumption is extremely weak and trivially satisfied if $\boldsymbol{\Gamma_{\mathbf{F}}}$
is positive, which is true in many applications. In the next section
we discuss a remedy for this issue. }\hfill{}$\square$\end{remark}\vspace{-0.5cm}

\begin{remark}[On the convergence rate] \emph{\label{Rmk_rate}As
shown in Proposition \ref{Proposition_contraction_best-response}
in Appendix \ref{Appendix:Proof-of-Theorem-convergence-best_response},
in the setting of Theorem \ref{Theo-async_best-response_NEP} the
best-response mapping {[}see (\ref{eq:best-response_mapping}){]}
is a contraction. Building on this and choosing for notational simplicity
in (\ref{eq:def_alpha_and_beta_Jac}) $\mathbf{C}_{i}=\mathbf{I}$
for all $i$, it is straightforward to show that the convergence rate
of the synchronous Jacobi version of Algorithm \ref{async_best-response_algo}
is geometric with factor $\left\Vert \boldsymbol{{\Gamma_{\mathbf{F}}}}\right\Vert <1$
(see, e.g., \cite[Prop. 1.1]{Bertsekas_Book-Parallel-Comp}). Therefore,
one can readily determine how many iterations are needed to surely
achieve a desired accuracy $\varepsilon>0$: 
\begin{equation}
\left\Vert \mathbf{x}^{n}-\mathbf{x}^{\star}\right\Vert \leq\varepsilon\quad\mbox{for}\,\,\mbox{any}\,\,\mbox{positive}\,\, n\geq\overline{{n}}\triangleq\log\left(\dfrac{{\varepsilon\,\left(1-\left\Vert \boldsymbol{{\Gamma_{\mathbf{F}}}}\right\Vert \right)}}{\left\Vert \mathbf{x}^{(1)}-\mathbf{x}^{(0)}\right\Vert }\right)/\log\left(\left\Vert \boldsymbol{{\Gamma_{\mathbf{F}}}}\right\Vert \right),\label{eq:error_bound}
\end{equation}
where $\mathbf{x}^{\star}$ is the unique NE of $\mathcal{G}$ and
$\left\Vert \boldsymbol{{\Gamma_{\mathbf{F}}}}\right\Vert <1$ is
the best-response contraction constant, with $\boldsymbol{{\Gamma_{\mathbf{F}}}}$
defined in (\ref{eq:Upsilon_matrix}). Note that when the joint feasible
set $\mathcal{Q}$ is bounded with diameter $d_{\mathcal{Q}}$, we
can obtain an overestimate of $\overline{{n}}$ that is independent
on $\mathbf{x}^{(1)}$ and $\mathbf{x}^{(2)}$: $\overline{{n}}\leq\log\left({\varepsilon\,\left(1-\left\Vert \boldsymbol{{\Gamma_{\mathbf{F}}}}\right\Vert \right)}/d_{\mathcal{Q}}\right)/\log\left(\left\Vert \boldsymbol{{\Gamma_{\mathbf{F}}}}\right\Vert \right)$. }

\emph{The case of asynchronous implementations is conceptually similar
but necessarily more complex; we refer the interested reader to \cite[Sec. 6.3.5]{Bertsekas_Book-Parallel-Comp}.}\hfill{}$\square$\end{remark}\vspace{-0.4cm}

\subsection{Proximal distributed algorithms for monotone NEPs\label{sub:Proximal-decomposition-algorithms_monotone_VI}\vspace{-0.1cm}}

In this section we deal with monotone NEPs (see Definition \ref{Def_monotone_NEP}).
Since monotone NEPs in general have multiple NE, Algorithm \ref{async_best-response_algo}
solving \emph{$P{}_{\boldsymbol{\Upsilon}}$} (or uniformly P) NEPs
may fail to converge. There is a host of solution methods available
in the literature to solve monotone real VIs and thus monotone NEPs
(see, e.g., \cite[Vol. II]{Facchinei-Pang_FVI03}), but these algorithms
are \emph{centralized}. Recently, in \cite{Yin-Shanbhag-Mehta_TAC10},
the authors proposed some distributed synchronous schemes for solving
monotone VIs, based on the gradient-response mapping; we have already
discussed the main drawbacks of these algorithms, see Sec. \ref{sec:Introduction-and-Motivation}
(see also Sec. \ref{sec:Applications} for some numerical results).

The development of distributed \emph{best-response} algorithms for
solving monotone NEPs with (possibly) multiple solutions is a challenging
task; in this subsection, we cope with this issue building on a regularization
technique known as proximal algorithms, see \cite[Ch 12]{Facchinei-Pang_FVI03}
for an introduction to proximal point methods for VIs. The proposed
approach is to reduce the solution of a {\em single} monotone NEP
to the solution of a {\em sequence} of \emph{$P{}_{\boldsymbol{\Upsilon}}$}
NEPs with a particular structure. The advantage of this method is
that we can efficiently solve each of the \emph{$P{}_{\boldsymbol{\Upsilon}}$}
NEPs with convergence guarantee using Algorithm 1 (cf. Sec. \ref{sub:Best-response-decomposition-algos});
the disadvantage is that, to recover the solution of the original
monotone NEP, one has to solve a (possibly infinite) number of \emph{$P{}_{\boldsymbol{\Upsilon}}$}
NEPs. However, it is important to remark from the outset that this
potential drawback is greatly mitigated by the fact that, as we discuss
shortly, (i) one only needs to solve these \emph{$P{}_{\boldsymbol{\Upsilon}}$}
NEPs inaccurately; (ii) the (inaccurate) solution of the \emph{$P{}_{\boldsymbol{\Upsilon}}$}
NEPs usually requires little computational effort; and (iii) in practice,
a fairly accurate solution of the original NEP is obtained after solving
a limited number of \emph{$P{}_{\boldsymbol{\Upsilon}}$} NEPs.

Before introducing the formal description of the algorithm, let us
begin with some simple observations motivating how the sequence of
\emph{$P{}_{\boldsymbol{\Upsilon}}$} NEPs is built. Let $\mathcal{G}=\left\langle \mathcal{Q},\mathbf{f}\right\rangle $
be a monotone NEP; consider a perturbation of this game defined as
$\mathcal{G}_{\tau,\mathbf{y}}=\left\langle \mathcal{Q},\,(f_{i}+(\tau/2)\|\bullet-\mathbf{y}_{i}\|^{2})_{i=1}^{I}\right\rangle $,
where $\tau$ is a positive parameter and $\mathbf{y}=(\mathbf{y}_{i})_{i=1}^{I}$
is a given vector in $\mathbb{R}^{n}$, with each $\mathbf{y}_{i}\in\mathbb{R}^{n_{i}}$;
we term $\mathbf{y}$ center of the regularization. Note that $\mathcal{G}_{\tau,\mathbf{y}}$
is the game wherein each player $i$, anticipating $\mathbf{x}_{-i}\in\mathcal{Q}_{-i}$,
solves the following convex optimization problem: 
\begin{equation}
\begin{array}{ll}
\limfunc{minimize}\limits _{\mathbf{x}_{i}} & f_{i}(\mathbf{x}_{i},\,\mathbf{x}_{-i})+\frac{\tau}{2}\|\mathbf{x}_{i}-\mathbf{y}_{i}\|^{2}\\[5pt]
\text{subject to} & \mathbf{x}_{i}\in{\cal Q}_{i}.
\end{array}\label{eq:game tau}
\end{equation}
Let us consider now the VI reformulations of $\mathcal{G}$ and $\mathcal{G}_{\tau,\mathbf{y}}$,
given by $\ensuremath{\limfunc{VI}(\mathcal{Q},\mathbf{F})}$ and
$\ensuremath{\limfunc{VI}(\mathcal{Q},\mathbf{F}_{\tau,\mathbf{y}})}$
respectively, where $\ensuremath{\mathbf{F}(\mathbf{x})\triangleq\left(\nabla_{\mathbf{x}_{i}}f_{i}(\mathbf{x})\right)_{i=1}^{I}}$
and $\mathbf{F}_{\tau,\mathbf{y}}(\mathbf{x})=\mathbf{F}(\mathbf{x})+\tau\,(\mathbf{x}-\mathbf{y})$,
and let us introduce the matrices $\boldsymbol{\Upsilon}_{\mathbf{F}}$
and $\boldsymbol{\Upsilon}_{\mathbf{F}_{\tau,\mathbf{y}}}$ associated
to $\mathcal{G}$ and $\mathcal{G}_{\tau,\mathbf{y}}$, respectively
{[}see (\ref{eq:Upsilon_matrix}){]}. It is not difficult to check
that\vspace{-0.1cm} 
\begin{equation}
\boldsymbol{\Upsilon}_{\mathbf{F}_{\tau,\mathbf{y}}}=\boldsymbol{\Upsilon}_{\mathbf{F}}+\tau\mathbf{I}.\vspace{-0.1cm}\label{eq:Upsilon_F_tau}
\end{equation}
Note that $\boldsymbol{\Upsilon}_{\mathbf{F}_{\tau,\mathbf{y}}}$
does not depend on $\mathbf{y}$. It follows readily from (\ref{eq:Upsilon_F_tau})
that if $\tau$ is large enough, $\Upsilon_{\mathbf{F}_{\tau,\mathbf{y}}}$
is a P matrix, meaning that $\mathcal{G}_{\tau,\mathbf{y}}$ is a
\emph{$P{}_{\boldsymbol{\Upsilon}}$} NEP, for any given $\mathbf{y}\in\mathbb{\mathbb{R}}^{n}$.
More specifically, using the definitions of $\beta_{ij}^{\max}$'s
and $\alpha_{i}^{\min}$'s as given in (\ref{eq:def_alpha_and_beta_Jac}),
we have the following.\vspace{-0.2cm}

\begin{lemma}\label{lem: monotone tau P} For any given $\mathbf{y}\in\mathbb{R}^{n}$,
the game $\mathcal{G}_{\tau,\mathbf{y}}=\left\langle \mathcal{Q},\,(f_{i}+(\tau/2)\cdot\right.$
$\left.\|\bullet-\mathbf{y}_{i}\|^{2})_{i=1}^{I}\right\rangle $ is
a \emph{$P{}_{\boldsymbol{\Upsilon}}$} NEP for every $\tau$ larger
than $\bar{\tau}$ (independent of $\mathbf{y}$), with\vspace{-0.2cm}
\begin{equation}
\bar{\tau}\triangleq\max_{1\leq i\leq I}\left\{ \sum_{j\neq i}\beta_{ij}^{\max}-\alpha_{i}^{\min}\right\} .\label{eq:upper_bound_on_tau}
\end{equation}
\end{lemma}

Nice as it is, the result above would be of no practical interest
if we were not able to connect the solutions of $\mathcal{G}_{\tau,\mathbf{y}}$
to those of $\mathcal{G}$. Indeed, the solutions of $\mathcal{G}$
and $\mathcal{G}_{\tau,\mathbf{y}}$ are in general different but,
nevertheless, there exists a connection between them: a point $\mathbf{x}^{\star}$
is a solution of $\mathcal{G}$ if and only if $\mathbf{x}^{\star}$
is a solution of $\mathcal{G}_{\tau,\mathbf{x}^{\star}}$.\vspace{-0.2cm}

\begin{proposition} \label{Proposition_G_and_G_tau_x}Let $\mathcal{G}=\left\langle \mathcal{Q},\mathbf{f}\right\rangle $
be a monotone NEP. For any given $\tau>0$, $\mathbf{x}^{\star}\in\mathcal{Q}$
is a solution of $\mathcal{G}$ if and only if $\mathbf{x}^{\star}$
is a solution of $\mathcal{G}_{\tau,\mathbf{x}^{\star}}$. \end{proposition}\vspace{-0.3cm}

\noindent \begin{proof}See Appendix \ref{Proof-of-Proposition_G_and_G_tau_x}.\end{proof}

\noindent \indent Lemma \ref{lem: monotone tau P} and Proposition
\ref{Proposition_G_and_G_tau_x} open the way to the design of convergent
distributed algorithms for monotone NEPs, as shown next. Let us choose
$\tau$ being large enough so that $\mathcal{G}_{\tau,\mathbf{y}}$
is a \emph{$P{}_{\boldsymbol{\Upsilon}}$} NEP (cf. Lemma \ref{lem: monotone tau P}).
It follows from Theorem \ref{Theo_existence_uniquenessNE} that $\mathcal{G}_{\tau,\mathbf{y}}$
has a unique solution, denoted by $\mathbf{S}_{\tau}(\mathbf{y})$.
Using $\mathbf{S}_{\tau}(\mathbf{y})$, Proposition \ref{Proposition_G_and_G_tau_x}
can be restated as follows: $\mathbf{x}^{\star}$ is a solution of
$\mathcal{G}$ if and only if it is a fixed point of $\mathbf{S}_{\tau}(\mathbf{\bullet})$,
i.e., $\mathbf{x}^{\star}=\mathbf{S}_{\tau}(\mathbf{x}^{\star})$.
It seems then natural to compute the solutions of $\mathcal{G}$ using
the fixed-point-type iteration $\mathbf{x}^{(n+1)}=\mathbf{S}_{\tau}(\mathbf{\mathbf{x}}^{(n)})$,
starting from a feasible point $\mathbf{x}^{(0)}$; which corresponds
to solving the sequence of NEPs $\mathcal{G}_{\tau,\mathbf{x}^{(n)}}$
for $n=0,1,\ldots$. If $\tau$ is sufficiently large {[}e.g., as
in (\ref{eq:upper_bound_on_tau}){]}, each $\mathcal{G}_{\tau,\mathbf{x}^{(n)}}$
is a \emph{$P{}_{\boldsymbol{\Upsilon}}$} NEP (cf. Lemma \ref{lem: monotone tau P}),
and thus its unique solution can be computed in a distributed way
with convergence guarantee by the asynchronous best-response algorithm
described in Algorithm \ref{async_best-response_algo} (cf. Theorem
\ref{Theo-async_best-response_NEP}). The above discussion motivates
the following algorithm for computing the solutions of a monotone
NEP, whose convergence properties are given in Theorem \ref{Exact ProxDecAlg_conv_theo}
below.

\begin{algo}{Proximal Decomposition Algorithm (PDA)} S$\textbf{Data}:$
Let $\tau>0$ be given. \\[1pt] \texttt{$\mbox{(\mbox{S.0})}:$}
Choose any feasible $\mathbf{x}^{(0)}\in\mathcal{Q}$ and set $n=0$.

\texttt{$\mbox{(S.1)}:$} \texttt{If} $\mathbf{x}^{(n)}$ satisfies
a suitable termination criterion: STOP.

\texttt{$\mbox{(S.2)}:$} Solve the game $\mathcal{G}_{\tau,\mathbf{x}^{(n)}}$
and set $\mathbf{x}^{(n+1)}\triangleq\mathbf{S}_{\tau}(\mathbf{x}^{(n)})$

\texttt{$\mbox{(S.3)}:$} $n\leftarrow n+1$; go to \texttt{$\mbox{(S.1)}.$}
\label{alg:PDA}\end{algo}

\begin{theorem}\label{Exact ProxDecAlg_conv_theo} Let $\mathcal{G}=\left\langle \mathcal{Q},\,\mathbf{f}\right\rangle $
be a monotone NEP with a nonempty solution set. Suppose that $\tau$
is large enough so that $\boldsymbol{\Upsilon}_{\mathbf{F}_{\tau,\,\mathbf{y}}}$
is a P matrix. Then, Algorithm \ref{alg:PDA} is well defined, and
the sequence $\{\mathbf{x}^{(n)}\}_{n=0}^{\infty}$ generated by the
algorithm converges to a solution of the game $\mathcal{G}$. \end{theorem}

\noindent \indent Algorithm \ref{alg:PDA} is of great conceptual
interest, but its applicability is limited, unless one is able to
easily compute $\mathbf{S}_{\tau}(\mathbf{x}^{(n)})$. Although there
are interesting problems in which this can be done efficiently (see
Sec. \ref{sec:Applications}), in general one is expected to solve
a number of \emph{$P{}_{\boldsymbol{\Upsilon}}$} NEPs, each of them
requiring an infinite iterative method to compute each $\mathbf{S}_{\tau}(\mathbf{x}^{(n)})$.
To overcome this issue, we propose next a variant of Algorithm \ref{alg:PDA},
in which suitable approximations of $\mathbf{S}_{\tau}(\mathbf{x}^{(n)})$
can be used. Algorithm \ref{alg:APDA} below describes such a variant,
where we have added a further degree of freedom in the updating rule:
the new iteration $\mathbf{x}^{(n+1)}$ is not necessarily given by
(an approximation of) $\mathbf{S}_{\tau}(\mathbf{x}^{(n)})$, but
lies instead on the line connecting the old iteration $\mathbf{x}^{(n)}$
to (the approximation of) $\mathbf{S}_{\tau}(\mathbf{x}^{(n)})$.\vspace{-0.4cm}

\noindent \begin{algo}{Approximate Proximal Decomposition Algorithm
(APDA)} S$\textbf{Data}:$ Let $\{\varepsilon^{(n)}\}_{n=0}^{\infty}$,
$\{\eta^{(n)}\}_{n=0}^{\infty}$ and $\tau>0$ be given.\\[1pt] \texttt{$\mbox{(\mbox{S.0})}:$}
Choose any feasible $\mathbf{x}^{(0)}\in\mathcal{Q}$ and set $n=0$.

\texttt{$\mbox{(S.1)}:$} \texttt{If} $\mathbf{x}^{(n)}$ satisfies
a suitable termination criterion: STOP.

\texttt{$\mbox{(S.2)}:$} Solve the game $\mathcal{G}_{\tau,\mathbf{x}^{(n)}}$
within the accuracy $\varepsilon^{(n)}$: Find a $\mathbf{z}^{(n)}$
s.t. $\|\mathbf{z}^{(n)}-\mathbf{S}_{\tau}(\mathbf{x}^{(n)})\|\leq\varepsilon^{(n)}$.

\texttt{$\mbox{(S.3)}:$} Set $\mathbf{x}^{(n+1)}\triangleq(1-\eta^{(n)})\mathbf{x}^{(n)}+\eta^{(n)}\mathbf{z}^{(n)}$.

\texttt{$\mbox{(S.4)}:$} $n\leftarrow n+1$; go to \texttt{$\mbox{(S.1)}.$}
\label{alg:APDA}\end{algo}

\noindent \indent The error term $\varepsilon^{(n)}$ measures the
accuracy used at iteration $n$ in computing the solution $\mathbf{S}_{\tau}(\mathbf{x}^{(n)})$
of $\mathcal{G}_{\tau,\mathbf{x}^{(n)}}$. The parameter $\eta^{(n)}$
instead, introduces a memory in the updating rule, it establishes
where exactly we move along the line passing through the old iterations
$\mathbf{x}^{(n)}$ and $\mathbf{z}^{(n)}$. Note that if we take
$\varepsilon^{(n)}=0$ and $\eta^{(n)}=1$ for all $n$, Algorithm
\ref{alg:APDA} reduces to Algorithm \ref{alg:PDA}. The advantage
of Algorithm \ref{alg:APDA} with respect to Algorithm \ref{alg:PDA}
is that $\mathbf{z}^{(n)}$ can be computed in a finite number of
steps, so that Algorithm \ref{alg:APDA} becomes implementable in
practice. Obviously, the errors $\varepsilon^{(n)}$'s and the parameters
$\eta^{(n)}$'s must be chosen according to some suitable conditions,
if one wants to guarantee convergence. These conditions are established
in the following theorem.\vspace{-0.2cm}

\begin{theorem}\label{ProxDecAlg_conv_theo} Let $\mathcal{G}=\left\langle \mathcal{Q},\,\mathbf{f}\right\rangle $
be a monotone NEP with a nonempty solution set. Suppose that $\tau$
is large enough so that $\boldsymbol{\Upsilon}_{\mathbf{F}_{\tau,\,\mathbf{y}}}$
is a P-matrix. Choose $\{\varepsilon^{(n)}\}\subset[0,\infty)$ such
that $\sum_{n=1}^{\infty}\varepsilon^{(n)}<\infty$ and $\{\eta^{(n)}\}\subset[R_{m},R_{M}]$,
with $0<R_{m}\leq R_{M}<2$. Then, Algorithm \ref{alg:APDA} is well
defined, and the sequence $\{\mathbf{x}^{(n)}\}_{n=0}^{\infty}$ generated
by the algorithm converges to a solution of $\mathcal{G}$. \end{theorem}\vspace{-0.2cm}

The proof of Theorem \ref{ProxDecAlg_conv_theo} (and thus also Theorem
\ref{Exact ProxDecAlg_conv_theo}) is a consequence of the following
facts and thus is omitted: i) \cite[Th. 12..3.9]{Facchinei-Pang_FVI03};
ii) The observation that $\mathcal{G}$ is equivalent to the VI$(\mathcal{Q},\mathbf{F})$
(Proposition \ref{VI_reformulation}), with $\mathbf{F}=(\nabla_{\mathbf{x}_{i}}f_{i})_{i=1}^{I}$,
and the VI$(\mathcal{Q},\mathbf{F})$ has a solution; and iii) Under
the P property of $\boldsymbol{\Upsilon}_{\mathbf{F}_{\tau,\,\mathbf{y}}}$,
Step S.2 of Algorithm \ref{alg:APDA} (and Algorithm \ref{alg:PDA})
is well defined (Theorem \ref{Theo-async_best-response_NEP}).

It is interesting to remark that for sake of simplicity we assumed
$\tau$ to be a fixed number. However, $\tau$ can be varied from
iteration to iteration provided that $\tau\in(\bar{{\tau}},\tau^{\max}]$,
where $\tau^{\max}$ is any finite number. We also note that the sequence
$\{\mathbf{x}^{(n)}\}$ generated by Algorithm \ref{alg:APDA} may
not be feasible, but all the limit points are feasible. If one is
interested in maintaining feasibility throughout the iterates, it
is enough to choose $\eta^{(n)}\leq1$ and compute in Step 2 a feasible
$\mathbf{z}^{(n)}$, which can be done, e.g., by applying Algorithm
\ref{async_best-response_algo} to $\mathcal{G}_{\tau,\mathbf{x}^{(n)}}$. 

While the utility of having possibly inexact solutions in Step S.2
of Algorithm \ref{alg:APDA} is apparent, the usefulness of Step S.3
is less evident. This kind of ``averaging'' is known as over-relaxation
and has its roots in classical successive over-relaxation methods
for solving systems of linear equations \cite[Sec. 7.4]{Ortega-Rheinboldt_book}.
In our context, the extra degree of freedom offered by Step S.3 can
bring numerical improvements; see, e.g., \cite{Facchinei-Pang_FVI03}.

\noindent \indent Algorithm \ref{alg:APDA} is conceptually a double
loop scheme wherein at each (outer) iteration $n$, given $\mathbf{x}^{(n)}$,
one needs to compute the approximation $\mathbf{z}^{(n)}$, which
requires an inner iterative process. Since the condition $\sum_{n=1}^{\infty}\varepsilon^{(n)}<\infty$
implies $\varepsilon^{(n)}\downarrow0$, when the iterations progress,
$\mathbf{S}_{\tau}(\mathbf{x}^{(n)})$ has to be estimated with an
increasing accuracy. However, in practice, this is not a problem since
when iterations progress, $\{\mathbf{x}^{(n)}\}$ usually converges,
implying $\|\mathbf{x}^{(n+1)}-\mathbf{x}^{(n)}\|\to0$. One can then
use $\mathbf{x}^{(n)}$ as a good approximation to initialize any
(inner) procedure in Step 2 to compute $\mathbf{z}^{(n)}$. It turns
out that, in spite of the increasing precision requirements, in practice
a suitable $\mathbf{z}^{(n)}$ in Step 2 can be computed very easily.

Finally, observe that a natural choice for computing $\mathbf{z}^{(n)}$
in Step 2 of Algorithm \ref{alg:APDA} is Algorithm \ref{async_best-response_algo}.
When this choice is made, Algorithm \ref{alg:APDA} also becomes an
asynchronous method, having all the desired features described in
the previous section (see Remark \ref{Rmk_flexibility_bestREspAlg}).
The only difference with Algorithm \ref{async_best-response_algo}
is that, in Algorithm \ref{alg:APDA}, ``from time to time'' (precisely
when the inner termination test $\|\mathbf{z}^{(n)}-\mathbf{S}_{\tau}(\mathbf{x}^{(n)})\|\leq\varepsilon^{(n)}$
in Step 2 is satisfied) the objective function of the players are
changed by updating the regularizing term from $\frac{\tau}{2}\|\mathbf{x}_{i}-\mathbf{x}_{i}^{(n)}\|^{2}$
to $\frac{\tau}{2}\|\mathbf{x}_{i}-\mathbf{x}_{i}^{(n+1)}\|^{2}$,
which generally requires some coordination among the players to establish
when a satisfactory approximation $\mathbf{z}^{(n)}$ has been reached.
The remark below discusses issues related to this aspect.\vspace{-0.2cm}

\begin{remark}[On the inner termination criterion] \emph{\label{Remark_termination}We
have seen that, in Step 2 of Algorithm \ref{alg:APDA}, the players
must be able to decide whether $\|\mathbf{z}^{(n)}-\mathbf{S}_{\tau}(\mathbf{x}^{(n)})\|\,\leq\,\varepsilon^{(n)}$
holds. This can be easily done if one uses a synchronous Jacobi version
of Algorithm \ref{async_best-response_algo};} \emph{in fact one can
readily estimate the number of iterations needed to achieve the accuracy
$\varepsilon^{(n)}$ using (\ref{eq:error_bound}), where $\boldsymbol{{\Gamma}}_{\mathbf{F}}$
is replaced by $\boldsymbol{{\Gamma}}_{\mathbf{F}_{\tau}}$. However,
this estimate can be very conservative and, in any case, it is not
applicable if an asynchronous version of Algorithm \ref{async_best-response_algo}
is used. In the following we suggest a different, simple, distributed
protocol to decide whether $\|\mathbf{z}^{(n)}-\mathbf{S}_{\tau}(\mathbf{x}^{(n)})\|\,\leq\,\varepsilon^{(n)}$,
which can be used in both synchronous and asynchronous implementations
of Algorithm \ref{async_best-response_algo}. }

\emph{Observe preliminarily that an error bound on the distance of
the current iteration $\mathbf{z}^{(n)}$ from the solution $\mathbf{S}_{\tau}(\mathbf{x}^{(n)})$
of $\mathcal{G}_{\tau,\mathbf{x}^{(n)}}$ can be obtained by solving
a convex (quadratic) problem (see, e.g., \cite[Prop. 6.3.1]{Facchinei-Pang_FVI03},
\cite[Prop. 6.3.7]{Facchinei-Pang_FVI03}). For example, under the
P properties of $\boldsymbol{\Upsilon}_{\mathbf{F}_{\tau,\mathbf{x}^{(n)}}}$
defined in (\ref{eq:Upsilon_F_tau}), the following error bound holds
for the game $\mathcal{G}_{\tau,\mathbf{x}^{(n)}}$ \cite[Prop. 6.3.1]{Facchinei-Pang_FVI03}:
a (finite) constant $c>0$ exists such that} 
\begin{equation}
\|\mathbf{z}-\mathbf{S}_{\tau}(\mathbf{x}^{(n)})\|\,\leq c\,\|\mathbf{F}_{\tau}^{\text{\emph{nat}}}\left(\mathbf{z}\right)\|,\quad\forall\mathbf{z}\in\mathcal{Q},\label{eq:error_bound_1}
\end{equation}
\emph{where $\mathbf{F}_{\tau}^{\text{{nat}}}\left(\mathbf{z}\right)\triangleq\mathbf{z}-\Pi_{\mathcal{Q}}\left(\mathbf{z}-\mathbf{F}(\mathbf{z})-\tau\,\mathbf{z}\right)$,
with $\Pi_{\mathcal{Q}}\left(\mathbf{x}\right)$ denoting the Euclidean
projection of $\mathbf{x}$ onto the closed and convex set $\mathcal{Q}$.
Note that, since $\mathcal{Q}$ has a Cartesian structure, $\mathbf{F}_{\tau}^{\text{{nat}}}\left(\mathbf{z}\right)$
can be partitioned as $\mathbf{F}_{\tau}^{\text{{nat}}}\left(\mathbf{z}\right)=\left(\left[\mathbf{F}_{\tau}^{\text{{nat}}}\left(\mathbf{z}\right)\right]_{i}\right)_{i=1}^{I}$,
where each $\left[\mathbf{F}_{\tau}^{\text{{nat}}}\left(\mathbf{z}\right)\right]_{i}=\mathbf{z}_{i}-\Pi_{\mathcal{Q}_{i}}\left(\mathbf{z}_{i}-\mathbf{F}_{i}(\mathbf{z})-\tau\,\mathbf{z}_{i}\right)$
can be locally computed by the associated player $i$ by solving a
quadratic program, as long as $\mathbf{F}_{i}(\mathbf{z})$ is available
at the player side.}%
\footnote{In many practical applications, as those considered in Sec. \ref{sub:Some-motivating-examples}
and Sec. \ref{Sec_applications}, each $\left[\mathbf{F}_{\tau}^{\text{{nat}}}\left(\mathbf{z}\right)\right]_{i}$
can be computed by local measurements from the players. 
}

\emph{Using (\ref{eq:error_bound_1}), the implementation of Step
2 of Algorithm \ref{alg:APDA} can be obtained as follows. Each player
$i$ choses preliminarily a suitable local termination sequence $\{\varepsilon_{i}^{(n)}\}_{n}\subset[0,\infty)$
such that $\sum_{n=1}^{\infty}\varepsilon_{i}^{(n)}<\infty$; the
termination criterion of each player $i$ becomes then $\left\Vert \left[\mathbf{F}_{\tau}^{\text{{nat}}}\left(\mathbf{x}^{(n)}\right)\right]_{i}\right\Vert \leq\varepsilon_{i}^{(n)}$,
which can be locally implemented. Once the desired local accuracy
is reached by all the players, they can all update the center of their
regularization. Note that this protocol guarantees that the requirement
on the sequence $\varepsilon^{(n)}$ in Step 2 as stated in Theorem
\ref{ProxDecAlg_conv_theo} is met, since $\varepsilon^{(n)}\triangleq\sum_{i=1}^{I}\varepsilon_{i}^{(n)}$
satisfies $\sum_{n=1}^{\infty}\varepsilon^{(n)}<\infty$.} \hfill{}$\square$\end{remark}\vspace{-0.7cm}

\begin{remark}[On the communications overhead] \emph{The last issue
to address for a practical implementation of the termination protocols
discussed in Remark \ref{Remark_termination} is how the players can
detect the others having reached the desired accuracy in Step 2 of
Algorithm \ref{alg:APDA}. If one uses a synchronous Jacobi version
of Algorithm \ref{async_best-response_algo} in Step 2 and the joint
feasible set $\mathcal{Q}$ is bounded, conceptually there is no need
of any information exchange, since each player can locally estimate
the number of iterations needed to reach the accuracy $\varepsilon^{(n)}$
as discussed in Remark \ref{Rmk_rate}. However, even when possible,
this approach is probably too conservative, since the estimated $\overline{{n}}$
can be unnecessarily large. In practice, and whenever one is not using
a synchronous Jacobi method, the termination criterion (\ref{eq:error_bound_1})
can be easily implemented, at the cost of limited signaling, performing
the following protocol. Each user sends out one bit when $\left\Vert \left[\mathbf{F}_{\tau}^{\text{{nat}}}\left(\mathbf{x}^{(n)}\right)\right]_{i}\right\Vert \leq\varepsilon_{i}^{(n)}$.
Once one and therefore all players receive bits from all the others,
they can update their regularization. Note that in the context of
distributed algorithms for the solution of optimization problems,
VIs, and games, this is a minimal signaling requirement; see, e.g.,
\cite{Palomar-Chiang_ACTran07-Num,ConvexSumSeparable-2} and \cite[Ch. 1.3, 1.4]{Bertsekas_Book-Parallel-Comp},
\cite[Ch. 3.5.7]{Bertsekas_Book-Parallel-Comp}, \cite[Ch. 8]{Bertsekas_Book-Parallel-Comp}
for some representative examples. Furthermore, we observe that in
many practical applications exchanging one bit is viable. For instance,
in the CR systems considered in Sec. \ref{sub:Some-motivating-examples}
and Sec. \ref{Sec_applications}, this can be done using network control
channels. }

\emph{Finally, we conclude this remark suggesting a simple heuristic
that does not require any signaling: each user updates its regularization
only after experiencing no changes in $\left\Vert \left[\mathbf{F}_{\tau}^{\text{{nat}}}\left(\mathbf{x}^{(n)}\right)\right]_{i}\right\Vert $
(or his cost function) for a prescribed number of iterations. While
this heuristic does not guarantee theoretical convergence, its effectiveness
in many practical scenarios, as those considered in Sec. \ref{sec:Applications},
is rather apparent.} \hfill{}$\square$ \end{remark}\vspace{-0.5cm}

\subsection{Equilibrium selection for monotone NEPs\label{sub:Equilibrium-Selection-Monotone-NEP}
\vspace{-0.1cm}}

In the previous section we discussed distributed algorithms for the
computation of a solution of monotone NEPs. A feature of these algorithms
is that they converge under mild conditions that do not imply the
uniqueness of the NE of the NEP. In the presence of multiple equilibria,
however, the proposed algorithms do not allow to perform any selection
of the solution they reach, but they may converge in principle to
\emph{any }NE of the game; which makes the achievable system performance
unpredictable. It would be interesting instead to be able to select,
among all the solutions of the game, the one(s) that satisfies some
additional criterion. We refer to this problem as {\em equilibrium
selection problem}. In this section, we address this issue; the outcome
will be a novel set of distributed algorithms along with their convergence
properties that solve the equilibrium selection problem; this additional
feature comes at the price of a (moderate) increase of the complexity
in computing the players' best-response solution and signaling among
the players.

Let us introduce first an informal description of the algorithm. Let
$\mathcal{G}=<\mathcal{Q},\,\mathbf{f}>$ be a monotone NEP and let
$\sol(\mathcal{Q},\,\mathbf{f})$ denote its solution set, assumed
to be nonempty without loss of generality. Recall that since $\mathcal{G}$
is monotone, $\sol(\mathcal{Q},\,\mathbf{f})$ is always convex (cf.
Theorem \ref{Theo_existence_uniquenessNE}). Stated in mathematical
terms, the equilibrium selection problem consists in solving the following
bi-level optimization problem:\vspace{-0.2cm}

\begin{equation}
\begin{array}{ll}
\limfunc{minimize}\limits _{\mathbf{x}} & \phi(\mathbf{x})\\[5pt]
\text{subject to} & \mathbf{x}\in\sol(\mathcal{Q},\,\mathbf{f}),
\end{array}\medskip\vspace{-0.2cm}\label{eq:VI-C min}
\end{equation}
where the function $\phi:\mathbb{R}^{n}\rightarrow\mathbb{R}$ is
assumed to be continuously differentiable and convex. The function
$\phi$ thus defines the additional criterion according to which one
wants to select a solution in the set of the NE of $\mathcal{G}$:
solving \eqref{eq:VI-C min} indeed corresponds to choosing the NE
of $\mathcal{G}$ that minimizes $\phi$. Note that, under the monotonicity
of $\mathcal{G}$, \eqref{eq:VI-C min} is a convex optimization problem
{[}$\sol(\mathcal{Q},\,\mathbf{f})$ is convex{]}. However, standard
solution techniques cannot be applied because the feasible set $\sol(\mathcal{K},\mathbf{f})$
is only implicitly defined and, in general, it is not expressed as
a standard system of inequalities. To overcome this difficulty, and
in the same spirit of the previous section, instead of attacking problem
\eqref{eq:VI-C min} directly, we propose to solve a sequence of standard
regularized NEPs (``standard'' means a game whose players' feasible
sets are not of an implicit type, and therefore can be solved by classic
methods, like Algorithm \ref{async_best-response_algo}). Each standard
regularized game has the following structure $\mathcal{G}_{\tau,\varepsilon,\mathbf{y}}=\left\langle \mathcal{Q},\,(f_{i}+\varepsilon\,\phi+(\tau/2)\|\bullet-\mathbf{y}_{i}\|^{2})_{i=1}^{I}\right\rangle $,
where $\varepsilon$ and $\tau$ are fixed positive constants and
$\mathbf{y}\triangleq(\mathbf{y}_{i})_{i=1}^{I}$ is a given point
in $\mathbb{R}^{n}$ with each $\mathbf{y}_{i}\in\mathbb{R}^{n_{i}}$;
$\mathcal{G}_{\tau,\varepsilon,\mathbf{y}}$ is a NEP wherein each
player $i$, anticipating $\mathbf{x}_{-i}\in\mathcal{Q}_{-i}$, solves
the following convex optimization problem:\vspace{-0.2cm}

\begin{equation}
\begin{array}{ll}
\limfunc{minimize}\limits _{\mathbf{x}_{i}} & f_{i}(\mathbf{x}_{i},\,\mathbf{x}_{-i})+\varepsilon\,\phi(\mathbf{x}_{i},\,\mathbf{x}_{-i})+\frac{\tau}{2}\|\mathbf{x}_{i}-\mathbf{y}_{i}\|^{2}\\[5pt]
\text{subject to} & \mathbf{x}_{i}\in{\cal Q}_{i}.
\end{array}\label{eq:game tau epsilon}
\end{equation}
Note that the players' problems in this game differ from \eqref{eq:game tau}
in the presence of the additional term $\varepsilon\phi(\mathbf{x})$
in the objective function. It is not surprising that $\phi$ appears
in the players' objective functions; roughly speaking, it represents
the additional amount of information to be included in the game to
``drive'' the system toward the desired solution.

Proceeding as in the previous section, we can now establish the connection
between the regularized NEPs $\mathcal{G}_{\tau,\varepsilon,\mathbf{y}}$
and the equilibrium selection problem \eqref{eq:VI-C min}. Fist of
all note that, in the setting of problem \eqref{eq:VI-C min}, the
NEPs $\mathcal{G}_{\tau,\varepsilon,\mathbf{y}}$ is equivalent to
the VI$(\mathcal{Q},\mathbf{F}_{\tau,\varepsilon,\mathbf{y}})$ where
$\mathbf{F}_{\tau,\varepsilon,\mathbf{y}}\triangleq\mathbf{F}+\varepsilon\,\nabla\phi+\tau\,(\mathbf{I}-\mathbf{y})$
and $\mathbf{I}:\mathbf{x}\mapsto\mathbf{x}$ is the identity map.
Denoting by $\boldsymbol{\Upsilon}_{\mathbf{F}_{\tau,\varepsilon,\mathbf{y}}}$
the matrix defined in (\ref{eq:Upsilon_matrix}) and associated to
$\mathbf{F}_{\tau,\varepsilon,\mathbf{y}}$, Lemma \ref{lem: monotone tau epsilon P}
below shows that there exists a sufficiently large $\tau$ such that
$\boldsymbol{\Upsilon}_{\mathbf{F}_{\tau,\varepsilon,\mathbf{y}}}$
is a P matrix, implying that $\mathcal{G}_{\tau,\varepsilon,\mathbf{y}}$
is a \emph{$P{}_{\boldsymbol{\Upsilon}}$} NEP.\vspace{-0.1cm}

\begin{lemma}\label{lem: monotone tau epsilon P} Let $\mathcal{G}=\left\langle \mathcal{Q},\,\mathbf{f}\right\rangle $
be a monotone NEP; let $\phi:\mathbb{R}^{n}\mapsto\mathbb{R}$ be
a continuously differentiable function on $\mathcal{Q}$ whose gradient
$\nabla\phi$ is Lipschitz continuous on $\mathcal{Q}$, with constant
$L_{\phi}$; and let $\bar{\varepsilon}>0$ be given. For any fixed
$\mathbf{y}\in\mathbb{R}^{n}$, the game $\mathcal{G}_{\tau,\varepsilon,\mathbf{y}}=\left\langle \mathcal{Q},\,(f_{i}+\varepsilon\phi+(\tau/2)\|\bullet-\mathbf{y}_{i}\|^{2})_{i=1}^{I}\right\rangle $
with $\varepsilon\in[0,\bar{\varepsilon}]$ is a \emph{$P{}_{\boldsymbol{\Upsilon}}$}
NEP for every $\tau$ larger than $\bar{\tau}_{\bar{\varepsilon}}$
\emph{(}independent on $\mathbf{y}$ and $\varepsilon$\emph{)}\vspace{-0.3cm}
\begin{equation}
\bar{\tau}_{\bar{\varepsilon}}\triangleq\max_{1\leq i\leq I}\left\{ \sum_{j\neq i}\beta_{ij}^{\max}-\alpha_{i}^{\min}\right\} +(I-1)\,\bar{\varepsilon}\, L_{\phi},\vspace{-0.2cm}\label{eq:upper_bound_tau_epsilon}
\end{equation}
where $\beta_{ij}^{\max}$'s and $\alpha_{i}^{\min}$'s are defined
in (\ref{eq:def_alpha_and_beta_Jac}).\end{lemma}\vspace{-0.2cm}

In the setting of Lemma \ref{lem: monotone tau epsilon P}, $\mathcal{G}_{\tau,\varepsilon,\mathbf{y}}$
is a \emph{$P{}_{\boldsymbol{\Upsilon}}$} NEP and thus has a unique
solution, denoted by $\mathbf{S}_{\tau,\varepsilon}(\mathbf{y})$
{[}cf. Theorem \ref{Theo_existence_uniquenessNE}{]}; such a $\mathbf{S}_{\tau,\varepsilon}(\mathbf{y})$
can be computed with convergence guaranteed using Algorithm \ref{async_best-response_algo}
on $\mathcal{G}_{\tau,\varepsilon,\mathbf{y}}$. The solution of the
original equilibrium selection problem \eqref{eq:VI-C min} can be
recovered using the fixed-point-type iteration $\mathbf{x}^{(n+1)}=\mathbf{S}_{\tau,\varepsilon^{(n)}}(\mathbf{\mathbf{x}}^{(n)})$,
starting from a feasible point $\mathbf{x}^{(0)}$ and by suitably
varying $\varepsilon^{(n)}$; which corresponds to solve the sequence
of NEPs $\mathcal{G}_{\tau,\varepsilon^{(n)},\mathbf{x}^{(n)}}$ for
$n=0,1,\ldots$. This procedure is made formal in Algorithm \ref{algo2}
below.

\begin{algo}{Proximal-Tikhonov Regularization Algorithm (PTRA)}
S$\textbf{Data}:$ Let $\{\varepsilon^{(n)}\}\downarrow0$ and $\tau>0$
be given. \\[1pt] \texttt{$\mbox{(\mbox{S.0})}:$} Choose any feasible
$\mathbf{x}^{(0)}\in\mathcal{Q}$ and set $n=0$. \\[1pt] 
\texttt{$\mbox{(\mbox{S.1})}:$} If $\mathbf{x}^{(n)}$ satisfies
a suitable termination criterion, \texttt{STOP}.\\[1pt] 
\texttt{$\mbox{(\mbox{S.2})}:$} Set $\mathbf{x}^{(n+1)}$ to be the
solution of the game $\mathcal{G}_{\tau,\varepsilon^{(n)},\mathbf{x}^{(n)}}$.
\\[1pt] \texttt{$\mbox{(\mbox{S.3})}:$} Set $n\leftarrow n+1$ and
return to \texttt{$\mbox{(\mbox{S.1})}$}. \label{algo2} \end{algo}

Note that, since in Algorithm \ref{algo2} the sequence $\{\varepsilon^{(n)}\}$
converges to zero, there always exists an $\bar{\varepsilon}>0$ such
that $\varepsilon^{(n)}\in[0,\bar{\varepsilon}]$, implying by Lemma
\ref{lem: monotone tau epsilon P} that for a sufficiently large $\tau$
all the games $\mathcal{G}_{\tau,\varepsilon^{(n)},\mathbf{x}^{(n)}}$
in the Step 2 of the algorithm are \emph{$P{}_{\boldsymbol{\Upsilon}}$}
NEPs and thus have a unique solution; this makes the sequence $\{\mathbf{x}^{(n)}\}_{n=0}^{\infty}$
generated by Algorithm \ref{algo2} well defined. The convergence
properties of the algorithm are given in the following theorem.\vspace{-0.1cm}

\begin{theorem}\label{the:mod} Let $\mathcal{G}=\left\langle \mathcal{Q},\,\mathbf{f}\right\rangle $
be a monotone NEP with a nonempty solution set $\sol(\mathcal{Q},\,\mathbf{f})$.
Consider the equilibrium selection problem \eqref{eq:VI-C min} and
suppose that i) $\phi$ is continuously differentiable and convex
on $\mathcal{Q}$; ii) the level sets of $\phi$ on $\sol(\mathcal{Q},\,\mathbf{f})$
are bounded; and iii) $\nabla\phi$ is Lipschitz continuous on \emph{$\mathcal{Q}$},
with constant $L_{\phi}$. Moreover, suppose that $\tau$ is large
enough so that $\mathcal{G}_{\tau,\varepsilon^{(n)},\mathbf{x}^{(n)}}$
is a \emph{$P{}_{\boldsymbol{\Upsilon}}$} NEP for any $n$, and choose
the sequence $\{\varepsilon^{(n)}\}$ such that $\varepsilon^{(n)}>0$
for all $n$, $\{\varepsilon^{(n)}\}\downarrow0$, and $\sum_{n=0}^{\infty}\,\varepsilon^{(n)}=\infty$.
Then Algorithm \ref{algo2} is well defined; the sequence $\{\mathbf{x}^{(n)}\}_{n=0}^{\infty}$
is bounded; and every of its limit points is a solution of \eqref{eq:VI-C min}.
\end{theorem} \vspace{-0.2cm}\begin{proof} See Appendix \ref{sec:Proof-of-Theorem_NE_Selection}.
\end{proof} 

Theorem \ref{the:mod} guarantees convergence of the algorithm under
mild assumptions. Conditions on $\phi$ are pretty standard; in particular,
assumptions i) and ii) together with the monotonicity of $\mathcal{G}$
state that the optimization problem \eqref{eq:VI-C min} is convex
and admits a solution; whereas iii) guarantees that there exists a
finite (large enough) $\tau$ such that $\mathcal{G}_{\tau,\varepsilon^{(n)},\mathbf{x}^{(n)}}$
is a \emph{$P{}_{\boldsymbol{\Upsilon}}$} NEP (cf. Lemma \ref{lem: monotone tau epsilon P}).
The assumptions on the sequence $\{\varepsilon^{(n)}\}_{n=1}^{\infty}$
are also rather weak and require $\varepsilon^{(n)}$ to go to zero,
but not too fast; which can be satisfied, e.g., by taking $\varepsilon^{(n)}=1/(1+na)$,
with $n=0,1,2,\ldots$ and $a$ being any positive constant. This
assumption on $\varepsilon^{(n)}$ is not new and has already been
used in a few papers dealing with the combination of Tikhonov and
proximal regularization; we refer the reader to \cite{Cabot_SIAM05}
and references therein for a wider discussion on this point.

The implementation of Algorithm \ref{algo2} requires the ability
of solving at each round $n$ the \emph{$P{}_{\boldsymbol{\Upsilon}}$}
NEP $\mathcal{G}_{\tau,\varepsilon^{(n)},\mathbf{x}^{(n)}}$. If one
is interested in distributed solution schemes, Algorithm \ref{async_best-response_algo}
applied to $\mathcal{G}_{\tau,\varepsilon^{(n)},\mathbf{x}^{(n)}}$
is the natural choice. Note that the convergence conditions of the
algorithm applied to $\mathcal{G}_{\tau,\varepsilon^{(n)},\mathbf{x}^{(n)}}$
as stated in Theorem \ref{Theo-async_best-response_NEP} are always
met, provided $\tau$ is large enough; see, e.g., (\ref{eq:upper_bound_tau_epsilon})
in Lemma \ref{lem: monotone tau epsilon P}.

As in the previous section, note that, unless one has simple ways
to compute the solutions of the games $\mathcal{G}_{\tau,\varepsilon^{(n)},\mathbf{x}^{(n)}}$,
Algorithm \ref{algo2} requires at each step the exact computation
of the solution of the regularized games $\mathcal{G}_{\tau,\varepsilon^{(n)},\mathbf{x}^{(n)}}$
(inner loop), which in principle requires a conceptually infinite
procedure. While in practice this might not be a problem, it still
leaves open the question of whether, paralleling the results of the
previous section, one can develop versions of Algorithm \ref{algo2}
where inexact solutions are used in the Step 2. The answer to this
question is positive, but the corresponding theory is rather complex
and, for sake of simplicity, we prefer to omit it here; the interest
reader can work it out using results in \cite{FacchineiPangScutariLampariello_MP11}.\vspace{-0.2cm}

\subsection{A bird's-eye view}

In the previous three sections we proposed several distributed algorithms
for real player-convex NEPs, which are applicable to different scenarios.
Fig. \ref{fig1_Roadmap} summarizes the results obtained so far, showing
that, in spite of apparent diversities, all the algorithms belong
to a same family. 
\begin{figure}[h]
\vspace{-1cm}
\center\includegraphics[height=12cm]{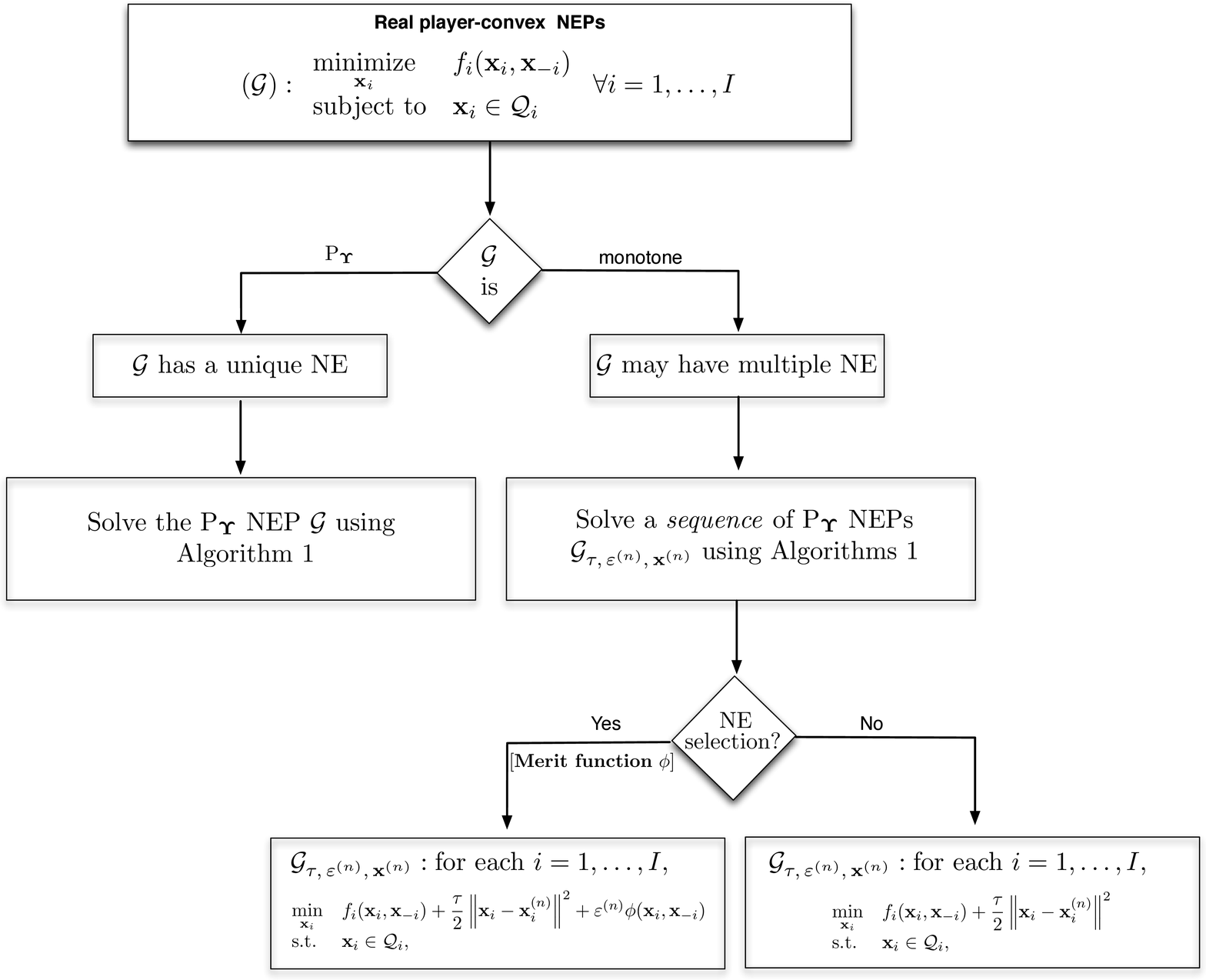}\vspace{-0.4cm}

\caption{{\small{The roadmap of the proposed distributed solution methods for
real player-convex NEPs}}}

{\small{\label{fig1_Roadmap} }} \vspace{-0.2cm}
\end{figure}

Conceptually, what we have proposed is indeed a unified algorithm,
where the users can explicitly choose the degree of desired cooperation
and signaling, converging to solutions having different performance,
namely: i) \emph{any one} NE, when there is no (or very limited) cooperation
among users, and ii) the ``best'' NE (according to an outer merit
function $\phi$), at the cost of more coordination. The choice of
one scheme in favor to the other as well as the merit function $\phi$
will depend then on the trade-off between signaling and performance
that the users are willing to exchange/achieve. The core of the proposed
solution methods can be summarized in the following unified updating
rule: at iteration $n$, the optimal strategy of each user $i$ is
\begin{equation}
\mathbf{x}_{i}^{(n+1)}=\underset{\mathbf{x}_{i}\in\mathcal{Q}_{i}}{\text{argmin}}\left\{ f_{i}\left(\mathbf{x}_{i},\mathbf{x}_{-i}^{(n)}\right)+\pi_{i}^{(\nu_{n})}\left(\mathbf{x}_{i},\,\mathbf{x}_{-i}^{(n)}\right)+\dfrac{\tau}{2}\,\left\Vert \mathbf{x}_{i}-\mathbf{x}_{i}^{(\nu_{n})}\right\Vert ^{2}\right\} \label{eq:unified_algorithm}
\end{equation}
where the first term $f_{i}(\mathbf{x}_{i},\mathbf{x}_{-i}^{(n)})$
is the usual term in an iterative best-response algorithm, the second
term $\pi_{i}^{(\nu_{n})}(\mathbf{x}_{i},\,\mathbf{x}_{-i}^{(n)})$
(whose update is performed at iteration $\nu_{n}$) can be interpreted
as a nonlinear pricing in the objective function of the users, and
the third term $\dfrac{\tau}{2}\,\left\Vert \mathbf{x}_{i}-\mathbf{x}_{i}^{(\nu_{n})}\right\Vert ^{2}$
is a (proximal) regularization. Observe that there are two iteration
indexes: $n$ is the main discrete-time unit, whereas $\nu_{n}$ is
increased every few discrete-time units (e.g., if $\nu_{n}=\left\lfloor n/10\right\rfloor $,
then $\nu_{n}$ is updated every 10 discrete-time units).

The price function $\pi_{i}^{(\nu_{n})}$ can be interpreted as a
measure of the ``altruism/selfishness'' of the users and represents
the trade-off factor between signaling and performance. Indeed, \textcolor{black}{we
may have the following: } 
\begin{itemize}
\item \textcolor{black}{$\pi_{i}^{(\nu_{n})}=0$ (no cooperation): The users
are not willing to cooperate; the best one can get is converge to
}\textcolor{black}{\emph{any one}}\textcolor{black}{{} solution of
the game (i.e., with no control on the quality of the solution); this
is guaranteed }\textcolor{black}{\emph{even in the presence of multiple
equilibria}}\textcolor{black}{{} if the NEP is monotone (}\emph{$P{}_{\boldsymbol{\Upsilon}}$}\textcolor{black}{);} 
\item $\pi_{i}^{(\nu_{n})}\neq0$ (some cooperation): The users may exchange
some signaling in the form of pricing through the function $\pi_{i}^{(\nu_{n})}(\mathbf{x}_{i},\,\mathbf{x}_{-i}^{(n)})=\varepsilon^{(\nu_{n})}\,\phi(\mathbf{x}_{i},\,\mathbf{x}_{-i}^{(n)})$
and converge to the NE\textcolor{black}{{} that minimizes the merit
function $\phi(\mathbf{x})$; convergence is guaranteed if the NEP
is monotone.} \vspace{-0.2cm}
\end{itemize}
\begin{remark}[Role of pricing] \emph{It is important to remark that
the pricing term $\pi_{i}^{(\nu_{n})}$ does not need to be linear;
moreover it has a well understood role in the system optimization.
This is a major departure from current literature that uses }linear\emph{
pricing as a heuristic to improve the performance of a NE in power
control games (see, e.g., \cite{Saraydar-Mandayam-Goodman_TCOM02_pricing}
for scalar power control problems); in these works there is neither
a proof of convergence of the modified game nor a theoretical explanation
of the performance improvement due to pricing.} \emph{As a direct
product of our framework, we obtain instead a clear understanding
of the meaning of the pricing; for example, a linear price in the
form $\varepsilon^{(\nu_{n})}\,\boldsymbol{{\pi}}_{i}^{T}\mathbf{x}_{i}$,
with $\varepsilon^{(\nu_{n})}\rightarrow0$, corresponds to the selection
of a NE that minimizes the linear function $\sum_{i}\boldsymbol{{\pi}}_{i}^{T}\mathbf{x}_{i}$,
resulting likely in better system performance.}\vspace{-0.2cm}\end{remark}

\section{Variational Inequalities and Games in the Complex Domain\label{sec:VI_complex_domain}\vspace{-0.2cm}
}

The results presented so far apply to real NEPs. However, in many
applications, e.g., in digital communications, array processing, and
signal processing, the variables involved in the optimization are
complex. For instance, in the MIMO problems introduced in Sec. \ref{sub:Some-motivating-examples},
players' optimization variables are complex matrices. For these applications
the reformulation of the problem into the real domain is awkward,
very difficult to handle, and generally leads to final conditions
that cannot be easily interpreted in terms of the original complex
setup. Indeed, this ``natural'' approach has long been abandoned
in the signal processing and communication communities, since it has
shown to be inadequate. It seems instead more convenient to work directly
in the complex domain. This requires the use of some sophisticated
tools hinging on involved analytic developments. However, once one
has mastered these tools, the prize is a very smooth and immediate
generalization of \emph{all} the results developed in the previous
section, thus easily providing a whole set of new methods for the
solution of NEPs in the complex domain. 

In order to follow this plan in this section we first recall some
basic results about the so-called Wirtinger calculus, prompted by
the lack of a well-established notation and definitions for $\mathbb{R}$-matrix
derivatives; two good tutorials on the subject are \cite{Are_book_MatrixDiff,Delgado_Cderivative09}.
We then proceed to the development of several new technical tools
that are then applied to the study of VIs and NEPs in the complex
domain. We start introducing in Sec. \ref{sub:The-minimum-principle}
the minimum principle for constrained convex optimization problems
in the domain of complex matrices, generalizing the already known
complex gradient-vanishing conditions obtained in \cite{Are_book_MatrixDiff}
for the \emph{unconstrained} case. As intermediate result, we also
introduce a Taylor expansion of real-valued functions of complex matrices
that is amenable to our MIMO applications. The second important contribution
is given in Sec. \ref{sub:NEP-and-VI_complex_domain}, where, after
introducing the VI problem in the complex domain and the associated
monotonicity and P properties, we provide new matrix conditions for
these properties to hold. These conditions are the natural generalization
of those obtained in Sec. \ref{sec:Nash-Equilibrium-Problems}, provided
that a new definition of Jacobian matrix for complex-valued matrix
functions as well as a tailored concept of positive (semi-)definiteness
are used. Finally, in Sec. \ref{sub:NEP_complex_domain}, we establish
the connection between VIs and NEPs in the complex domain, and discuss
its main implications. \vspace{-0.3cm}

\subsection{$\mathbb{R}$-matrix derivatives \label{sub:Preliminaries}\vspace{-0.2cm}}

In practical applications, we often deal with optimization of real-valued
functions $f:\mathbb{C}\ni z\mapsto f(z)\in\mathbb{R}$ of a complex
variable $z$ that are not differentiable in $\mathbb{C}$ (termed
also $\mathbb{C}$-differentiable or holomorphic).%
\footnote{It is a known fact that nonconstant \emph{real-valued} functions (of
complex variables) are not $\mathbb{C}$-differentiable.%
} However, the same univariate function $f:\mathbb{C}\rightarrow\mathbb{R}$
can also be viewed as a bivariate function of its real and imaginary
components, i.e., $f(z)=g(z_{R},z_{I})$, where $g:\mathbb{R}^{2}\mapsto\mathbb{R}$
is a real-valued function of the real variables $z_{R}\triangleq\text{{Re}}(z)$
and $z_{I}\triangleq\text{{Im}}(z)$. This way, one may be able to
replace the nonexistence of the $\mathbb{C}$-derivative of $f$ with
the existence of the real partial derivatives of $g(z_{R},z_{I})$,
which is actually what one needs to compute a stationary point of
the function. This motivates the introduction of the so-called $\mathbb{R}$-derivative
and conjugate $\mathbb{R}$-derivative of $f:\mathbb{C}\rightarrow\mathbb{R}$
at $z_{0}\in\mathbb{C}$, formally defined as 
\begin{equation}
\frac{{\partial}f}{\partial z}(z_{0})\triangleq\dfrac{{1}}{2}\,\left.\left(\frac{{\partial}f(z)}{\partial z_{R}}-j\frac{{\partial}f(z)}{\partial z_{I}}\right)\right|_{z=z_{0}}\quad\mbox{and}\quad\frac{{\partial}f}{\partial z^{\ast}}(z_{0})\triangleq\dfrac{{1}}{2}\,\left.\left(\frac{{\partial}f(z)}{\partial z_{R}}+j\frac{{\partial}f(z)}{\partial z_{I}}\right)\right|_{z=z_{0}},\label{eq:complex_derivatives-1}
\end{equation}
respectively, where $j=\sqrt{{-1}}$. Note that the derivatives above
must be interpreted formally, because $z$ and its conjugate $z^{\ast}$
in (\ref{eq:complex_derivatives-1}) are treated as they were mutually
independent; the derivatives $\frac{{\partial}f}{\partial z_{R}}$
and $\frac{{\partial}f}{\partial z_{I}}$ represent instead the true
(non-formal) partial derivatives of $f$ viewed as a bivariate function
of $z_{R}$ and $z_{I}$, i.e., $f=\check{f}(z_{R},z_{I})$. When
$\frac{{\partial}f}{\partial z_{R}}$ and $\frac{{\partial}f}{\partial z_{I}}$
exist (and are continuous), implying that (\ref{eq:complex_derivatives-1})
is well-defined, we say that $f$ is $\mathbb{R}$-differentiable
(or continuously $\mathbb{R}$-differentiable); similarly to the real
case, when we say that a function $f:\mathcal{K}\rightarrow\mathbb{R}$
is $\mathbb{R}$-differentiable (or continuously $\mathbb{R}$-differentiable)
on the closed set $\mathcal{K}$, we mean that the function is so
on an open set containing $\mathcal{K}$.

The $\mathbb{R}$-derivatives defined in (\ref{eq:complex_derivatives-1})
for a real-valued function can be naturally extended to \emph{complex-valued}
functions of a complex argument, that is, $f:\mathbb{C}\rightarrow\mathbb{C}$;
formally we still have (\ref{eq:complex_derivatives-1}), but now
$f(z)=\check{f}(z_{R},z_{I})\triangleq\check{f}_{R}(z_{R},z_{I})+j\cdot\check{f}_{I}(z_{R},z_{I})$,
with $\check{f}:\mathbb{R}^{2}\mapsto\mathbb{C}$ and $\check{f}_{R},\check{f}_{I}:\mathbb{R}^{2}\mapsto\mathbb{R}$,
and by ${\partial}f/\partial z_{R}$ we mean ${\partial}f/\partial z_{R}\triangleq{\partial}\check{f}_{R}/{\partial}z_{R}+j\cdot{\partial}\check{f}_{I}/{\partial}z_{R}$
(similarly for ${\partial}f/\partial z_{I}$).

When $f$ is a (complex-valued) scalar function of complex \emph{matrices,}
that is $f:\mathbb{C}^{n\times m}\rightarrow\mathbb{C}$, we have
$n\cdot m$ component-wise $\mathbb{R}$-derivatives $\frac{{\partial}f}{\partial\left(\mathbf{Z}\right)_{ij}}$
and $n\cdot m$ conjugate $\mathbb{R}$-derivatives $\frac{{\partial}f}{\partial\left(\mathbf{Z}^{\ast}\right)_{ij}}$.
The question naturally arises how to order these $n\cdot m$ complex
terms; obviously this can be done in many ways. It is worthwhile noticing
that, even though they all contain the same $n\cdot m$ derivatives,
not all definitions have the same properties; for instance for some
of them a useful chain rule does not exist. Next, we introduce two
definitions, both useful for our derivations and widely used in the
literature \cite{Are_book_MatrixDiff}; in the former definition,
the $n\cdot m$ (conjugate) $\mathbb{R}$-derivatives are displayed
in the same order as $(\mathbf{Z})_{ij}$ and $(\mathbf{Z}^{\ast})_{ij}$
appear in $\mathbf{Z}$ and $\mathbf{Z}^{\ast}$, whereas in the latter
we arrange all the elements in a row vector. Given $f:\mathbb{C}^{n\times m}\rightarrow\mathbb{C}$,
the (matrix) gradient and co(njugate)-gradient of $f$ at $\mathbf{Z}_{0}\in\mathbb{C}^{n\times m}$
are defined as
\begin{equation}
\begin{array}{l}
\,\,\nabla_{\mathbf{Z}}f(\mathbf{Z}_{0})\triangleq\left.\dfrac{{\partial}\, f(\mathbf{Z})}{\partial\mathbf{Z}}\right|_{\mathbf{Z}=\mathbf{Z}_{0}}\quad\mbox{with}\quad\left[\dfrac{{\partial}\, f}{\partial\mathbf{Z}}\right]_{ij}=\dfrac{{\partial}\, f}{\partial\,(\mathbf{Z})_{ij}},\quad\,\,\,\forall i=1,\ldots,n\quad\mbox{and}\quad j=1,\dots,m\bigskip\\
\nabla_{\mathbf{\mathbf{Z}^{\ast}}}f(\mathbf{Z}_{0})\triangleq\left.\dfrac{{\partial}f(\mathbf{Z})}{\partial\mathbf{Z^{\ast}}}\right|_{\mathbf{Z}=\mathbf{Z}_{0}}\quad\mbox{with}\quad\left[\dfrac{{\partial}\, f}{\partial\mathbf{Z}^{\ast}}\right]_{ij}=\dfrac{{\partial}\, f}{\partial\,(\mathbf{Z}^{\ast})_{ij}},\quad\forall i=1,\ldots,n\quad\mbox{and}\quad j=1,\dots,m,
\end{array}\label{eq:complex_matrix_derivative}
\end{equation}
where $\frac{{\partial}f}{\partial(\mathbf{Z})_{ij}}$ and $\frac{{\partial}f}{\partial(\mathbf{Z})_{ij}^{\ast}}$
are the $\mathbb{R}$-derivative and conjugate $\mathbb{R}$-derivative
of the complex-valued function $f$ w.r.t. $\left(\mathbf{Z}\right)_{ij}$
and $\left(\mathbf{Z}^{\ast}\right)_{ij}$, respectively. Note that
$\nabla_{\mathbf{Z}}f(\mathbf{Z}_{0})$ and $\nabla_{\mathbf{Z}^{\ast}}f(\mathbf{Z}_{0})$
are matrices having the same size of $\mathbf{Z}$. Alternatively,
one can arrange the elements $\frac{{\partial}f}{\partial(\mathbf{Z})_{ij}}$
and $\frac{{\partial}f}{\partial(\mathbf{Z}^{\ast})_{ij}}$ in a row
vector, and define $D_{\mathbf{Z}}f(\mathbf{Z})$ and $D_{\mathbf{Z^{\ast}}}f(\mathbf{Z})$
at $\mathbf{Z}_{0}\in\mathbb{C}^{n\times m}$ as 
\begin{equation}
\begin{array}{l}
D_{\mathbf{Z}}f(\mathbf{Z}_{0})\triangleq\left.\dfrac{{\partial}\, f(\mathbf{Z})}{\partial\,\text{{vec}}\left(\mathbf{Z}\right)^{T}}\right|_{\mathbf{Z}=\mathbf{Z}_{0}}=\text{{vec}}\left(\nabla_{\mathbf{Z}}f(\mathbf{Z}_{0})\right)^{T}\,\,\mbox{and\,}\,\, D_{\mathbf{Z}^{\ast}}f(\mathbf{Z}_{0})\triangleq\left.\dfrac{{\partial}\, f(\mathbf{Z})}{\partial\,\text{{vec}}\left(\mathbf{Z}^{\ast}\right)^{T}}\right|_{\mathbf{Z}=\mathbf{Z}_{0}}=\text{{vec}}\left(\nabla_{\mathbf{Z}^{\ast}}f(\mathbf{Z}_{0})\right)^{T}\end{array},\label{eq:Jacobian_scalar_f}
\end{equation}
where $\text{{vec}}\left(\mathbf{A}\right)^{T}$ stands for $\left(\text{{vec}}\left(\mathbf{A}\right)\right)^{T}$.
For (complex-valued) \emph{matrix} functions of complex matrices,
$\mathbf{F}^{\mathbb{C}}:\mathbb{C}^{n\times m}\rightarrow\mathbb{C}^{p\times q}$,
we arrange the $pq\cdot nm$ (conjugate) $\mathbb{R}$-derivatives
in the following $pq\times nm$ matrices 
\begin{equation}
\hspace{-0.2cm}\begin{array}{c}
D_{\mathbf{Z}}\mathbf{F}^{\mathbb{C}}(\mathbf{Z}_{0})\triangleq\left.\dfrac{{\partial}\,\text{{vec}}\left(\mathbf{F}^{\mathbb{C}}(\mathbf{Z})\right)}{\partial\,\text{{vec}}\left(\mathbf{Z}\right)^{T}}\right|_{\mathbf{Z}=\mathbf{Z}_{0}},\,\,\mbox{with}\,\,\left[\dfrac{{\partial}\,\text{{vec}}\left(\mathbf{F}^{\mathbb{C}}\right)}{\partial\,\text{{vec}}\left(\mathbf{Z}\right)^{T}}\right]_{ij}=\dfrac{{\partial}\,\left[\text{{vec}}\left(\mathbf{F}^{\mathbb{C}}\right)\right]_{i}}{\partial\,\left[\text{{vec}}\left(\mathbf{Z}\right)\right]_{j}},\,\,\forall i=1,\ldots,pq\,\,\mbox{and}\,\, j=1,\dots,nm,\bigskip\\
D_{\mathbf{Z}^{\ast}}\mathbf{F}^{\mathbb{C}}(\mathbf{Z}_{0})\triangleq\left.\dfrac{{\partial}\,\text{{vec}}\left(\mathbf{F}^{\mathbb{C}}(\mathbf{Z})\right)}{\partial\,\text{{vec}}\left(\mathbf{Z}^{\ast}\right)^{T}}\right|_{\mathbf{Z}=\mathbf{Z}_{0}},\,\,\mbox{with}\,\,\left[\dfrac{{\partial}\,\text{{vec}}\left(\mathbf{F}^{\mathbb{C}}\right)}{\partial\,\text{{vec}}\left(\mathbf{Z}^{\ast}\right)^{T}}\right]_{ij}=\dfrac{{\partial}\,\left[\text{{vec}}\left(\mathbf{F}^{\mathbb{C}}\right)\right]_{i}}{\partial\,\left[\text{{vec}}\left(\mathbf{Z}^{\ast}\right)\right]_{j}},\,\,\forall i=1,\ldots,pq\,\,\mbox{and}\,\, j=1,\dots,nm.
\end{array}\label{eq:Jacobians_complex_derivatives}
\end{equation}
 The $pq\times nm$ matrices $D_{\mathbf{Z}}\mathbf{F}^{\mathbb{C}}$
and $D_{\mathbf{Z}^{\ast}}\mathbf{F}^{\mathbb{C}}$ are called Jacobian
and conjugate Jacobian of $ $$\mathbf{F}^{\mathbb{C}}$. Note that
when $\mathbf{F}^{\mathbb{C}}$ is a \emph{scalar} function of $\mathbf{Z}$,
i.e., $\mathbf{F}^{\mathbb{C}}(\mathbf{Z})=f(\mathbf{Z})$ with $f:\mathbb{C}^{n\times m}\rightarrow\mathbb{C}$,
definitions (\ref{eq:Jacobians_complex_derivatives}) reduce to (\ref{eq:Jacobian_scalar_f}).
Practical rules to compute $\mathbb{R}$-derivatives and conjugate
$\mathbb{R}$-derivatives introduced above can be found in \cite{Are_book_MatrixDiff}.\vspace{-0.2cm}

\subsection{The minimum principle\label{sub:The-minimum-principle}}

Let $\mathbb{C}^{n\times m}$ be the space of complex $n\times m$
matrices, and let $\mathcal{K}\subseteq\mathbb{C}^{n\times m}$ be
a closed and convex set. We consider the optimization problem 
\begin{equation}
\begin{array}{ll}
\limfunc{minimize}\limits _{\mathbf{Z}} & f(\mathbf{Z})\\[5pt]
\text{subject to} & \mathbf{Z}\in\mathcal{K},
\end{array}\label{eq:Opt_complex}
\end{equation}
where $f:\mathcal{K}\rightarrow\mathbb{R}$ is a real-valued convex
and continuously $\mathbb{R}$-differentiable function on $\mathcal{K}$.
At the basis of the minimum principle there is the first-order Taylor
expansion of $f$ at $\mathbf{Z}_{0}\in\mathcal{K}$ as proved in
Appendix \ref{sec:Proof-of-Lemma_min_principle}: 
\begin{align}
f(\mathbf{\mathbf{Z}_{0}}+\Delta\mathbf{Z})-f(\mathbf{Z}_{0})\simeq & \,\text{2\,\text{{Re}}\ensuremath{\left(\,\text{{tr}}\left(\left(\nabla_{\mathbf{Z}}f(\mathbf{Z}_{0})\right)^{T}\Delta\mathbf{Z}\right)\,\right)}}\triangleq2\,\left\langle \Delta\mathbf{Z},\,\nabla_{\mathbf{Z}^{\ast}}f(\mathbf{Z}_{0})\right\rangle \label{eq:Taylor_coomplex_domain_r2}
\end{align}
where we used $\left(\nabla_{\mathbf{Z}}f\right)^{\ast}=\nabla_{\mathbf{\mathbf{Z}^{\ast}}}f$
since $f$ is real {[}see (\ref{eq:complex_derivatives-1}){]}, and
we introduced the inner product $\left\langle \bullet,\,\bullet\right\rangle :\mathbb{C}^{n\times n}\times\mathbb{C}^{n\times n}\rightarrow\mathbb{R}$,
defined as 
\begin{equation}
\left\langle \mathbf{A},\,\mathbf{B}\right\rangle \triangleq\text{{Re}}\left(\,\text{{tr}}\left(\mathbf{A}^{H}\mathbf{B}\right)\,\right).\label{eq:inner_product_matrix}
\end{equation}
Note that the norm induced by the inner product $\left\langle \bullet,\,\bullet\right\rangle $
is the Frobenius norm, i.e., $\left\langle \mathbf{A},\,\mathbf{A}\right\rangle =\text{Tr}(\mathbf{A}^{H}\mathbf{A})=\left\Vert \mathbf{A}\right\Vert _{F}^{2}.$
Using (\ref{eq:Taylor_coomplex_domain_r2}) we can now introduce the
minimum principle as given next.\vspace{-0.2cm}

\begin{lemma}\label{Lemma_min_principle}Given the convex optimization
problem (\ref{eq:Opt_complex}) in the setting above, $\mathbf{X}\in\mathcal{K}$
is an optimal solution of (\ref{eq:Opt_complex}) if and only if $\mathbf{X}$
satisfies $\left\langle \mathbf{Z}-\mathbf{X},\,\nabla_{\mathbf{Z}^{\ast}}f\mathbf{(X)}\right\rangle \geq0$
for all $\mathbf{Z}\in\mathcal{K}$. \end{lemma}\vspace{-0.3cm}\begin{proof}See
Appendix \ref{sec:Proof-of-Lemma_min_principle}.\end{proof}

It is interesting to observe that if the optimal solution $\mathbf{X}$
is in the interior of $\mathcal{K}$ {[}e.g., the optimization problem
(\ref{eq:Opt_complex}) is unconstrained, implying $\mathcal{K}=\mathbb{C}^{n\times m}$),
then the above optimality conditions reduce to $\nabla_{\mathbf{Z}^{\ast}}f\mathbf{(X)}=\mathbf{0}$,
or equivalently $\nabla_{\mathbf{Z}}f\mathbf{(X)}=\mathbf{0}$, which
are the well-established complex gradient-vanishing conditions obtained
in \cite{Are_book_MatrixDiff} for the \emph{unconstrained} minimization.
We conclude this section with an example of application of the minimum
principle, which is instrumental for the analysis in Sec. \ref{sec:Applications}.\vspace{-0.3cm}

\begin{example}[An application of the minimum principle]\label{Example_MIMO_complex_min_princ}
\emph{Consider the following single-user rate maximization problem
\begin{equation}
\begin{array}{ll}
\limfunc{maximize}\limits _{\mathbf{Z}} & f(\mathbf{Z})\triangleq\log\det\left(\mathbf{R}_{n}+\mathbf{H}\mathbf{Z}\mathbf{H}^{H}\right)\\[5pt]
\text{subject to} & \mathbf{Z}\in\mathcal{K},
\end{array}\label{eq:rate_max}
\end{equation}
where $\mathbf{R}_{n}\in\mathbb{C}^{m\times m}$ is a positive definite
matrix, $\mathbf{H}\in\mathbb{C}^{m\times n}$$ $, and $\mathcal{K}$
is any convex and compact subset of the $n\times n$ complex positive
semidefinite matrices (assumed to be nonempty). Note that $f(\mathbf{Z})$
is a concave (real-valued) function on the feasible set $\mathcal{K}$
but is not real if defined on $\mathbb{C}^{n\times n}$. Since we
are interested in minimizing $f$ over $\mathcal{K}$, in order to
apply the minimum principle, one approach we can follow is to consider
without loss of generality the modified function $\tilde{f}:\widetilde{\mathcal{K}}\rightarrow\mathbb{R}$,
defined as $\tilde{f}(\mathbf{Z})\triangleq2\,\text{{Re}}\left(f(\mathbf{Z})\right)$,
where $\mathcal{K}\subset\widetilde{\mathcal{K}}\subseteq\mathbb{C}^{n\times n}$
is any open set over which $f(\mathbf{Z})$ is well defined (it is
sufficient that} $\det\left(\mathbf{R}_{n}+\mathbf{H}\mathbf{Z}\mathbf{H}^{H}\right)\neq0$);\emph{
indeed, $\tilde{f}$ coincides with $f$ over $\mathcal{K}$, but
it is real everywhere (in its domain). Moreover, $\tilde{f}$ is $\mathbb{R}$-differentiable
on $\widetilde{\mathcal{K}}$.}%
\footnote{The introduction of the auxiliary function $\tilde{{f}}$ might appear
an unnecessary complication, which needs clarification. The original
function $f$ is defined over a (sub)set of positive semidefinite
matrices. The theory of matrix derivatives introduced in this paper
cannot be applied however to functions of matrices having a structure.
The function $\tilde{{f}}$ is introduced just to overcome this issue;
indeed, it is defined over an open set of \emph{unpatterned} matrices
while being equal to $f$ over the set of interest. An alternative
approach would be working directly with the original $f$ and using
the so-called complex (patterned) generalized derivatives \cite{Are_book_MatrixDiff}.
However, up to date there are no rules to compute matrix derivatives
over \emph{arbitrary} manifolds, which strongly limits the applicability
of this methodology in practice. This motivates the former approach. %
}\emph{ The conjugate (matrix) $\mathbb{R}$-derivative of $\tilde{f}$
at $\mathbf{Z}_{0}\in\widetilde{\mathcal{K}}$ is (see Appendix \ref{App: Example 24-Intermediate})
\begin{equation}
\nabla_{\mathbf{\mathbf{Z}^{\ast}}}\tilde{{f}}(\mathbf{Z}_{0})=\mathbf{H}^{H}(\mathbf{R}_{n}+\mathbf{H}\mathbf{Z}_{0}^{H}\mathbf{H}^{H})^{-1}\mathbf{H}.\label{eq:nabl_f_log}
\end{equation}
 Introducing $\mathbf{G}:\mathcal{K}\rightarrow\mathbb{C}^{n\times n}$,
defined as $\mathbf{G}=\mathbf{G}(\mathbf{Z})\triangleq-\nabla_{\mathbf{\mathbf{Z}^{\ast}}}\tilde{f}(\mathbf{Z})=-\mathbf{H}^{H}(\mathbf{R}_{n}+\mathbf{H}\mathbf{Z}\mathbf{H}^{H})^{-1}\mathbf{H}$
(note that $\mathbf{Z}\in\mathcal{K}$ and thus $\mathbf{Z}=\mathbf{Z}^{H}$),
and invoking Lemma \ref{Lemma_min_principle}, the optimization problem
(\ref{eq:rate_max}) is then equivalent to the minimum principle:
find a $\mathbf{Z}\in\mathcal{K}$ such that $\left\langle \mathbf{Y}-\mathbf{Z},\,\mathbf{G}(\mathbf{Z})\right\rangle \geq0$,
for all $\mathbf{Y}\in\mathcal{K}$. \hfill{}$\square$}\end{example}\vspace{-0.5cm}

\subsection{The VI problem in the complex domain\label{sub:NEP-and-VI_complex_domain}}

With the developments of the previous section at hand, we can now
introduce the definition of the VI problem in the domain of complex
matrices, termed the \emph{complex VI problem}. Similarly to the real
case (cf. Appendix \ref{sec:A-Theory-of_partitioned_VI}), one can
think of the VI problem as the generalization of the minimum principle
(cf. Lemma \ref{Lemma_min_principle}), where the co-gradient $\nabla_{\mathbf{Z}^{\ast}}f$
is replaced with a complex-valued matrix mapping. The formal definition
is given next. \vspace{-0.5cm}

\begin{definition}\label{Def_VI_complex}Given a convex and closed
set $\mathcal{K}\subseteq\mathbb{C}^{n\times m}$ and a complex-valued
matrix function $\mathbf{F}^{\mathbb{C}}(\mathbf{Z}):\mathcal{K}\ni\mathbf{Z}\rightarrow\mathbb{C}^{n\times m}$,
the \emph{complex VI} problem, denoted by $\mathcal{\text{\emph{{VI}}}}\left(\mathcal{K},\mathbf{F}^{\mathbb{C}}\right)$,
consists in finding a point $\mathbf{Z}\in\mathcal{K}$ such that
$\left\langle \mathbf{Y}-\mathbf{Z},\mathbf{\mathbf{F}^{\mathbb{C}}}(\mathbf{Z})\right\rangle \geq0$
for all $\mathbf{Y}\in\mathcal{K}$. The solution set of the $\text{\emph{{VI}}}\left(\mathcal{K},\mathbf{F}^{\mathbb{C}}\right)$
is denoted by $\sol\left(\mathcal{K},\mathbf{F}^{\mathbb{C}}\right).$\end{definition}

When $\mathcal{K}$ has a Cartesian structure, i.e., $\mathcal{K}\triangleq\prod_{i=1}^{I}\mathcal{K}_{i}$
with each $\mathcal{K}_{i}\subseteq\mathbb{\mathbb{C}}^{n_{i}\times m_{i}}$,
we write $\mathbf{F}^{\mathbb{C}}(\mathbf{Z})\triangleq\left(\mathbf{F}_{i}^{\mathbb{C}}(\mathbf{Z})\right)_{i=1}^{I}$
and $\mathbf{Z}\triangleq(\mathbf{Z}_{i})_{i=1}^{I}$, with $\mathbf{F}_{i}^{\mathbb{C}}(\mathbf{Z}):\mathcal{K}\rightarrow\mathbb{C}^{n_{i}\times m_{i}}$
and $\mathbf{Z}_{i}\in\mathbb{C}^{n_{i}\times m_{i}}$. In such a
case, with a slight abuse of notation, we will still use for the partitioned
$\text{{VI}}\left(\mathcal{K},\mathbf{\mathbf{F}^{\mathbb{C}}}\right)$
the compact notation $\left\langle \mathbf{Y}-\mathbf{Z},\mathbf{\mathbf{F}^{\mathbb{C}}}(\mathbf{Z})\right\rangle $,
by meaning $\sum_{i=1}^{I}\left\langle \mathbf{Y}_{i}-\mathbf{Z}_{i},\mathbf{\mathbf{F}^{\mathbb{C}}}_{i}(\mathbf{Z})\right\rangle $.
Moreover, the definitions of $D_{\mathbf{Z}}\mathbf{F}^{\mathbb{C}}(\mathbf{Z})$
and $D_{\mathbf{Z}^{\ast}}\mathbf{F}^{\mathbb{C}}(\mathbf{Z})$ as
given in (\ref{eq:Jacobians_complex_derivatives}) depend in principle
on the ordering according to which the components of $\mathbf{F}^{\mathbb{C}}(\mathbf{Z})$
and $\mathbf{Z}$ are grouped in the vec operator. For our purposes,
the following ordering is the most convenient, which is tacitly assumed
throughout the paper: $\text{{vec}}\left((\mathbf{F}_{i}^{\mathbb{C}}(\mathbf{Z}))_{i=1}^{I}\right)\triangleq\text{\,}\left[\mbox{{vec}}(\mathbf{F}_{1}^{\mathbb{C}}(\mathbf{Z}))^{T},\ldots,\mbox{{vec}}(\mathbf{F}_{I}^{\mathbb{C}}(\mathbf{Z}))^{T}\right]^{T}$
and $\text{{vec}}\left((\mathbf{Z}_{i})_{i=1}^{I}\right)\triangleq\left[\mbox{{vec}}(\mathbf{Z}_{1})^{T},\ldots,\mbox{{vec}}(\mathbf{Z}_{I})^{T}\right]^{T}$.
\vspace{-0.5cm}

\subsection{\noindent Monotonicity and P properties of $\text{{VI}}\left(\mathcal{K},\mathbf{F}^{\mathbb{C}}\right)$\vspace{-0.1cm}}

\noindent We can now readily extend the definitions of monotonicity
and P property for real-valued vector functions (see Definition \ref{Def_monotonicity}
in Appendix A) to complex-value matrix maps $\mathbf{F}^{\mathbb{C}}$;
the aforementioned definitions are in fact formally the same, with
the only difference that the scalar product and the Euclidean norm
are replaced with the inner product $\left\langle \bullet,\,\bullet\right\rangle $
defined in (\ref{eq:inner_product_matrix}) and the Frobenius norm,
respectively. The non-trivial task is instead to derive easy conditions
to check guaranteeing these properties. These conditions are indeed
instrumental to study convergence of algorithms for complex NEPs.
The interesting result we prove next is that we can obtain necessary
and sufficient conditions for a continuously ($\mathbb{R}$-)differentiable
$\mathbf{F}^{\mathbb{C}}(\mathbf{Z})$ to be a monotone function or
(sufficient conditions to be) a P function on $\mathcal{K}$ that
are formally equivalent to those obtained for real-valued vector functions
$\mathbf{F}(\mathbf{x})$ {[}cf. (\ref{monotonicity-convexity-connection})
and Proposition \ref{monotonicity}{]}, provided that we introduce
a novel definition of Jacobian matrix suitable for complex-valued
functions of complex variables; such a Jacobian will contain both
$\mathbb{R}$-derivatives and conjugate $\mathbb{R}$-derivatives
of $\mathbf{F}^{\mathbb{C}}(\mathbf{Z})$.

Given the complex $\text{{VI}}\left(\mathcal{K},\mathbf{\mathbf{F}^{\mathbb{C}}}\right)$,
suppose that $\mathbf{F}^{\mathbb{C}}(\mathbf{Z})$ is a continuously
($\mathbb{R}$-)differentiable matrix function on $\mathcal{K}$.
Then, the $nm\times nm$ Jacobian matrices $D_{\mathbf{Z}}\mathbf{F}^{\mathbb{C}}(\mathbf{Z})$
and $D_{\mathbf{Z}^{\ast}}\mathbf{F}^{\mathbb{C}}(\mathbf{Z})$ in
(\ref{eq:Jacobians_complex_derivatives}) are well-defined at $\mathbf{Z}\in\mathcal{K}$.
Let us introduce the $2nm\times2nm$ matrix $\mathbf{JF}^{\mathbb{C}}(\mathbf{Z})$,
defined as 
\begin{equation}
\mathbf{JF}^{\mathbb{C}}(\mathbf{Z})\triangleq\dfrac{{1}}{2}\,\left[\begin{array}{ll}
D_{\mathbf{Z}}\mathbf{F}^{\mathbb{C}}(\mathbf{Z}) & D_{\mathbf{Z}^{\ast}}\mathbf{F}^{\mathbb{C}}(\mathbf{Z})\\
D_{\mathbf{Z}}\left(\mathbf{F}^{\mathbb{C}}(\mathbf{Z})^{\ast}\right) & D_{\mathbf{Z}^{\ast}}\left(\mathbf{F}^{\mathbb{C}}(\mathbf{Z})^{\ast}\right)
\end{array}\right],\label{eq:Augmented_Jacobian}
\end{equation}
which we call ``augmented Jacobian'' for obvious reasons. For notational
simplicity, in the sequel we will write $D_{\mathbf{Z}}\mathbf{F}^{\mathbb{C}}(\mathbf{Z})^{\ast}$
and $D_{\mathbf{Z}^{\ast}}\mathbf{F}^{\mathbb{C}}(\mathbf{Z})^{\ast}$
for $D_{\mathbf{Z}}\left(\mathbf{F}^{\mathbb{C}}(\mathbf{Z})^{\ast}\right)$
and $D_{\mathbf{Z}^{\ast}}\left(\mathbf{F}^{\mathbb{C}}(\mathbf{Z})^{\ast}\right)$,
respectively. Note that the following relationships hold between the
blocks of $\mathbf{JF}^{\mathbb{C}}(\mathbf{Z})$: $\left(D_{\mathbf{Z}}\mathbf{F}^{\mathbb{C}}(\mathbf{Z})\right)^{\ast}=D_{\mathbf{Z}^{\ast}}\mathbf{F}^{\mathbb{C}}(\mathbf{Z})^{\ast}$
and $\left(D_{\mathbf{Z}^{\ast}}\mathbf{F}^{\mathbb{C}}(\mathbf{Z})\right)^{\ast}=D_{\mathbf{Z}}\mathbf{F}^{\mathbb{C}}(\mathbf{Z})^{\ast}$.
Finally, under the assumption that $\mathbf{F}^{\mathbb{C}}(\mathbf{Z})$
and $\mathcal{K}$ have a partitioned structure and $\mathbf{F}^{\mathbb{C}}(\mathbf{Z})$
has bounded ($\mathbb{R}$)-derivatives on $\mathcal{K},$ let us
introduce the ``condensed'' $I\times I$ matrix $\boldsymbol{{\Upsilon}}_{\mathbf{F}^{\mathbb{C}}}$
given by 
\begin{equation}
\left[\boldsymbol{{\Upsilon}}_{\mathbf{F}^{\mathbb{C}}}\right]_{ij}\,\triangleq\,\left\{ \begin{array}{lcl}
\kappa_{i}^{\min}, &  & \mbox{if }i=j,\\
-\xi_{ij}^{\max}, &  & \mbox{otherwise},
\end{array}\right.\label{eq:Upsilon_matrix_complexVI}
\end{equation}
with 
\begin{equation}
\kappa_{i}^{\min}\,\triangleq\,\inf_{\mathbf{Z}\in\mathcal{K}}\,\lambda_{\text{{least}}}\left(\mathbf{A}_{i}^{H}\,\mathbf{J}_{i}\mathbf{F}_{i}^{\mathbb{C}}(\mathbf{Z})\,\mathbf{A}_{i}\right)\,\quad\mbox{and}\quad\xi_{ij}^{\max}\,\triangleq\,\sup_{\mathbf{Z}\in\mathcal{K}}\,\left\Vert \mathbf{A}_{i}^{H}\,\mathbf{J}_{j}\mathbf{F}_{i}^{\mathbb{C}}(\mathbf{Z})\,\mathbf{A}_{j}\right\Vert _{F},\label{eq:def_alpha_and_beta_Jac_complex_VI}
\end{equation}
where $\mathbf{J}_{i}\mathbf{F}_{i}^{\mathbb{C}}(\mathbf{Z})$ and
$\mathbf{J}_{j}\mathbf{F}_{i}^{\mathbb{C}}(\mathbf{Z})$ represent
the augmented Jacobians of $\mathbf{F}_{i}^{\mathbb{C}}(\mathbf{Z})$
as defined in (\ref{eq:Augmented_Jacobian}), whose $\mathbb{R}$-derivatives
are taken with respect to $\mathbf{Z}_{i}$ and $\mathbf{Z}_{j}$
(and their conjugates), respectively; $\mathbf{A}_{i}\in\mathbb{C}^{2n_{i}m_{i}\times2n_{i}m_{i}}$
are nonsingular arbitrary matrices; and $\left\Vert \mathbf{A}\right\Vert _{F}$
denotes the Frobenius norm of $\mathbf{A}$. As we show shortly, $\mathbf{JF}^{\mathbb{C}}(\mathbf{Z})$
and $\boldsymbol{{\Upsilon}}_{\mathbf{F}^{\mathbb{C}}}$ play for
complex VIs the same role as $\mathbf{J}\mathbf{F}$ and $\boldsymbol{{\Upsilon}}_{\mathbf{F}}$
introduced in Sec. \ref{sub:Existence-and-uniqueness_NE} for real
VIs. 

Before stating the main results (Propositions \ref{monotonicity_complexVI}
and \ref{VI_monotonicity_closed_sets}), we need to introduce a novel
relaxed definition of (uniformly) positive (semi-)definiteness for
matrices in the form (\ref{eq:Augmented_Jacobian}), which takes explicitly
into account the special structure of those matrices. Instead of checking
the sign of the quadratic form $\mathbf{y}^{H}\mathbf{JF}^{\mathbb{C}}(\mathbf{Z})\mathbf{y}$
for arbitrary $\mathbf{y}\in\mathbb{C}^{2nm}$, it turns out that
one can restrict the check to structured vectors in the form $\mathbf{y}=[\mathbf{y}_{1},\mathbf{y}_{1}^{\ast}]$
for all $\mathbf{y}_{1}\in\mathbb{C}^{nm}$, which is actually the
size of the vector space where $\mathbf{Z}$ lies. This motivates
the following definition of ``augmented'' (uniformly) positive (semi-)definiteness
for matrices in the form of (\ref{eq:Augmented_Jacobian}).\vspace{-0.2cm}

\begin{definition}\label{The-augmented-Jacobian}The augmented Jacobian
$\mathbf{JF}^{\mathbb{C}}(\mathbf{Z})$ is said to be: \vspace{-0.1cm}
\begin{description}
\item [{i)}] \emph{augmented }positive semidefinite on $\mathcal{K}$ if
for all $\mathbf{Y}\in\mathbb{C}^{n\times m}$ and $\mathbf{Z}\in\mathcal{K}$,
\begin{equation}
\text{\emph{{vec}}}([\mathbf{\mathbf{Y},\mathbf{Y}^{\ast}}])^{H}\mathbf{JF}^{\mathbb{C}}(\mathbf{Z})\,\emph{{vec}}([\mathbf{\mathbf{Y},\mathbf{Y}^{\ast}}])\geq0;\vspace{-0.2cm}\label{eq:aug_psd}
\end{equation}

\item [{ii)}] \emph{augmented }positive definite on $\mathcal{K}$ if for
all $\mathbf{0}\neq\mathbf{Y}\in\mathbb{C}^{n\times m}$ and $\mathbf{Z}\in\mathcal{K}$,
the inequality in (\ref{eq:aug_psd}) is strict; \vspace{-0.2cm}
\item [{iii)}] uniformly \emph{augmented }positive definite on $\mathcal{K}$
with constant $c>0$ if for all $\mathbf{Y}\in\mathbb{C}^{n\times m}$
and $\mathbf{Z}\in\mathcal{K}$, there exists a positive constant
$c>0$ such that 
\begin{equation}
\emph{{vec}}([\mathbf{\mathbf{Y},\mathbf{Y}^{\ast}}])^{H}\mathbf{JF}^{\mathbb{C}}(\mathbf{Z})\,\text{\emph{{vec}}}([\mathbf{\mathbf{Y},\mathbf{Y}^{\ast}}])\geq c\,\left\Vert \mathbf{Y}\right\Vert _{F}^{2}.\vspace{-0.2cm}\label{eq:aug_unif_pd}
\end{equation}

\end{description}
For i), ii), and iii) we will write $\mathbf{JF}^{\mathbb{C}}(\mathbf{Z})\overset{\mathcal{A}}{\succeq}\mathbf{0}$,
$\mathbf{JF}^{\mathbb{C}}(\mathbf{Z})\overset{\mathcal{A}}{\succ}\mathbf{0}$,
and $\mathbf{JF}^{\mathbb{C}}(\mathbf{Z})-c\mathbf{I}\overset{\mathcal{A}}{\succeq}\mathbf{0}$,
respectively. \end{definition} \vspace{-0.2cm}

Note that $\mathbf{JF}^{\mathbb{C}}(\mathbf{Z})$ is not Hermitian;
which implies that $\text{{vec}}(\mathbf{W})^{H}\mathbf{JF}^{\mathbb{C}}(\mathbf{Z})\text{{vec}}(\mathbf{W})$
is generally not a real number for arbitrary $\mathbf{W}\in\mathbb{C}^{n\times2m}$.
However, because of the structure of $\mathbf{JF}^{\mathbb{C}}(\mathbf{Z})$
and $\text{{vec}}([\mathbf{\mathbf{Y},\mathbf{Y}^{\ast}}])$, with
$\mathbf{Y}\in\mathbb{C}^{n\times m}$, the quadratic form $\text{{vec}}([\mathbf{\mathbf{Y},\mathbf{Y}^{\ast}}])^{H}\mathbf{JF}^{\mathbb{C}}(\mathbf{Z})\,\text{{vec}}([\mathbf{\mathbf{Y},\mathbf{Y}^{\ast}}])$
introduced in the proposition is always real. Note also that if $\mathbf{JF}^{\mathbb{C}}(\mathbf{Z})$
is positive semidefinite, positive definite, or uniformly positive
definite on $\mathcal{K}$ (and thus Hermitian), then it is also augmented
positive semidefinite, positive definite, or uniformly positive definite,
respectively; but the converse in general is not true (because $\mathbf{JF}^{\mathbb{C}}(\mathbf{Z})$
is not Hermitian). Using Definition \ref{The-augmented-Jacobian},
we can now establish the connection with the monotonicity properties
of $\mathbf{F}^{\mathbb{C}}$.\vspace{-0.1cm}

\begin{proposition} \label{monotonicity_complexVI} Let $\mathbf{F}^{\mathbb{C}}:\mathcal{K}\mathcal{\rightarrow\mathbb{C}}^{n\times m}$
be \emph{(}$\mathbb{R}$-\emph{)}continuously differentiable on the
convex set $\mathcal{K}$. Suppose that $\mathcal{K}$ has nonempty
interior. The following statements hold: \vspace{-0.2cm}
\begin{description}
\item [{{{\rm}(a)}}] $\mathbf{F}^{\mathbb{C}}$ is \emph{monotone}
on $\mathcal{K}$ if and only if \textbf{\emph{$\mathbf{JF}^{\mathbb{C}}(\mathbf{Z})$
}}is augmented positive semidefinite on $\mathcal{K}$; \vspace{-0.2cm}
\item [{{{\rm}(b)}}] If\textbf{\emph{ $\mathbf{JF}^{\mathbb{C}}(\mathbf{Z})$}}\emph{
}is augmented positive definite on $\mathcal{K}$,\emph{ }then\emph{
}$\mathbf{F}^{\mathbb{C}}$ is \emph{strictly monotone} on $\mathcal{K}$;
\vspace{-0.2cm}
\item [{{{\rm}(c)}}] $\mathbf{F}^{\mathbb{C}}$ is \emph{strongly
monotone} on $\mathcal{K}$\emph{ }with constant $c_{\text{{sm}}}>0$
if and only if $\mathbf{JF}^{\mathbb{C}}(\mathbf{Z})$ is uniformly
augmented positive definite on $\mathcal{K}$ with constant $c_{\text{{sm}}}/2$.
\vspace{-0.2cm}
\end{description}
If we assume a Cartesian structure, i.e. $\mathbf{F}^{\mathbb{C}}=(\mathbf{F}_{i}^{\mathbb{C}})_{i=1}^{I}$
and $\mathcal{K}=\prod_{i=1}^{I}\mathcal{K}_{i}$, and bounded \emph{(}$\mathbb{R}$\emph{)}-derivatives
of $\mathbf{F}^{\mathbb{C}}$ on $\mathcal{K},$ then: \vspace{-0.2cm}
\begin{description}
\item [{{(d)}}] If $\boldsymbol{{\Upsilon}}_{\mathbf{F}^{\mathbb{C}}}$
is positive semidefinite/P$_{0}$-matrix, then $\mathbf{F}^{\mathbb{C}}$
is a monotone/P$_{0}$ function on $\mathcal{K}$; \vspace{-0.2cm}
\item [{{(e)}}] If $\boldsymbol{{\Upsilon}}_{\mathbf{F}^{\mathbb{C}}}$
is a P-matrix, then $\mathbf{F}^{\mathbb{C}}$ is a \emph{uniformly
P-function} on $\mathcal{K}$. \vspace{-0.2cm}
\end{description}
\end{proposition}\vspace{-0.1cm}\begin{proof} See Appendix \ref{sub:Proof-of-Proposition_monotonicity_complexVI}.\end{proof}

The above proposition is the generalization of (\ref{monotonicity-convexity-connection})
and Proposition \ref{monotonicity} to complex VIs. Note that, if
the set $\mathcal{K}$ has empty interior, necessary conditions in
(a) and (c) generally do not hold, whereas sufficient conditions in
(a)-(c) may be too restrictive. Since some of the optimization problems
of our interest have feasible sets that fall into this class \textcolor{black}{{[}e.g.,
think of the set of Hermitian matrices{]}}, it is worth extending
Proposition \ref{monotonicity_complexVI} to sets with empty interior.
The next result is valid for arbitrary (nonempty) convex sets. \vspace{-0.2cm}

\begin{proposition} \label{VI_monotonicity_closed_sets} Consider
the setting of Proposition \ref{monotonicity_complexVI}, but with
$\mathcal{K}$ being any nonempty convex subset of $\mathbb{C}^{n\times m}$.
Let $\mathcal{S}_{\mathcal{K}}$ be the subspace that is parallel
to the affine hull of $\mathcal{K}$.%
\footnote{We recall that, given a subset $\mathcal{K}$ of $\mathbb{C}^{n\times m}$,
the affine hull of $\mathcal{K}$, denoted by $\text{{Aff}}(\mathcal{K})$,
is the set of all affine combinations of elements in $\mathcal{K}$,
that is $\text{{Aff}}(\mathcal{K})\triangleq\left\{ \mathbf{Y}\in\mathbb{C}^{n\times m}\,:\,\mathbf{Y}=\sum_{i=1}^{k}\alpha_{i}\mathbf{X}_{i},\,\, k>0,\,\mathbf{X}_{i}\in\mathcal{K},\,\alpha_{i}\in\mathbb{R},\,\sum_{i=1}^{k}\alpha_{i}=1\right\} $. %
} The following statements hold: \vspace{-0.2cm}
\begin{description}
\item [{{{\rm}(a)}}] $\mathbf{F}^{\mathbb{C}}$ is \emph{monotone}
on $\mathcal{K}$ if and only if for all $\mathbf{Y}\in\mathbb{C}^{n\times m}$
such that $\mathbf{Y}\in\mathcal{S}_{\mathcal{K}}$ and $\mathbf{Z}\in\mathcal{K}$,
it holds $\text{\emph{{vec}}}\left([\mathbf{Y},\mathbf{Y}^{\ast}]\right)^{H}\mathbf{JF}^{\mathbb{C}}(\mathbf{Z})\,\text{\emph{{vec}}}\left([\mathbf{Y},\mathbf{Y}^{\ast}]\right)\geq0$;\vspace{-0.2cm}
\item [{{{\rm}(b)}}] If for all $\mathbf{0}\neq\mathbf{Y}\in\mathbb{C}^{n\times m}$
such that $\mathbf{Y}\in\mathcal{S}_{\mathcal{K}}$ and $\mathbf{Z}\in\mathcal{K}$,
it holds $\text{\emph{{vec}}}\left([\mathbf{Y},\mathbf{Y}^{\ast}]\right)^{H}\mathbf{JF}^{\mathbb{C}}(\mathbf{Z})\,\text{\emph{{vec}}}\left([\mathbf{Y},\mathbf{Y}^{\ast}]\right)>0$,
then\emph{ }$\mathbf{F}^{\mathbb{C}}$ is \emph{strictly monotone}
on $\mathcal{K}$; \vspace{-0.2cm}
\item [{{{\rm}(c)}}] $\mathbf{F}^{\mathbb{C}}$ is \emph{strongly
monotone} on $\mathcal{K}$\emph{ }with constant $c_{\text{{sm}}}>0$
if and only for all $\mathbf{Y}\in\mathbb{C}^{n\times m}$ such that
$\mathbf{Y}\in\mathcal{S}_{\mathcal{K}}$ and $\mathbf{Z}\in\mathcal{K}$,
it holds $\text{\emph{{vec}}}\left([\mathbf{Y},\mathbf{Y}^{\ast}]\right)^{H}\mathbf{JF}^{\mathbb{C}}(\mathbf{Z})\,\text{\emph{{vec}}}\left([\mathbf{Y},\mathbf{Y}^{\ast}]\right)\geq$\textbf{\emph{$\,(c_{\text{{sm}}}/2)\,\left\Vert \mathbf{Y}\right\Vert _{F}^{2}$}}.
\vspace{-0.2cm}
\end{description}
If we assume a Cartesian product structure, i.e. $\mathbf{F}^{\mathbb{C}}=(\mathbf{F}_{i}^{\mathbb{C}})_{i=1}^{I}$
and $\mathcal{K}=\prod_{i=1}^{I}\mathcal{K}_{i}$, and bounded \emph{(}$\mathbb{R}$\emph{)}-derivatives
of $\mathbf{F}^{\mathbb{C}}(\mathbf{Z})$ on $\mathcal{K},$ then
statements (d) and (e) of Proposition \ref{monotonicity_complexVI}
hold.

\end{proposition}\vspace{-0.2cm}\begin{proof} See Appendix \ref{sub:Proof-of-Proposition_monotonicity_complexVI}.\end{proof}\smallskip{}

A set $\mathcal{K}$ of special interest for our applications (cf.
Sec. \ref{sub:Some-motivating-examples}) is the set of complex $n\times n$
positive semidefinite matrices (and thus Hermitian). This set has
empty interior, implying that one needs to use Proposition \ref{VI_monotonicity_closed_sets}.
It is not difficult to see that the affine hull of such a $\mathcal{K}$
is the set of Hermitian matrices, which is already a subspace. Therefore,
when Proposition \ref{VI_monotonicity_closed_sets} applies to such
a $\mathcal{K}$, the matrices $\mathbf{Y}$ are restricted to the
set of Hermitian matrices. It is worth observing that, when $\mathcal{K}$
has nonempty interior, Proposition \ref{VI_monotonicity_closed_sets}
reduces to Proposition \ref{monotonicity_complexVI}; indeed, we have
$\text{{Aff}}(\mathcal{K})=\mathbb{C}^{n\times m}$, and thus $S_{\mathcal{K}}=\mathbb{C}^{n\times m}$.

Using Proposition \ref{monotonicity_complexVI} (or Proposition \ref{VI_monotonicity_closed_sets})
and building on the structure of $\mathbf{JF}^{\mathbb{C}}$ one can
obtain sufficient conditions for $\mathbf{JF}^{\mathbb{C}}$ to be
augmented positive (semi-)definite or uniformly positive semidefinite,
similarly to what we have done in Sec. \ref{sub:Existence-and-uniqueness_NE}
for real valued vector functions $\mathbf{F}$; one can then extend
the solution analysis and methods developed for the $\text{{VI}}(\mathcal{Q},\mathbf{F})$
to the complex $\text{{VI}}\left(\mathcal{K},\mathbf{\mathbf{F}^{\mathbb{C}}}\right)$;
because of the space limitation, we leave these tasks to the reader.
In Sec. \ref{sec:Applications}, we will show an instance of these
conditions when specialized to the MIMO games along with their physical
interpretations.

We conclude this section by applying Proposition \ref{monotonicity_complexVI}
(or Proposition \ref{VI_monotonicity_closed_sets}) to the conjugate
gradient of real-valued functions of complex variables {[}cf. (\ref{eq:Opt_complex}){]}.
The result is a set of novel necessary and sufficient conditions for
a (continuously $\mathbb{R}$-differentiable) real-valued function
of complex variables to be (strictly) convex or strongly convex, in
terms of $\mathbb{R}$-derivatives. This provides an easy way to check
convexity directly in the complex domain. In order to apply Propositions
\ref{monotonicity_complexVI} or \ref{VI_monotonicity_closed_sets},
we need the following intermediate result, which can be proved using
the Taylor expansion (\ref{eq:Taylor_coomplex_domain_r2}) and similar
approach used in the real case. Given a continuously $\mathbb{R}$-differentiable
real-valued function $f:\mathbb{C}^{n\times m}\rightarrow\mathbb{R}$
, $f$ is convex, strictly convex, or strongly convex on $\mathcal{K}$
if and only if its conjugate gradient $\nabla_{\mathbf{Z}^{\ast}}f$
is monotone, strictly monotone, or strongly monotone on $\mathcal{K}$,
respectively. Using Proposition \ref{monotonicity_complexVI} (or
Proposition \ref{VI_monotonicity_closed_sets}), the convexity properties
of $f(\mathbf{Z})$ can be then restated in terms of properties of
the augmented Jacobian matrix $\mathbf{JF}^{\mathbb{C}}(\mathbf{Z})$
of $\mathbf{\mathbf{F}^{\mathbb{C}}}(\mathbf{Z})$, with $\mathbf{\mathbf{F}^{\mathbb{C}}}(\mathbf{Z})=\nabla_{\mathbf{Z}^{\ast}}f(\mathbf{Z})$,
which we term \emph{augmented Hessian} of $f$, $\mathcal{H}_{\mathbf{Z}\mathbf{Z}^{\ast}}f(\mathbf{Z})$,
given by {[}cf. (\ref{eq:Augmented_Jacobian}){]}: 
\begin{equation}
\mathcal{H}_{\mathbf{Z}\mathbf{Z}^{\ast}}f(\mathbf{Z})\triangleq\dfrac{{1}}{2}\,\left[\begin{array}{cc}
D_{\mathbf{Z}}\left(\nabla_{\mathbf{Z}^{\ast}}f({\mathbf{Z}})\right) & D_{\mathbf{Z}^{\ast}}\left(\nabla_{\mathbf{Z}^{\ast}}f({\mathbf{Z}})\right)\\
D_{\mathbf{Z}}\left(\left(\nabla_{\mathbf{Z}^{\ast}}f({\mathbf{Z}})\right)^{\ast}\right) & D_{\mathbf{Z}^{\ast}}\left(\left(\nabla_{\mathbf{Z}^{\ast}}f({\mathbf{Z}})\right)^{\ast}\right)
\end{array}\right]\triangleq\dfrac{{1}}{2}\,\left[\begin{array}{ll}
\nabla_{\mathbf{Z}\mathbf{Z}^{\ast}}^{2}f({\mathbf{Z}}) & \nabla_{\mathbf{Z^{\ast}}\mathbf{Z}^{\ast}}^{2}f({\mathbf{Z}})\\
\nabla_{\mathbf{Z}\mathbf{Z}}^{2}f({\mathbf{Z}}) & \nabla_{\mathbf{Z^{\ast}}\mathbf{Z}}^{2}f({\mathbf{Z}})
\end{array}\right].\vspace{0.2cm}\label{eq:Hessian_augmented}
\end{equation}
Note that {[}cf. $\!\!$(\ref{eq:Jacobians_complex_derivatives}){]}
$\nabla_{\mathbf{Z}\mathbf{Z}^{\ast}}^{2}f({\mathbf{Z}})=\left(\nabla_{\mathbf{Z^{\ast}}\mathbf{Z}}^{2}f({\mathbf{Z}})\right)^{\ast}$
and $\nabla_{\mathbf{Z^{\ast}}\mathbf{Z}^{\ast}}^{2}f({\mathbf{Z}})=\left(\nabla_{\mathbf{Z}\mathbf{Z}}^{2}f({\mathbf{Z}})\right)^{\ast}$.
It follows from Proposition \ref{monotonicity_complexVI} applied
to $\mathbf{\mathbf{F}^{\mathbb{C}}}(\mathbf{Z})=\nabla_{\mathbf{Z}^{\ast}}f(\mathbf{Z})$
that $\mathcal{H}_{\mathbf{Z}\mathbf{Z}^{\ast}}f(\mathbf{Z})$ plays
the role of the classical Hessian matrix of $f$: let $\mathcal{K}\subseteq\mathbb{C}^{n\times m}$
be any convex set with nonempty interior, then \vspace{-0.3cm}
 
\begin{equation}
\begin{array}{lll}
f(\mathbf{Z})\mbox{ is convex on \ensuremath{\mathcal{K}}} & \,\,\Leftrightarrow\,\,\  & \mathcal{H}_{\mathbf{Z}\mathbf{Z}^{\ast}}f(\mathbf{Z})\overset{\mathcal{A}}{\succeq}\mathbf{0},\,\,\forall\mathbf{Z}\in\mathcal{K};\\
f(\mathbf{Z})\mbox{ is strictly convex \ensuremath{\mathcal{K}}} & \,\,\Leftarrow\,\,\  & \mathcal{H}_{\mathbf{Z}\mathbf{Z}^{\ast}}f(\mathbf{Z})\overset{\mathcal{A}}{\succ}\mathbf{0},\,\,\forall\mathbf{Z}\in\mathcal{K};\\
f(\mathbf{Z})\mbox{ is strongly convex on \ensuremath{\mathcal{K}}} & \,\,\Leftrightarrow\,\,\  & \mathcal{H}_{\mathbf{Z}\mathbf{Z}^{\ast}}f(\mathbf{Z})-c_{\text{{sm}}}\,\mathbf{I}\overset{\mathcal{A}}{\succeq}\mathbf{0},\,\,\forall\mathbf{Z}\in\mathcal{K}\,\mbox{and some }c_{\text{{sm}}}>0.
\end{array}\label{eq:convexity_complex_case}
\end{equation}

If the set $\mathcal{K}$ has empty interior, conditions (\ref{eq:convexity_complex_case})
are replaced by those in Proposition \ref{VI_monotonicity_closed_sets}
applied to $\mathbf{JF}^{\mathbb{C}}(\mathbf{Z})=\mathcal{H}_{\mathbf{Z}\mathbf{Z}^{\ast}}f(\mathbf{Z})$;
we leave this easy task to the reader. Note that our conditions in
(\ref{eq:convexity_complex_case}) (and Proposition \ref{VI_monotonicity_closed_sets})
generalize those obtained in \cite[Prop. 1.2.6 and Exercise 1.8]{BertsekasNedicOzdaglar_book_convex03}
for real-valued functions of real variables.

\noindent\textbf{Example \ref{Example_MIMO_complex_min_princ} Revisited}.
Going back to the optimization problem (\ref{eq:rate_max}), we can
recover the well-known concavity property of $f(\mathbf{Z})$ on the
compact and convex set $\mathcal{K}$ by a direct application of Proposition
\ref{VI_monotonicity_closed_sets}. The expression of the augmented
Hessian of $f(\mathbf{Z})$ will be also used in Sec. \ref{sub:The-MIMO-case_revised}
to study MIMO games. 

Let $\widetilde{\mathcal{K}}$ be any open set containing $\mathcal{K}$
over which $f(\mathbf{Z})$ is well defined, and let $\tilde{f}:\mathcal{\widetilde{\mathcal{K}}}\rightarrow\mathbb{R}$
be $\tilde{f}(\mathbf{Z})\triangleq2\,\text{{Re}}\left(f(\mathbf{Z})\right)$.
Since $\tilde{f}=f$ on $\mathcal{K}$, concavity of $f$ on $\mathcal{K}$
follows from that of $ $$\tilde{f}$ on $\mathcal{K}$. Since $\mathcal{K}$
has empty interior, one needs to use Proposition \ref{VI_monotonicity_closed_sets}.
Observing that, for the specific set $\mathcal{K}$ under consideration,
the set $S_{\mathcal{K}}$ in Proposition \ref{VI_monotonicity_closed_sets}
is $S_{\mathcal{K}}=\{\mathbf{X}\in\mathbb{C}^{n\times n}\,:\,\mathbf{X}=\mathbf{X}^{H}\}$,
it is sufficient to show that\vspace{-0.2cm} 
\begin{equation}
-\left[\begin{array}{c}
\text{{vec}}(\mathbf{Y})\\
\text{{vec}}(\mathbf{Y}^{\ast})
\end{array}\right]^{H}\mathcal{H}_{\mathbf{Z}\mathbf{Z}^{\ast}}\tilde{{f}}(\mathbf{Z})\left[\begin{array}{c}
\text{{vec}}(\mathbf{Y})\\
\text{{vec}}(\mathbf{Y}^{\ast})
\end{array}\right]\geq0,\quad\forall\mathbf{Z}\in\mathcal{K}\,\mbox{and}\,\forall\mathbf{Y}=\mathbf{Y}^{H},\label{eq:concavity-Hessian_f}
\end{equation}
where $\mathcal{H}_{\mathbf{Z}\mathbf{Z}^{\ast}}\tilde{{f}}(\mathbf{Z}_{\mathcal{K}})$
is the augmented Hessian of $\tilde{f}(\mathbf{Z})$. In Appendix
\ref{App: Example 24-Intermediate}, we show that 
\begin{equation}
\mathcal{H}_{\mathbf{Z}\mathbf{Z}^{\ast}}\tilde{{f}}(\mathbf{Z})=-\dfrac{{1}}{2}\,\left[\begin{array}{cc}
\mathbf{0} & \left[\mathbf{G}(\mathbf{Z})^{T}\otimes\mathbf{G}(\mathbf{Z})\right]\mathbf{K}_{n^{2}n^{2}}\\
\left[\mathbf{G}(\mathbf{Z})^{H}\otimes\mathbf{G}^{\ast}(\mathbf{Z})\right]\mathbf{K}_{n^{2}n^{2}} & \mathbf{0}
\end{array}\right]\label{eq:augmented_Hessian_f_tilde}
\end{equation}
where $\mathbf{G}(\mathbf{Z})\triangleq\mathbf{H}^{H}\left(\mathbf{R}_{n}+\mathbf{H}\mathbf{Z}^{H}\mathbf{H}^{H}\right)^{-1}\mathbf{H}$
and $\mathbf{K}_{n^{2}n^{2}}$ is an $n^{2}\times n^{2}$ permutation
matrix such that $\text{{vec}}(\mathbf{Z}^{T})=\mathbf{K}_{n^{2}n^{2}}\text{{vec}}(\mathbf{Z})$
(also termed commutation matrix \cite[Def. 2.9]{Are_book_MatrixDiff}).
Using (\ref{eq:augmented_Hessian_f_tilde}), condition (\ref{eq:concavity-Hessian_f})
becomes
\begin{align}
0 & \leq\left[\begin{array}{c}
\text{{vec}}(\mathbf{Y})\\
\text{{vec}}(\mathbf{Y}^{\ast})
\end{array}\right]^{H}\left[\begin{array}{cc}
\mathbf{0} & \mathbf{G}(\mathbf{Z})^{T}\otimes\mathbf{G}(\mathbf{Z})\\
\mathbf{G}(\mathbf{Z})^{H}\otimes\mathbf{G}(\mathbf{Z})^{\ast} & \mathbf{0}
\end{array}\right]\left[\begin{array}{c}
\mathbf{K}_{n^{2}n^{2}}\,\text{{vec}}(\mathbf{Y})\\
\mathbf{K}_{n^{2}n^{2}}\,\text{{vec}}(\mathbf{Y}^{\ast})
\end{array}\right]\medskip\nonumber \\
 & =\left[\begin{array}{c}
\text{{vec}}(\mathbf{Y})\\
\text{{vec}}(\mathbf{Y}^{\ast})
\end{array}\right]^{H}\left[\begin{array}{cc}
\mathbf{G}(\mathbf{Z})^{T}\otimes\mathbf{G}(\mathbf{Z}) & \mathbf{0}\\
\mathbf{0} & \mathbf{G}(\mathbf{Z})^{H}\otimes\mathbf{G}(\mathbf{Z})^{\ast}
\end{array}\right]\left[\begin{array}{c}
\text{{vec}}(\mathbf{Y})\\
\text{{vec}}(\mathbf{Y}^{\ast})
\end{array}\right],\quad\quad\forall\mathbf{Z}\in\mathcal{K}\,\mbox{and}\,\forall\mathbf{Y}=\mathbf{Y}^{H},\label{eq:psd_augmented_hessian}
\end{align}
where in the equality we used the property $\mathbf{K}_{n^{2}n^{2}}\text{{vec}}(\mathbf{Z})=\text{{vec}}(\mathbf{Z}^{T})$
and $\mathbf{Y}=\mathbf{Y}^{H}$. It turns our that (\ref{eq:psd_augmented_hessian})
is satisfied if $\mathbf{G}(\mathbf{Z})^{T}\otimes\mathbf{G}(\mathbf{Z})$
is positive semidefinite for all $\mathbf{Z}\in\mathcal{K}$. Since
$\mathbf{G}(\mathbf{Z})^{T}\otimes\mathbf{G}(\mathbf{Z})$ is Hermitian
on $\mathcal{K}$, it is sufficient to check that the minimum eigenvalue
of $\mathbf{G}(\mathbf{Z})^{T}\otimes\mathbf{G}(\mathbf{Z})$, denoted
by $\lambda_{\text{{min}}}\left(\mathbf{G}(\mathbf{Z})^{T}\otimes\mathbf{G}(\mathbf{Z})\right)$,
is nonnegative for all $\mathbf{Z}$ on $\mathcal{K}$. The result
follows from $\lambda_{\text{{min}}}\left(\mathbf{G}(\mathbf{Z})^{T}\otimes\mathbf{G}(\mathbf{Z})\right)=\lambda_{\text{{min}}}\left(\mathbf{G}(\mathbf{Z})^{T}\right)\cdot\lambda_{\text{{min}}}\left(\mathbf{G}(\mathbf{Z})\right)=\lambda_{\text{{min}}}\left(\mathbf{G}(\mathbf{Z})\right)^{2}\geq0$
for all $\mathbf{Z}\in\mathcal{K}$, which proves concavity of $\tilde{f}(\mathbf{Z})$
on $\mathcal{K}$ and thus of $f(\mathbf{Z})$ on $\mathcal{K}$.

It is worth observing that while $-\mathcal{H}_{\mathbf{Z}\mathbf{Z}^{\ast}}f(\mathbf{Z})$
satisfies (\ref{eq:concavity-Hessian_f}) {[}Proposition \ref{VI_monotonicity_closed_sets}(a){]},
it is not positive (semi-) definite, showing that the latter condition
may be too restrictive for checking the convexity of a (real-valued)
function of complex variables. This strengths the importance of the
proposed new concept of augmented positive (semi-)definiteness (Definition
\ref{The-augmented-Jacobian}) and the role of Propositions \ref{monotonicity_complexVI}
and \ref{VI_monotonicity_closed_sets}. \hfill{}$\square$

\subsection{NEPs in the complex domain\label{sub:NEP_complex_domain}\vspace{-0.2cm}}

\noindent We can now establish the formal connection between complex
NEPs and complex VIs. Let $\mathcal{G}_{\mathbb{C}}\triangleq\left\langle \mathcal{K},\mathbf{f}\right\rangle $
be a complex NEP where each player controls a complex matrix $\mathbf{Z}_{i}\in\mathbb{C}^{n_{i}\times m_{i}}$
that must belong to the player's feasible set $\mathcal{K}_{i}\subseteq\mathbb{C}^{n_{i}\times m_{i}}$;
the cost function of each player is denoted by $f_{i}:\mathcal{K}\rightarrow\mathbb{R}$;
and the joint strategy set of the game is $\mathcal{K}=\prod_{i}\mathcal{K}_{i}$.
We also write $\mathbf{Z}\triangleq(\mathbf{Z}_{i})_{i=1}^{I}$, $\mathbf{Z}_{-i}\triangleq(\mathbf{Z}_{1},\ldots,\mathbf{Z}_{i-1},\mathbf{Z}_{i+1},\ldots\mathbf{Z}_{I})$,
and $\mathcal{K}_{-i}\triangleq\prod_{j\neq i}\mathcal{K}_{j}$. The
NEP problem $\mathcal{G}_{\mathbb{C}}$ consists then, for each player
$i=1,\ldots,I$, in solving the following convex optimization problem:
given $\mathbf{Z}_{-i}\in\mathcal{K}_{-i}$, 
\begin{equation}
\begin{array}{ll}
\limfunc{minimize}\limits _{\mathbf{Z}_{i}} & f_{i}(\mathbf{Z}_{i},\mathbf{Z}_{-i})\\[5pt]
\text{subject to} & \mathbf{Z}_{i}\in\mathcal{K}_{i}.
\end{array}\label{eq:NEP_complex}
\end{equation}
Building on $ $Lemma \ref{Lemma_min_principle}, it is not difficult
to prove the following.\vspace{-0.3cm}

\begin{proposition}\label{Lemma_NEP_complex_VI} Given the complex
NEP $\mathcal{G}_{\mathbb{C}}\triangleq\left\langle \mathcal{K},\mathbf{f}\right\rangle $
, suppose that for each player $\ensuremath{i}$ the following hold:\vspace{-0.3cm}

\begin{description}
\item [{{i)}}] the (nonempty) strategy set $\ensuremath{\mathcal{K}_{i}}$
is closed and convex;\vspace{-0.3cm}

\item [{{ii)}}] the payoff function $f_{i}(\mathbf{Z}_{i},\mathbf{Z}_{-i})$
is convex and continuously $\mathbb{R}$-differentiable in $\ensuremath{\mathbf{Z}_{i}}$
for every fixed $\mathbf{Z}_{-i}$. \vspace{-0.3cm}

\end{description}
Then, the complex NEP $\mathcal{G}_{\mathbb{C}}$ is equivalent to
the complex \emph{VI}$(\mathcal{K},\mathbf{G}^{\mathbb{C}})$, where
$\mathbf{G}^{\mathbb{C}}(\mathbf{Z})\triangleq\left(\nabla_{\mathbf{Z}_{i}^{\ast}}f_{i}(\mathbf{Z})\right)_{i=1}^{I}$.\end{proposition}\vspace{-0.3cm}

We now have all the tools necessary to extend the developments in
Sec. \ref{sec:Nash-Equilibrium-Problems} and \ref{sub:Distributed-algorithms-for_NEP}
to the solution of the complex $\mathcal{G}_{\mathbb{C}}$. In fact,
by using the results in this section about monotonicity/P properties
of $\text{{VI}}(\mathcal{K},\mathbf{G}^{\mathbb{C}})$ we can verbatim
mimic the developments of Sec. \ref{sec:Nash-Equilibrium-Problems}
and \ref{sub:Distributed-algorithms-for_NEP}. We remark that this
possibility was not obvious and is instead a result of careful choices
about the way to deal with complex functions. Because of space limitations,
we do not provide here all the analytic developments; however, in
the next section we will illustrate them on the specific MIMO games
introduced in Sec. \ref{sub:The-MIMO-case}.\vspace{-0.2cm}

\section{Noncooperative Games Over Interference Channels Revisited \label{sec:Applications}\vspace{-0.2cm}}

\label{Sec_applications} In this section, we focus on the application
of the general theory developed in the previous sections to some concrete
examples of practical interest. In particular, we show how the real/complex
NEPs introduced in Sec. \ref{sub:Some-motivating-examples} can be
naturally casted in the proposed framework and thus efficiently solved.
The main result in the SISO case is a novel iterative water-filling
like algorithm where the users can choose the degree of desired cooperation
via local pricing, converging to solutions having different performance/signaling
trade-off; we also prove that the best-response of each player has
a multi-level waterfilling-like expression and provide an efficient
algorithm for its computation. We then extend our analysis to MIMO
games and obtain similar results. Numerical experiments show the superiority
of our novel distributed algorithms with respect to plain noncooperative
solutions as well as very good performance with respect to centralized
schemes.\vspace{-0.2cm}

\subsection{The SISO case\label{sub:Game G_the-SISO}}

We study here the game $\mathcal{G}_{\texttt{{siso}}}=\left\langle \mathcal{P^{\,\texttt{{siso}}}},\,(r_{i})_{i=1}^{I}\right\rangle $
introduced in (\ref{eq:max_rate_game}). The VI function associated
with $\mathcal{G}_{\texttt{siso}}$\textbf{ }is $\mathbf{G}(\mathbf{p})\triangleq(\mathbf{G}_{i}(\mathbf{p}))_{i=1}^{I}:\mathcal{P^{\,\texttt{{siso}}}}\rightarrow\mathbb{R}^{N\, I}$
, where each $\mathbf{G}_{i}(\mathbf{p})$ is defined as 
\begin{equation}
\mathbf{G}_{i}(\mathbf{p})\triangleq-\nabla_{\mathbf{p}_{i}}r_{i}(\mathbf{p})=\left(-\dfrac{{|H_{ii}(k)|^{2}}}{\sigma_{i}^{2}(k)+\sum_{j\neq i}|H_{ij}(k)|^{2}p_{j}(k)}\right)_{k=1}^{N}.\label{F_VI}
\end{equation}
Note that in this section, due to the nature of the problems at hand,
we called the VI mapping $\mathbf{G}$ instead of $\mathbf{F}$ used
previously, and the VI variables $\mathbf{p}$ instead of $\mathbf{x}$
as used previously.

According to Proposition \ref{monotonicity}, the monotonicity/P properties
of $\mathbf{G}$ are related to the matrices $\mathbf{{J}}\mathbf{G}_{\text{{low}}}$
and $\mathbf{\boldsymbol{\Upsilon}}_{\mathbf{G}}$ defined in (\ref{eq:def_lower_of_comparison_of_Jacobian})
and (\ref{eq:Upsilon_matrix}). We recall that the matrices $\mathbf{B}$
and $\mathbf{C}_{i}$'s in the definition of $\mathbf{{J}}\mathbf{G}_{\text{{low}}}$
and $\mathbf{\boldsymbol{\Upsilon}}_{\mathbf{G}}$ represent a degree
of freedom that one can use; in this case it is convenient to make
the following choices. Let us rearrange the components of $\mathbf{p}$
by subcarriers, meaning that the vector $\mathbf{p}=(\mathbf{p}_{i})_{i=1}^{I}$
is permuted into $\mathbf{\bar{{\mathbf{p}}}}=(\bar{{\mathbf{p}}}(k))_{k=1}^{N}$,
with $\bar{{\mathbf{p}}}(k)=(p_{i}(k))_{i=1}^{I}$; it is not difficult
to see that $\bar{{\mathbf{p}}}$ can be written as $\bar{{\mathbf{p}}}=\mathbf{P}\mathbf{p}$,
where $\mathbf{P}$ is a permutation matrix such that $[\mathbf{P}]_{ij}=1$
if $j=[(i\,\text{{mod}\,\ }I)-1]N+\left\lceil i/I\right\rceil \,\text{{mod}}(I\cdot N)$,
and $[\mathbf{P}]_{ij}=0$ otherwise. Using this new ordering for
the variables, matrix $\mathbf{{J}}\mathbf{G}$ becomes $\mathbf{P}^{T}\mathbf{{J}}\mathbf{G}\mathbf{P}$;
$\mathbf{{J}}\mathbf{G}_{\text{{low}}}$ is then obtained from $\mathbf{P}^{T}\mathbf{{J}}\mathbf{G}\mathbf{P}$
according to (\ref{eq:def_lower_of_comparison_of_Jacobian}), where
$\mathbf{B}\triangleq\text{{Diag}}\{\left(\mathbf{B}(k)\right)_{k=1}^{N}\}$
is a block diagonal matrix, with each block $\mathbf{B}(k)\in\mathbb{R}^{I\times I}$
being a positive diagonal matrix with the $i$-th entry equal to $\left[\mathbf{B}(k)\right]_{ii}\triangleq\sigma_{i}^{2}(k)/|H_{ii}(k)|^{2}+$
$\sum_{j}(|H_{ij}(k)|^{2}/|H_{ii}(k)|^{2})p_{j}^{\max}(k)$. Matrix
$\mathbf{\boldsymbol{\Upsilon}}_{\mathbf{G}}$ comes directly from
the original $ $$\mathbf{{J}}\mathbf{G}$ by choosing each $\mathbf{C}_{i}\in\mathbb{R}^{N\times N}$
as a diagonal matrix, defined as $\mathbf{C}_{i}\triangleq{\mbox{Diag}}\left\{ \left(\left(\sigma_{i}^{2}(k)+\sum_{j}|H_{ij}(k)|^{2}p_{j}^{\max}(k)\right)/|H_{ii}(k)|^{2}\right)_{k=1}^{N}\right\} $.
The explicit expressions of $\mathbf{{J}}\mathbf{G}_{\text{{low}}}$
and $\mathbf{\boldsymbol{\Upsilon}}_{\mathbf{G}}$ are the following:
$\mathbf{{J}}\mathbf{G}_{\text{{low}}}\triangleq\mathbf{\text{{Diag}}}\{(\mathbf{{J}}\mathbf{G}_{\text{{low}}}(k))_{k=1}^{N}\}\in\mathbb{R}^{N\, I\times N\, I}$
is a block diagonal matrix, whose $k$-th diagonal block $\mathbf{{J}}\mathbf{G}_{\text{{low}}}(k)\in\mathbb{R}^{I\times I}$
is\vspace{-0.3cm} 
\begin{equation}
[\,\mathbf{{J}}\mathbf{G}_{\text{{low}}}(k)\,]_{ij}\,\,\triangleq\,\left\{ \begin{array}{ll}
1, & \mbox{if \ensuremath{i=j}}\\[5pt]
{\displaystyle -\dfrac{\left|H_{ij}(k)\right|^{2}}{\left|H_{jj}(k)\right|^{2}}\cdot{\textsf{innr}}{}_{ij}(k),} & \mbox{if \ensuremath{i\neq j}},
\end{array}\right.\label{eq:def_gamma_k-1}
\end{equation}
and ${\boldsymbol{\Upsilon}}_{\mathbf{G}}\in\mathbb{R}^{I\times I}$
is given by\vspace{-0.3cm} 
\begin{equation}
[\,{\boldsymbol{\Upsilon}}_{\mathbf{G}}\,]_{ij}\,\,\triangleq\,\left\{ \begin{array}{ll}
1 & \mbox{if \ensuremath{i=j}}\\[5pt]
-{\displaystyle {\max_{1\leq\, k\,\leq N}}\,\left\{ \dfrac{\left|H_{ij}(k)\right|^{2}}{\left|H_{jj}(k)\right|^{2}}\cdot{\textsf{innr}}{}_{ij}(k)\right\} } & \mbox{if \ensuremath{i\neq j}},
\end{array}\right.\label{eq:def_gamma_hat}
\end{equation}
with 
\begin{equation}
{\textsf{innr}}{}_{ij}(k)\triangleq\dfrac{\sigma_{j}^{2}(k)+\sum_{r}|H_{jr}(k)|^{2}p_{r}^{\max}(k)}{\sigma_{i}^{2}(k)}.\label{eq:innr}
\end{equation}
Using the above definitions along with Proposition \ref{monotonicity},
Theorem \ref{Theo_existence_uniquenessNE}, and Corollary \ref{Cor_SF_for_Pmat},
the main properties of $\mathcal{G}$ are then given in the following
proposition; Corollary \ref{Cor_SF_G_for_Pmat} follows from Proposition
\ref{Cor_SF_for_Pmat}. \vspace{-0.2cm}

\begin{proposition} \label{Prop:NEP_G_properties} Given the real
convex-player NEP $\mathcal{G}_{\texttt{{siso}}}=\left\langle \mathcal{\mathcal{P^{\,\texttt{{siso}}}}},\,(r_{i})_{i=1}^{I}\right\rangle $,
the following hold.

\vspace{-0.3cm}

\begin{description}
\item [{(a)}] $\mathcal{G}_{\texttt{{siso}}}$ is equivalent to the \emph{VI}$(\mathcal{P^{\,\texttt{{siso}}}},\mathbf{G})$,
which has a nonempty and compact solution set;\vspace{-0.3cm}

\item [{(b)}] Suppose that $\mathbf{{J}}\mathbf{G}_{\text{{low}}}$ is
positive semidefinite (positive definite). Then $\mathbf{G}$ is monotone
(strongly monotone) on $\mathcal{P^{\,\texttt{{siso}}}}$; therefore
$\mathcal{G}_{\texttt{{siso}}}$ is a monotone NEP;\vspace{-0.3cm}

\item [{(c)}] Suppose that ${\boldsymbol{\Upsilon}}_{\mathbf{G}}$ is a
P-matrix (positive definite matrix). Then $\mathbf{G}$ is a uniformly
P-function (strongly monotone function) on $\mathcal{P^{\,\texttt{{siso}}}}$;
therefore $\mathcal{G}_{\texttt{{siso}}}$ is a \emph{$P{}_{\boldsymbol{\Upsilon}}$}
NEP and has a unique NE. 
\end{description}
\end{proposition}\vspace{-0.4cm}

\begin{corollary} \label{Cor_SF_G_for_Pmat} The matrix ${\boldsymbol{\Upsilon}}_{\mathbf{G}}$
in \emph{(\ref{eq:def_gamma_hat})} is a P-matrix (or a positive definite
matrix) if one (or both) of the following two sets of conditions are
satisfied: for some $\mathbf{w}=(w_{i})_{i=1}^{I}>\mathbf{0}$,\vspace{-0.1cm}
\begin{equation}
\begin{array}{rll}
\mbox{\emph{Low}\,\ \emph{received}\,\ \emph{MUI}\emph{:}} & \dfrac{{1}}{w_{i}}\,\dsum\limits _{j\neq i}w_{j}{\displaystyle {\max_{1\leq\, k\,\leq N}}\,\left\{ \dfrac{\left|H_{ij}(k)\right|^{2}}{\left|H_{jj}(k)\right|^{2}}\cdot\mathcal{\textsf{\emph{innr}}}{}_{ij}(k)\right\} }<1,\quad & \forall i=1,\cdots,I,\medskip\\
\mbox{\emph{Low}\,\ \emph{generated}\,\ \emph{MUI}\emph{:}} & \dfrac{{1}}{w_{j}}\dsum\limits _{i\neq j}w_{i}{\displaystyle {\max_{1\leq\, k\,\leq N}}\,}\left\{ \dfrac{\left|H_{ij}(k)\right|^{2}}{\left|H_{jj}(k)\right|^{2}}\cdot{\textsf{\emph{innr}}}{}_{ij}(k)\right\} <1, & \forall j=1,\cdots,I\vspace{-0.1cm}
\end{array}\label{eq:Low_TxTx_MUI}
\end{equation}
\end{corollary} Similar sufficient conditions can be obtained for
$\mathbf{{J}}\mathbf{G}_{\text{{low}}}(k)$ to be positive semidefinite.

These conditions have an interesting physical interpretation: the
uniqueness of the NE is ensured if the interference among the SUs
is sufficiently small; this is clearly shown by Corollary \ref{Cor_SF_G_for_Pmat}.
Specifically, the first condition in (\ref{eq:Low_TxTx_MUI}) can
be interpreted as a constraint on the maximum amount of interference
that each receiver can tolerate, whereas the second condition introduces
an upper bound on the maximum level of interference that each transmitter
is allowed to generate. We will show shortly that these conditions
play a role also in the convergence of the proposed distributed iterative
algorithms. Moreover, depending on the level of interference in the
network, the NEP $\mathcal{G}_{\texttt{{siso}}}$ is a \emph{$P{}_{\boldsymbol{\Upsilon}}$}
or monotone NEP, implying different properties and solution schemes
of the game, as described next; we classify these two scenarios as
\emph{low-interference} and \emph{medium/high-interference} regime,
respectively.

\medskip{}
\noindent\textbf{The case of }\emph{$P{}_{\boldsymbol{\Upsilon}}$}\textbf{
NEP} (low-interference regime). When the matrix ${\boldsymbol{\Upsilon}}_{\mathbf{G}}$
is a P matrix (or positive definite), the NEP $\mathcal{G}_{\texttt{{siso}}}$
is a \emph{$P{}_{\boldsymbol{\Upsilon}}$} NEP {[}Proposition \ref{Prop:NEP_G_properties}
(c){]}. Invoking Theorem \ref{Theo-async_best-response_NEP}, the
unique NE of the game can be computed with convergence guarantee using
Algorithm \ref{async_best-response_algo} on $\mathcal{G}_{\texttt{{siso}}}$,
as stated in the next theorem.

\noindent \begin{theorem}\label{Th:AIWFA} Suppose that $\mathcal{G}_{\texttt{{siso}}}$
is a P$_{{\boldsymbol{\Upsilon}}}$ real NEP. Then, any sequence $\{\mathbf{p}^{(n)}\}_{n=0}^{\infty}$
generated by Algorithm \ref{async_best-response_algo} applied to
$\mathcal{G}_{\texttt{{siso}}}$ converges to the unique NE of the
NEP, for any given updating schedule of the players satisfying assumptions
A1\emph{)}-A3\emph{)}.\end{theorem}

When implementing Algorithm \ref{async_best-response_algo}, each
user needs to compute his best-response solution, given the interference
generated by the others. In Sec. \ref{sub:Preliminary-results:-Efficient_proximalBR}
(cf. Lemma \ref{Lemma_best-response_SISO}), we prove that the best-response
for the game $\mathcal{G}_{\texttt{{siso}}}$ has a multi-level waterfilling-like
expression, implying that each user can compute his optimal solution
locally and very efficiently (he only needs to measure the overall
MUI experienced at his receiver and ``waterfill'' over it). Therefore,
Algorithm \ref{async_best-response_algo} results to be totally distributed
and computationally efficient, which makes it appealing for practical
implementation in CR scenarios.

\smallskip{}
\noindent \textbf{The case of monotone NEP} (medium/high-interference
regime). When $\mathbf{{J}}\mathbf{G}_{\text{{low}}}$ is positive
semidefinite, the NEP $\mathcal{G}_{\texttt{{siso}}}$ is a monotone
NEP, having in general multiple solutions. In such a case, to compute
a solution of $\mathcal{G}_{\texttt{{siso}}}$ with convergence guarantee,
there are two available options, namely: PDAs (either in their exact
or inexact form) and PTRA. The former are the only feasible choice
when the SUs are not willing to cooperate; whereas the latter requires
some (albeit very limited) cooperation among the SUs in favor of better
performance (one can perform equilibrium selection). To the best of
our knowledge, the above algorithms are in the signal processing and
communication literature the first example of distributed power control
schemes that converge \emph{even in the presence of multiple Nash
equilibria}. Note that, in all the aforementioned algorithms, the
best-response of the SUs can be efficiently computed via a multi-level
waterfilling expression (cf. Sec. \ref{sub:Preliminary-results:-Efficient_proximalBR}).
We provide next an instance of the PTRA along with its convergence
conditions; PDAs are obtained as special cases of the PTRA, and thus
its description is omitted. \smallskip{}

\noindent \emph{Equilibrium Selection via Proximal-Tikhonov Regularization
Algorithm}. The first step is to choose a merit function that quantifies
the quality of a NE of $\mathcal{G}_{\texttt{{siso}}}$. Different
heuristics can be used; as an example, here we focus on the following
merit function: given the vector $\mathbf{w}\triangleq(w_{i})_{i=1}^{I}\geq0$,
let 
\begin{equation}
\phi(\mathbf{p})\triangleq\sum_{i=1}^{I}w_{i}\,\sum_{j\neq i}\,\sum_{k=1}^{N}|H_{ij}(k)|^{2}\, p_{\, i}(k).\label{outer_merit_function}
\end{equation}
This choice is motivated by the intuition that among all the solutions
of $\mathcal{G}_{\texttt{{siso}}}$, a good candidate is the one that
minimizes the overall interference among the users, measured by $\phi(\mathbf{p})$,
likely resulting in an ``higher'' sum-rate $\sum_{i=1}^{I}r_{i}(\mathbf{p})$.
The NE selection problem based on the merit function $\phi$ can be
then formulated as: 
\begin{equation}
\begin{array}{ll}
\limfunc{minimize}\limits _{\mathbf{p}} & \phi(\mathbf{p})\\[5pt]
\text{subject to} & \mathbf{p}\,\in\,\sol(\mathcal{\mathcal{P^{\,\texttt{{siso}}}}},\,\mathbf{r}).
\end{array}\label{eq:NEP constrained opt}
\end{equation}

Problem (\ref{eq:NEP constrained opt}) is an instance of (\ref{eq:VI-C min});
we can then solve it by applying Algorithm \ref{algo2} (cf. Sec.
\ref{sub:Equilibrium-Selection-Monotone-NEP}); which corresponds
to solving a sequence of perturbed \emph{$P{}_{\boldsymbol{\Upsilon}}$}
NEPs given by $\mathcal{G}_{\tau,\,\varepsilon^{(n)},\bar{{\mathbf{p}}}}=\left\langle \mathcal{\mathcal{P^{\,\texttt{{siso}}}}},\,(-r_{i}(\mathbf{p})+\right.$
$\left.\varepsilon^{(n)}\,\boldsymbol{{\gamma}}_{i}^{T}\mathbf{p}_{i}+\dfrac{\tau}{2}\,\,\left\Vert \mathbf{p}_{i}-\bar{{\mathbf{p}}}_{i}\right\Vert ^{2})_{i-1}^{I}\right\rangle $,
whose player $i$'s optimization problem is: given $\mathbf{p}_{-i}$,
$\bar{{\mathbf{p}}}$, and $\varepsilon^{(n)}>0$,\vspace{-0.2cm}
\begin{equation}
\begin{array}{l}
\underset{\mathbf{p}_{i}\in\mathcal{P}_{i}^{\,\texttt{{siso}}}}{\textnormal{maximize}}\quad r_{i}(\mathbf{p}_{i},\mathbf{p}_{-i})-\varepsilon^{(n)}\,\boldsymbol{{\gamma}}_{i}^{T}\mathbf{p}_{i}-\dfrac{\tau}{2}\,\,\left\Vert \mathbf{p}_{i}-\bar{{\mathbf{p}}}_{i}\right\Vert ^{2}\end{array}\vspace{-0.3cm}\label{eq:sub-VI-game}
\end{equation}
where $\boldsymbol{\gamma}\triangleq(\boldsymbol{\gamma}_{i})_{i=1}^{I}$,
with each $\boldsymbol{\gamma}_{i}\triangleq(\sum_{j\neq i}w_{i}|H_{ji}(k)|^{2})_{k=1}^{N}$.
A partially asynchronous version of Algorithm \ref{algo2} applied
to (\ref{eq:NEP constrained opt}) is described in Algorithm \ref{algo_power_bi_level}
below, and its convergence conditions are given in Theorem \ref{convergence_NE_sel},
which is a direct application of results in Theorem \ref{Theo-async_best-response_NEP}
and Theorem \ref{the:mod}.

\begin{algo}{NE selection for the real NEP $\mathcal{G}_{\texttt{{siso}}}$}
S\texttt{$\mbox{(\mbox{Data})}:$} $\{\varepsilon^{(n)}\}\downarrow0$
and $\tau>0$.\\[1pt] \texttt{$\mbox{(\mbox{S.0}):}$} Choose any
$\mathbf{p}^{(0)}\in\mathcal{\,\mathcal{P}}^{\,\texttt{{siso}}}$
and a center $\bar{{\mathbf{p}}}\geq\mathbf{0}$ of the regularization;
set $\bar{{\varepsilon}}=\varepsilon^{(0)}$ and $n=0$. \\[1pt]
\texttt{$\mbox{(\mbox{S.1}):}$} If $\mathbf{p}^{(n)}$ satisfies
a suitable termination criterion, \texttt{STOP}.\\[1pt] 
 \texttt{$\mbox{(\mbox{S.2}):}$} For each $i=1,\ldots,I$, compute
$\mathbf{p}_{i}^{(n+1)}$ as 
\begin{equation}
\mathbf{p}_{i}^{(n+1)}=\left\{ \begin{array}{lll}
\mathbf{p}_{i}^{\star}\in\underset{\mathbf{p}_{i}\in\mathcal{\,\mathcal{P}}_{i}^{\,\texttt{{siso}}}}{\text{{argmax}}}\left\{ r_{i}\left(\mathbf{p}_{i},\mathbf{p}_{-i}^{(\boldsymbol{{\tau}}_{-i}(n))}\right)-\bar{{\varepsilon}}\,\boldsymbol{{\gamma}}_{i}^{T}\mathbf{p}_{i}-\dfrac{\tau}{2}\,\,\left\Vert \mathbf{p}_{i}-\bar{{\mathbf{p}}}_{i}\right\Vert ^{2}\right\} , &  & \mbox{if }n\in\mathcal{T}_{i}\bigskip\\
\mathbf{p}_{i}^{(n)}, &  & \mbox{otherwise}
\end{array}\right.\label{eq:Async_update_NE_sel}
\end{equation}
\texttt{$\mbox{(\mbox{S.3})}:$} If $\mathbf{p}^{(n+1)}$ is a NE
of $\mathcal{G}_{\tau,\,\varepsilon^{(n)},\bar{{\mathbf{p}}}}$, then
update $\bar{{\varepsilon}}$ and the center $\bar{{\mathbf{p}}}$:\vspace{-0.2cm}
\begin{equation}
\bar{{\varepsilon}}={\varepsilon}^{(n+1)}\quad\mbox{and}\quad\bar{{\mathbf{p}}}_{i}=\mathbf{p}_{i}^{(n+1)}\,\,\,\,\,\forall i=1,\ldots,I;\vspace{-0.1cm}\label{eq:step_size_and_centroid_update}
\end{equation}
\texttt{$\mbox{(\mbox{S.4})}:$} $n\leftarrow n+1$ and return to
\texttt{$\mbox{(\mbox{S.1})}$}. \label{algo_power_bi_level} \end{algo} 

\begin{theorem}\label{convergence_NE_sel} Suppose that: i) $\mathcal{G}_{\texttt{{siso}}}$
is a monotone NEP; ii) $\{\varepsilon^{(n)}\}$ is such that $\varepsilon^{(n)}\rightarrow0$
and $\sum_{n=0}^{\infty}\,\varepsilon^{(n)}=\infty$; and iii) $\tau$
is chosen so that ${\boldsymbol{\Upsilon}}_{\mathbf{G}}+\tau\,\mathbf{I}$
is a P matrix. Then, the sequence $\{\mathbf{p}^{(n)}\}_{n=0}^{\infty}$
generated by Algorithm \ref{algo_power_bi_level} has a limit point
and every such limit point is a solution of the optimization problem
\eqref{eq:NEP constrained opt}\emph{.}\end{theorem} A sufficient
condition for matrix ${\boldsymbol{\Upsilon}}_{\mathbf{G}}+\tau\,\mathbf{I}$
in Theorem \ref{convergence_NE_sel} to be P is 
\begin{equation}
\tau\,>\,{\displaystyle {\max_{1\leq i\leq I}}\,\left\lbrace {\displaystyle {\sum_{j\neq i}}{\displaystyle {\max_{1\leq k\leq N}}\left\lbrace \dfrac{|H_{ij}(k)|^{2}}{|H_{ii}(k)|^{2}}\,{\textsf{innr}}{}_{ij}(k)\right\rbrace }}\right\rbrace -1.}\label{eq:bounds}
\end{equation}

Algorithm \ref{algo_power_bi_level} shows that, in the presence of
multiple equilibria, one can still have converge even when best-response
based schemes (cf. Algorithm \ref{async_best-response_algo}) fail,
provided that the SUs play a ``sequence\textquotedbl{} of games rather
than a one-shot game; moreover, to reach the NE that minimizes the
overall MUI among the users, the players' objective functions need
to be modified in order to contain an additional term$-$the linear
term $\varepsilon^{(n)}\boldsymbol{\gamma}_{i}^{T}\mathbf{p}_{i}-$whose
task is to ``measure\textquotedbl{} on the way the quality of the
solution that the algorithm is going to reach. Such a term has also
a physical interpretation: it represents a punishment imposed to the
users for using too much power and thus generating too much MUI.

Note that the computation of the optimal power allocations of the
users in Algorithm \ref{algo_power_bi_level} can be performed locally
and distributively by the users as previously discussed for the \emph{$P{}_{\boldsymbol{\Upsilon}}$}
NEP, once $\varepsilon^{(n)}$ and $\boldsymbol{\gamma}_{i}$ are
given. The computation of $\boldsymbol{\gamma}_{i}$ requires an estimate
from each user $i$ of the cross-channel between its transmitter and
the receivers of all SUs being in the coverage radius of user $i$.
This estimate needs to be computed only once (before running the algorithm)
and updated at the rate of the coherence time of the channel. When
the computation of $\boldsymbol{\gamma}_{i}$ is not possible, one
can still use Algorithm \ref{algo_power_bi_level}, setting $\boldsymbol{\gamma}_{i}=\mathbf{0}$
in (\ref{eq:Async_update_NE_sel}), which corresponds to solving the
optimization problem \eqref{eq:NEP constrained opt} with $\phi(\mathbf{p})=0$,
and thus computing just \emph{one} of the solutions of $\mathcal{G}_{\texttt{{siso}}}$;
Theorem \ref{convergence_NE_sel} still guarantees convergence of
the algorithm, even in the presence of multiple equilibria.

\subsubsection{Efficient computation of the SISO proximal best-response solutions
\label{sub:Preliminary-results:-Efficient_proximalBR}}

In this section, we provide an efficient method for computing the
best-response solutions of the rate maximization problems introduced
in the previous section. Motivated also by other resource allocation
problems, such as \cite{ScutariPalomarFacchineiPang_NETGCOP11}, we
introduce next a very general formulation that encompasses the optimization
problems studied in this paper, whose optimal solution is proved to
have a multi-level waterfilling-like expression, and provide an efficient
algorithm to compute the optimal water-levels (dual variables).

Given $\{H_{k}\}_{k=1}^{N}$, $\boldsymbol{{\lambda}}=(\lambda_{k})_{k=1}^{N}$,
$\mathbf{c}=(c_{k})_{k=1}^{N}$, $\mathbf{w}_{k}\triangleq(w_{ki})_{i=1}^{N_{c}}$,
$\mathbf{p}^{\max}\triangleq(p_{k}^{\max})_{k=1}^{N}$, and $\boldsymbol{\alpha}\in\mathbb{R}_{++}^{N_{c}}$,
and $\tau>0$, with each $H_{k}>0$ and $\lambda_{k}>0$, consider
the following optimization problem 

\begin{equation}
\begin{array}{lcl}
\underset{\mathbf{p}}{\textnormal{maximize}} &  & \sum_{k=1}^{N}\left[\log\left(1+H_{k}\, p_{k}\right)-\lambda_{k}p_{k}\right]-\frac{\tau}{2}\left\Vert \mathbf{p}-\mathbf{c}\right\Vert ^{2}\\
 &  & \sum_{k=1}^{N}\mathbf{w}_{k}\, p_{k}\leq\boldsymbol{\alpha}\\
 &  & \mathbf{0}\leq\mathbf{p}\leq\mathbf{p}^{\max}.
\end{array}\label{eq:single_user_general_proximal_problem}
\end{equation}
We will tacitly assume w.l.o.g. that: $\mathbf{0}<(H_{k})_{k=1}^{N}<\infty$;
$\mathbf{0}\leq\mathbf{w}_{k}\triangleq(w_{ki})_{i=1}^{N_{c}}<\infty$
for all $k=1,\ldots,N$, and linearly independent; $\mathbf{0}\leq(\lambda_{k})_{k=1}^{N}<\infty$;
$\mathbf{0}\leq(c_{k})_{k=1}^{N}<\infty$; $\mathbf{0}<\mathbf{p}^{\max}<\infty$;
and $\sum_{k=1}^{N}\mathbf{w}_{k}\, p_{k}^{\max}>\boldsymbol{\alpha}$.

Problem (\ref{eq:single_user_general_proximal_problem}) is a convex
problem with a polyhedral feasible set; the KKT are then necessary
and sufficient conditions for the optimality. By solving the KKT system
one can prove the following result, whose proof is omitted because
of the space limitation; see \cite[Sec. 7.1.1]{Scutari-Facchinei-Pang-Palomar_IT_PI}. 

\begin{lemma} \label{Lemma_best-response_SISO}The optimal solution
of the optimization problem (\ref{eq:single_user_general_proximal_problem})
is given by 

\begin{equation}
p_{k}^{\star}=\left[\frac{1}{2}\left(c_{k}-\frac{1}{H_{k}}\right)-\frac{1}{2\tau}\left[\tilde{\mu}_{k}-\sqrt{\left[\tilde{\mu}_{k}-\tau\left(c_{k}+\frac{1}{H_{k}}\right)\right]^{2}+4\tau}\right]\right]_{0}^{p_{k}^{\max}}\qquad k=1,\ldots,N\label{eq:gen_proximal_BR}
\end{equation}
where $[x]_{a}^{b}$ denotes the Euclidean projection of $x$ onto
$[a,\, b]$, i.e., $[x]_{a}^{b}\triangleq\max(a,\min(x,b))$, each
$\tilde{\mu}_{k}\triangleq\lambda_{k}+\boldsymbol{\mu}^{T}\mathbf{w}_{k}$,
and the water-level vector $\boldsymbol{\mu}$ has to be chosen to
satisfy the complementarity conditions
\begin{eqnarray}
0 & \leq & \mu_{i}\perp\alpha_{i}-\sum_{k=1}^{N}w_{ki}\, p_{k}^{\star}\,\geq0,\qquad\forall i=1,\ldots,N_{c}.\label{eq:CCa_prox_BR}
\end{eqnarray}
\end{lemma}

The computation of the water-level $\boldsymbol{\mu}$ in (\ref{eq:gen_proximal_BR})
so that the complementarity conditions in (\ref{eq:CCa_prox_BR})
are satisfied can be done efficiently using the multiple nested bisection
method described in Algorithm \ref{alg:multiple_bisection_gen_prox_BR}. 

The basic idea of the algorithm is to employ a bisection algorithm
in $\mu_{1}$; then, for a given $\mu_{1}$, use a bisection algorithm
in $\mu_{2}$; then in $\mu_{3}$, and so on. For the $i$th bisection
level, the interval can be chosen as $\left[0,\max_{k}\left\{ \left(H_{k}-\tau c_{k}-\lambda_{k}-\sum_{j<i}\mu_{j}w_{kj}\right)/w_{ki}\right\} \right]$.
The convergence of the nested bisection method is given in the following
proposition, whose proof is omitted and can be found in \cite[Sec. 7.1.1]{Scutari-Facchinei-Pang-Palomar_IT_PI}.

\begin{proposition} \label{prop:convergence_multiple_bisection_gen_prox_BR}
Algorithm \ref{alg:multiple_bisection_gen_prox_BR} converges in no
more than 
\[
\prod_{i}\left\lceil \log_{2}\left(\max_{k}\left\{ \left(H_{k}-\tau c_{k}-\lambda_{k}\right)/w_{ki}\right\} /\epsilon\right)\right\rceil 
\]
iterations, where $\epsilon$ is the desired accuracy in the computation
of the parameter $\boldsymbol{\mu}$.\end{proposition}

\begin{algo}{Multiple nested bisection algorithm for the computation
of the proximal best-response in (\ref{eq:gen_proximal_BR}).} S\texttt{$\mbox{(\mbox{S.0})}:$}
Choose some accuracy $\epsilon$.

\texttt{$\mbox{(S.1)}:$} Set $\underline{\mu}_{1}=0$ and $\overline{\mu}_{1}=\max_{k}\left\{ \left(H_{k}-\tau c_{k}-\lambda_{k}\right)/w_{k1}\right\} $.

\texttt{$\mbox{(S.2)}:$} Set $\mu_{1}=\left(\underline{\mu}_{1}+\overline{\mu}_{1}\right)/2$.

\texttt{$\mbox{(S.3)}:$} Solve for $\mu_{2},\ldots$

\texttt{\qquad{}$\mbox{(S-2.1)}:$} Set $\underline{\mu}_{2}=0$
and $\overline{\mu}_{2}=\max_{k}\left\{ \left(H_{k}-\tau c_{k}-\lambda_{k}-\sum_{j<2}w_{kj}\mu_{j}\right)/w_{k2}\right\} $.

\texttt{\qquad{}$\mbox{(S-2.2)}:$} Set $\mu_{2}=\left(\underline{\mu}_{2}+\overline{\mu}_{2}\right)/2$.

\texttt{\qquad{}$\mbox{(S-2.3)}:$} Solve for $\mu_{3},\ldots$

\texttt{\qquad{}\qquad{}$\mbox{(S-3.1)}:$} Set $\underline{\mu}_{3}=0$
and

\texttt{\qquad{}\qquad{}\qquad{}\qquad{}\qquad{}}$\overline{\mu}_{3}=\max_{k}\left\{ \left(H_{k}-\tau c_{k}-\lambda_{k}-\sum_{j<3}w_{kj}\mu_{j}\right)/w_{k3}\right\} $.

\texttt{\qquad{}\qquad{}$\mbox{(S-3.2)}:$} Set $\mu_{3}=\left(\underline{\mu}_{3}+\overline{\mu}_{3}\right)/2$.

\texttt{\qquad{}\qquad{}$\mbox{(S-3.3)}:$} Solve for $\mu_{4},\ldots$\vspace{-0.2cm}

$\qquad\qquad\qquad\qquad\qquad\qquad\qquad\qquad\qquad\qquad\qquad$\vdots{}

\texttt{\qquad{}\qquad{}$\mbox{(S-3.4)}:$} Using (\ref{eq:gen_proximal_BR}),
if $\sum_{k=1}^{N}w_{k3}p_{k}^{\star}<\alpha_{3}$ then set $\overline{\mu}_{3}=\mu_{3}$,
otherwise set $\underline{\mu}_{3}=\mu_{3}$.

\texttt{\qquad{}\qquad{}$\mbox{(S-3.5)}:$ }If $\overline{\mu}_{3}-\underline{\mu}_{3}>\epsilon$,
then go to \texttt{$\mbox{(S-3.2)}$}.

\texttt{\qquad{}$\mbox{(S-2.4)}:$} Using (\ref{eq:gen_proximal_BR}),
if $\sum_{k=1}^{N}w_{k2}p_{k}^{\star}<\alpha_{2}$ then set $\overline{\mu}_{2}=\mu_{2}$,
otherwise set $\underline{\mu}_{2}=\mu_{2}$.

\texttt{\qquad{}$\mbox{(S-2.5)}:$ }If $\overline{\mu}_{2}-\underline{\mu}_{2}>\epsilon$,
then go to \texttt{$\mbox{(S-2.2)}$}.

\texttt{$\mbox{(S.4)}:$} Using (\ref{eq:gen_proximal_BR}), if $\sum_{k=1}^{N}w_{k1}p_{k}^{\star}<\alpha_{1}$
then set $\overline{\mu}_{1}=\mu_{1}$, otherwise set $\underline{\mu}_{1}=\mu_{1}$.

\texttt{$\mbox{(S.5)}:$} If $\overline{\mu}_{1}-\underline{\mu}_{1}>\epsilon$
then go to \texttt{$\mbox{(S.2)}$}; otherwise \texttt{STOP}.

\label{alg:multiple_bisection_gen_prox_BR}\end{algo}

\subsection{The MIMO case\label{sub:The-MIMO-case_revised}\vspace{-0.1cm}}

Let us consider now MIMO games. Given $\mathcal{G}_{\texttt{{mimo}}}$,
let us assume w.l.o.g. that all the matrices $\mathbf{Q}_{i}$ in
the game have the same dimensions $n_{T}\times n_{T}$. It follows
from Proposition \ref{Lemma_min_principle} (see also Example \ref{Example_MIMO_complex_min_princ})
that $\mathcal{G}_{\texttt{{mimo}}}$ is equivalent to the complex
VI$(\mathcal{P}^{\texttt{{mimo}}},\mathbf{F^{\mathbb{C}}})$ where
$\mathbf{F^{\mathbb{C}}}(\mathbf{Q})\triangleq\left(\mathbf{F}_{i}^{\mathbb{C}}(\mathbf{Q})\right)_{i=1}^{I}:\mathcal{P}^{\texttt{{mimo}}}\rightarrow\mathbb{C}^{Qn_{T}\times n_{T}}$,
with each 
\begin{equation}
\mathbf{F}_{i}^{\mathbb{C}}(\mathbf{Q})\triangleq-\left(\nabla_{\mathbf{Q}_{i}}R_{i}(\mathbf{Q})\right)^{\ast}=-\mathbf{H}_{ii}^{H}\left(\mathbf{R}_{n_{i}}+\sum\limits _{j=1}^{I}\mathbf{H}_{ij}\mathbf{Q}_{j}\mathbf{H}_{ij}^{H}\right)^{^{-1}}\mathbf{H}_{ii}.\label{eq:F_MIMO_G}
\end{equation}
According to Proposition \ref{monotonicity_complexVI}, the monotonicity/P
properties of $\mathbf{F}^{\mathbb{C}}(\mathbf{Q})$ on $\mathcal{P}^{\texttt{{mimo}}}$
are related to the properties of the augmented Jacobian matrix $\mathbf{JF}^{\mathbb{C}}(\mathbf{Q})$.
To obtain sufficient conditions easy to be checked, let us introduce
the following $I\times I$ matrix $\boldsymbol{{\Upsilon}}_{\mathbf{F}^{\mathbb{C}}}^{\text{\texttt{{mimo}}}}\in\mathbb{R}^{I\times I}$
obtained by properly ``condensing'' $\mathbf{JF}^{\mathbb{C}}(\mathbf{Q})$
(see Appendix \ref{sec:Proof-of-Proposition_MIMOGame}), and defined
as (we implicitly assume that all the channel matrices $\mathbf{H}_{ii}$
are full column rank): 
\begin{equation}
[\,\boldsymbol{{\Upsilon}}_{\mathbf{F}^{\mathbb{C}}}^{\text{\texttt{{mimo}}}}\,]_{ij}\,\,\triangleq\,\left\{ \begin{array}{ll}
1 & \mbox{if \ensuremath{i=j}}\\[5pt]
-{\displaystyle \rho\left(\mathbf{H}_{ii}^{\dagger H}\mathbf{H}_{ij}^{H}\mathbf{H}_{ij}\mathbf{H}_{ii}^{\dagger}\right)\cdot{\textsf{INNR}}{}_{ij}} & \mbox{if \ensuremath{i\neq j}},
\end{array}\right.\label{eq:Upsilon_F_Q}
\end{equation}
where $\mathbf{A}^{\dagger}$ denotes the Moore\textendash{}Penrose
pseudoinverse of $\mathbf{A}$ {[}since $\mathbf{H}_{ii}$ are full
column rank, we have $\mathbf{H}_{ii}^{\dagger}=\left(\mathbf{H}_{ii}^{H}\mathbf{H}_{ii}\right)^{-1}\mathbf{H}_{ii}^{H}${]},
and ${\textsf{INNR}}{}_{ij}$ is defined as\vspace{-0.2cm} 
\begin{equation}
{\textsf{INNR}}{}_{ij}\triangleq\dfrac{\rho\left(\mathbf{R}_{n_{i}}+\sum\limits _{j=1}^{I}P_{j}\mathbf{H}_{ij}\mathbf{H}_{ij}^{H}\right)}{\lambda_{\mbox{least}}(\mathbf{R}_{n_{i}})}.\label{eq:INNR_MIMO}
\end{equation}

Using Proposition \ref{Lemma_min_principle} and Proposition \ref{monotonicity_complexVI},
we obtain the following characterization for $\mathcal{G}_{\texttt{{mimo}}}$.\vspace{-0.2cm}

\begin{proposition}\label{Prop:MIMO_Game} Given the complex NEP
\emph{$\mathcal{G}_{\texttt{{mimo}}}=\left\langle \mathcal{\mathcal{P}^{\texttt{{mimo}}}},\,(R_{i})_{i=1}^{I}\right\rangle $},
the following hold.

\vspace{-0.3cm}

\begin{description}
\item [{(a)}] \emph{$\mathcal{G}_{\texttt{{mimo}}}$} is equivalent to
the complex \emph{$\text{{VI}}(\mathcal{P}^{\texttt{{mimo}}},\,\mathbf{\mathbf{F}^{\mathbb{C}}})$},
which has a nonempty and compact solution set;\vspace{-0.3cm}

\item [{(b)}] Suppose that \emph{$\boldsymbol{{\Upsilon}}_{\mathbf{F}^{\mathbb{C}}}^{\text{\texttt{{mimo}}}}$}
is positive semidefinite. Then $\mathbf{F}^{\mathbb{C}}$ is monotone
on \emph{$\mathcal{P}^{\texttt{{mimo}}}$}; therefore, \emph{$\mathcal{G}_{\texttt{{mimo}}}$
is a monotone complex NEP;} \vspace{-0.3cm}

\item [{(c)}] Suppose that \emph{$\boldsymbol{{\Upsilon}}_{\mathbf{F}^{\mathbb{C}}}^{\text{\texttt{{mimo}}}}$}
is a P-matrix (positive definite matrix). Then $\mathbf{F}^{\mathbb{C}}$
is a uniformly P-function (strongly monotone function) on \emph{$\mathcal{P}^{\texttt{{mimo}}}$};
therefore, \emph{$\mathcal{G}_{\texttt{{mimo}}}$} is a complex P$_{\boldsymbol{{\Upsilon}}}$
NEP and thus has a unique NE. \vspace{-0.2cm}
\end{description}
\end{proposition}

\begin{proof} See Appendix \ref{sec:Proof-of-Proposition_MIMOGame}.
\end{proof}

\begin{corollary} \label{Cor_SF_G_for_Pmat_MIMO} The matrix \emph{$\boldsymbol{{\Upsilon}}_{\mathbf{F}^{\mathbb{C}}}^{\text{\texttt{{mimo}}}}$}
is a P (or positive definite) matrix if one (or both) of the following
two sets of conditions are satisfied: for some $\mathbf{w}=(w_{i})_{i=1}^{I}>\mathbf{0}$,
\begin{equation}
\begin{array}{rl}
\mbox{\mbox{\emph{Low}\,\ \emph{received}\,\ \emph{MUI}\emph{:}}\quad} & \dfrac{{1}}{w_{i}}\,\sum\limits _{j\neq i}w_{j}\,\left\{ {\displaystyle \rho\left(\mathbf{H}_{ii}^{\dagger H}\mathbf{H}_{ij}^{H}\mathbf{H}_{ij}\mathbf{H}_{ii}^{\dagger}\right)\cdot\textsf{\emph{INNR}}{}_{ij}}\right\} <1,\quad\forall i=1,\ldots,I\medskip\\
\mbox{\emph{Low}\,\ \emph{generated}\,\ \emph{MUI}\emph{:\quad}} & \dfrac{{1}}{w_{j}}\,\sum\limits _{i\neq j}w_{i}\,\left\{ {\displaystyle \rho\left(\mathbf{H}_{ii}^{\dagger H}\mathbf{H}_{ij}^{H}\mathbf{H}_{ij}\mathbf{H}_{ii}^{\dagger}\right)\cdot\textsf{\emph{INNR}}{}_{ij}}\right\} <1,\quad\forall j=1,\ldots,I\vspace{-0.3cm}
\end{array}\label{eq:Low_TxTx_MUI_MIMO}
\end{equation}
\end{corollary}\medskip{}

It is interesting to observe that conditions for $\mathbf{F}^{\mathbb{C}}$
to be a uniformly P-function on $\mathcal{P}^{\texttt{{mimo}}}$ are
the natural generalization of those obtained for $\mathcal{G}_{\texttt{{siso}}}$
to be a P$_{\boldsymbol{{\Upsilon}}}$ game; they thus have the same
physical interpretation, for which we refer the reader to Sec. \ref{sub:Game G_the-SISO}.
Based on that, in Proposition \ref{Prop:MIMO_Game} we used the same
terminology as in Definition \ref{Def_monotone_NEP}, namely: $\mathcal{G}_{\texttt{{mimo}}}$
is a complex P$_{\boldsymbol{{\Upsilon}}}$ NEP if $\boldsymbol{{\Upsilon}}_{\mathbf{F}^{\mathbb{C}}}^{\text{\texttt{{mimo}}}}$
is a P matrix, whereas is a complex monotone NEP if $\boldsymbol{{\Upsilon}}_{\mathbf{F}^{\mathbb{C}}}^{\text{\texttt{{mimo}}}}$
is a semidefinite matrix. For these two classes of NEPs we can devise
distributed algorithms having the same convergence properties and
features of those developed in Sec. \ref{sub:Best-response-decomposition-algos}
and Sec. \ref{sub:Proximal-decomposition-algorithms_monotone_VI}
for real P$_{\boldsymbol{{\Upsilon}}}$ and monotone NEPs, respectively.
The main results are briefly listed next; proofs are based on the
same techniques used to prove Lemma \ref{Lemma_min_principle} and
Proposition \ref{Prop:MIMO_Game}, and thus are omitted. \medskip{}

\noindent \textbf{The case of }\emph{$P{}_{\boldsymbol{\Upsilon}}$}\textbf{
NEP} $\mathcal{G}_{\texttt{{mimo}}}$ (low-interference regime). When
the matrix $\boldsymbol{{\Upsilon}}_{\mathbf{F}^{\mathbb{C}}}^{\text{\texttt{{mimo}}}}$
is a P matrix, the unique NE of the game can be computed using Algorithm
\ref{async_best-response_algo} on $\mathcal{G}_{\texttt{{mimo}}}$,
as stated in the next theorem.

\begin{theorem}\label{Th:AIWFA_MIMO} Suppose that \emph{$\mathcal{G}_{\texttt{{mimo}}}$}
is a complex P$_{\boldsymbol{{\Upsilon}}}$ NEP. Then, any sequence
$\{\mathbf{Q}^{(n)}\}_{n=0}^{\infty}$ generated by Algorithm \ref{async_best-response_algo}
applied to \emph{$\mathcal{G}_{\texttt{{mimo}}}$} converges to the
unique NE of the NEP, for any given updating schedule of the players
satisfying assumptions A1\emph{)}-A3\emph{)}.\end{theorem}

The algorithm has the same desired features as the one obtained for
the SISO case; see Sec. \ref{sub:Game G_the-SISO}.\textbf{ }\medskip{}

\noindent \textbf{The case of monotone NEP} $\mathcal{G}_{\texttt{{mimo}}}$
(medium-interference regime).\textbf{ }When $\boldsymbol{{\Upsilon}}_{\mathbf{F}^{\mathbb{C}}}^{\text{\texttt{{mimo}}}}\succeq\mathbf{0}$,
$\mathcal{G}_{\texttt{{mimo}}}$ is a monotone complex NEP; in the
presence of multiple solutions, we need to choose a merit function
assessing the quality of a NE of $\mathcal{G}_{\texttt{{mimo}}}$;
similarly to the SISO case, we consider the overall interference generated
by all the SUs: 
\begin{equation}
\phi(\mathbf{Q})\triangleq\sum_{i=1}^{I}w_{i}\,\sum_{j\neq i}\text{{tr}}\left(\mathbf{H}_{ij}\mathbf{Q}{}_{j}\mathbf{H}_{ij}^{H}\right),\label{outer_merit_function_MIMO}
\end{equation}
where $w_{i}$'s are given positive weights. Building on the equivalence
between the game $\mathcal{G}_{\texttt{{mimo}}}$ and the $\text{{VI}}(\mathcal{P}^{\texttt{{mimo}}},$
$\,\mathbf{F}^{\mathbb{C}})$ (Lemma \ref{Lemma_min_principle}),
the NE selection problem based on the merit function $\phi$ becomes:
\begin{equation}
\begin{array}{ll}
\limfunc{minimize}\limits _{\mathbf{Q}} & \phi(\mathbf{Q})\\[5pt]
\text{subject to} & \mathbf{Q}\,\in\,\sol_{\mathbb{C}}(\mathcal{P}^{\texttt{{mimo}}},\,\mathbf{F}^{\mathbb{C}}).
\end{array}\label{eq:NEP constrained opt_MIMO}
\end{equation}

A solution of (\ref{eq:NEP constrained opt_MIMO}) can be computed
in a distributed way using Algorithm \ref{algo2} applied to (\ref{eq:NEP constrained opt_MIMO});
the interested reader can find the lengthy details in \cite[Sec. 6.2]{Scutari-Facchinei-Pang-Palomar_IT_PI}.
The convergence result is stated in the following theorem. 

\begin{theorem}\label{convergence_NE_sel-MIMO} Given the optimization
problem \eqref{eq:NEP constrained opt_MIMO}, suppose that: i) \emph{$\mathcal{G}_{\texttt{{mimo}}}$}
is a complex monotone  NEP; ii) $\{\varepsilon^{(n)}\}$ is such that
$\varepsilon^{(n)}\rightarrow0$ and $\sum_{n=0}^{\infty}\,\varepsilon^{(n)}=\infty$;
and iii) $\tau$ is chosen so that \emph{$\boldsymbol{{\Upsilon}}_{\mathbf{F}^{\mathbb{C}}}^{\text{\texttt{{mimo}}}}+\tau\,\mathbf{I}$}
is a P matrix. Then, the sequence $\{\mathbf{Q}^{(n)}\}_{n=0}^{\infty}$
generated by Algorithm \ref{algo2} applied to (\ref{eq:NEP constrained opt_MIMO})
has a limit point and every such point is a solution of \eqref{eq:NEP constrained opt_MIMO}\emph{.}\end{theorem}
A sufficient condition for matrix $\boldsymbol{{\Upsilon}}_{\mathbf{F}^{\mathbb{C}}}^{\text{\texttt{{mimo}}}}+\tau\,\mathbf{I}$
in Theorem \ref{convergence_NE_sel-MIMO} to be P is 
\begin{equation}
\tau\,>\,{\displaystyle {\max_{1\leq i\leq I}}\,\left\lbrace {\displaystyle {\sum_{j\neq i}}{\displaystyle \rho\left(\mathbf{H}_{ii}^{\dagger H}\mathbf{H}_{ij}^{H}\mathbf{H}_{ij}\mathbf{H}_{ii}^{\dagger}\right)\cdot{\textsf{INNR}}{}_{ij}}}\right\rbrace -1.}\label{eq:bounds_MIMO}
\end{equation}

\section{Numerical Results\label{sub:Numerical-Results}\vspace{-0.1cm}}

In this section, we compare some of the proposed algorithms solving
NEP $\mathcal{G}_{\texttt{{siso}}}$ in (\ref{eq:max_rate_game})
and $\mathcal{G}_{\texttt{{mimo}}}$ in (\ref{eq:MIMO_Game}) in terms
of achievable rates and convergence speed. We also compare the performance
of our distributed schemes with those achievable computing a stationary
solution of the related sum-rate optimization problem (under the same
power and interference constraints). For the latter, we considered
the pricing-based algorithm proposed in \cite{ScutariPalomarFacchineiPang_NETGCOP11}
for the SISO case, and the solution methods proposed in \cite{ScutariPalomarFacchineiPang_NETGCOP11,KimGiannakisIT11}
for the MIMO setting; we slight modified the schemes in \cite{ScutariPalomarFacchineiPang_NETGCOP11,KimGiannakisIT11}
in order to include the interference constraints we have proposed
in this paper. Algorithms in \cite{ScutariPalomarFacchineiPang_NETGCOP11,KimGiannakisIT11}
are the benchmark methods for this kind of problems. Finally, we contrast
our best-response schemes with gradient-response ones \cite{Yin-Shanbhag-Mehta_TAC10}.

\noindent \textbf{Example \#1}\textbf{\emph{ }}\textbf{(NE selection
vs. no selection)}. In Fig. \ref{fig1a}, we plot the SUs' sum-rate
$\sum_{i=1}^{I}r_{i}(\mathbf{p})$ versus inner iterations (i.e.,
number of overall iterations required by the algorithm to converge),
achieved by the following algorithms applied to the NEP $\mathcal{G}_{\texttt{{siso}}}$:
i) The Jacobi version of the proximal decomposition algorithm described
in Algorithm \ref{alg:PDA} (green line curve); ii) the Jacobi version
of Algorithm \ref{algo_power_bi_level} (blue line curve), where $\phi(\mathbf{p})$
is given by (\ref{outer_merit_function}); iii) the same Algorithm
\ref{algo_power_bi_level} applied to (\ref{eq:NEP constrained opt}),
where $\phi(\mathbf{p})$ in (\ref{outer_merit_function}) is replaced
by $-\phi(\mathbf{p})$ (red line curve); and iv) the Jacobi Dynamic
Pricing-based Algorithm (DPA) reaching a stationary solution of the
sum-rate maximization problem (black line curve) \cite{ScutariPalomarFacchineiPang_NETGCOP11}.
The choice of the merit function $-\phi(\mathbf{p})$ leads to the
selection of the NE solution that maximizes the overall MUI in the
system, which provides a benchmark of the sum-rate variability and
an estimate of the worst-case performance over the set of the Nash
equilibria of the game. Any best-response solution involved in the
optimization problems is computed using the waterfilling-like expression
introduced in Sec. \ref{sub:Preliminary-results:-Efficient_proximalBR}.

The above algorithms are tested in the following setting. We considered
a CR network modeled as a Gaussian parallel IC, composed of $I=25$
active users and two PUs; all the users are randomly placed within
an hexagonal cell; the channels of all the links are simulated as
FIR filter of order $L=10$, where each tap is a zero mean complex
Gaussian random variable with variance equal to $1/L$; the number
of carriers is $N=128$.. We focused on two scenarios, namely low
interference regime (corresponding to $\boldsymbol{{\mathbf{\Upsilon}}}_{\mathbf{G}}$
being a P matrix) and high interference regime (corresponding to $\mathbf{J}\mathbf{G}_{\text{{low}}}$
being positive semidefinite). The thermal noise variance ${\sigma_{i}}^{2}(k)$
is set to one for all $k$ and $i$, and the Signal-to-Noise-Ratio
(SNR) of each user is set to $\text{SNR}_{i}\triangleq10\,\log_{10}\left(P_{i}/{\sigma_{i}}^{2}(k)\right)=5$dB
for all $i$ and $k$. In the interference-temperature limit constraints,
for the sake of simplicity, we set the same interference thresholds
$\boldsymbol{{\alpha}}_{i}=\alpha\,\mathbf{1}$ for all the SUs, with
$\alpha=10^{-3}$ (this choice of $\alpha$ is such that the power
budget constraints of the SUs are not active at any optimal solution).
All the algorithms are initialized by the same starting point, chosen
randomly in the set $\mathcal{P}^{\,\texttt{{siso}}}$, and are terminated
when the Euclidean norm of the error in two consecutive iterations
becomes smaller than $10^{-6}$. In the PTDA, we chose the center
$\bar{{\mathbf{p}}}$ of the regularization randomly in $\mathcal{P}^{\,\texttt{{siso}}}$,
$\tau=3.5$, and $\varepsilon^{(n)}=\varepsilon^{(0)}/(1+10\, n)$,
where $\varepsilon^{(0)}=0.5$; in all the algorithms, the termination
criterion of the inner loop, if any, is the same as the outer loop.
The above choice of the free parameters is the result of some preliminary
tests; however we remark that the proposed algorithm has been observed
to be robust against the variation of such parameters.

The following comments are in order from Fig. \ref{fig1a}. In the
case of multiple NE, the sum-rate performance of the network can vary
significantly over the set of the NE; the relative sum-rate gap between
the ``worst'' and ``best'' NE is more than $90\%$. As expected,
Algorithm \ref{algo_power_bi_level} outperforms Algorithm \ref{alg:PDA},
which validates the use of criterion (\ref{outer_merit_function})
in choosing the NE. Moreover the sum-rate loss with respect to the
DPA is very limited, and more than acceptable if one takes into account
that, to be implemented, the DPA requires a significant signaling
among the users at each iteration. There are scenarios where such
a signaling exchange is not feasible (e.g., when the users are heterogeneous
systems that are not willing to cooperate); in all these cases Algorithm
\ref{algo_power_bi_level} is a good candidate. When the NE of the
game is unique {[}$\boldsymbol{{\Upsilon}}_{\mathbf{G}}$ is a P matrix{]},
as expected, both Algorithms \ref{alg:PDA} and \ref{algo_power_bi_level}
converge to the same sum-rate solution. Interestingly, this solution
seems to coincide also with the one achieved by the DPA. Finally,
note that our algorithms converge quite fast. 

Fig. \ref{fig1b} shows the average performance of algorithms i)-iii)
considered in Fig. \ref{fig1a}. We plotted the average sum-rate versus
the $\text{SNR}\triangleq P$, with $P_{i}=P$ and ${\sigma_{i}^{2}(k)}=1$
for all $i$ and $k$, achievable at the NE reached by Algorithm \ref{alg:PDA}
and Algorithm \ref{algo_power_bi_level}. The curves are averages
over $5000$ random channel realization chosen so that the $\mathbf{J}\mathbf{G}_{\text{{low}}}\succeq\mathbf{0}$.
The rest of the parameters are the same as in Fig. \ref{fig1a}. Fig.
\ref{fig1b} confirms the superior performance of Algorithm \ref{algo_power_bi_level}
with respect to Algorithm \ref{alg:PDA} that does not perform any
equilibrium selection. Finally, it is worth observing that\emph{ }projection-response
algorithms proposed in\emph{ }\cite{Yin-Shanbhag-Mehta_TAC10} and
\cite{Facchinei-Pang_FVI03,Konnov_VI_book} cannot be used to solve
the P$_{\boldsymbol{{\Upsilon}}}$ NEP $\mathcal{G}_{\texttt{{siso}}}$
(unless it is also monotone), even if the game has a unique NE.
\begin{figure*}[t]
\center\subfigure[Sum-rate of the SUs versus inner iterations \label{fig1a}]{\includegraphics[scale=0.25]{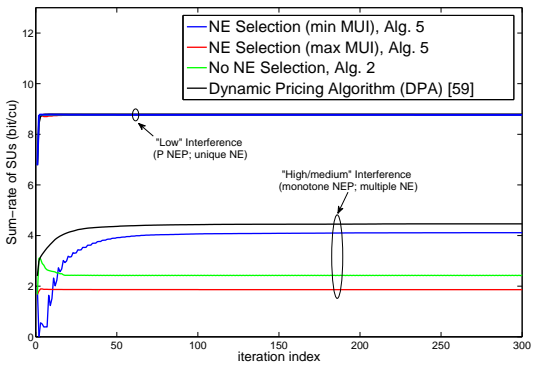}}\hspace{1cm}
\subfigure[Average sum-rate of the SUs versus the SNR \label{fig1b}]{\includegraphics[scale=0.27]{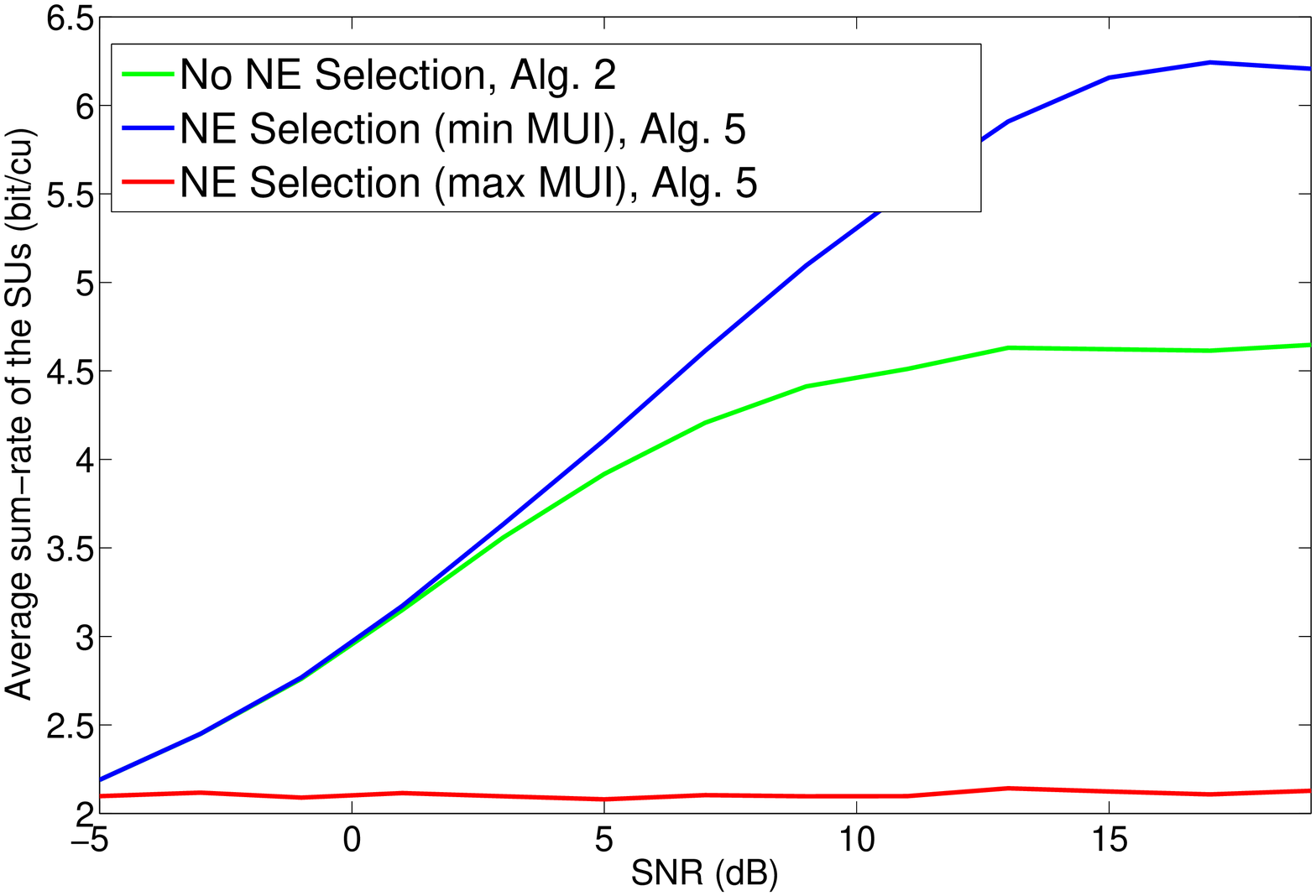}}

\caption{{\small{\label{fig1} Comparison of distributed algorithms solving
$\mathcal{G}_{\texttt{{siso}}}$.}}}
\end{figure*}

\noindent \textbf{Example \#2}\textbf{\emph{ }}\textbf{(Comparison
with gradient-response algorithms for monotone VIs)}\emph{.} In Fig.
\ref{fig:fig_intro} we compare some algorithms solving the game $\mathcal{G}_{\texttt{{siso}}}$
under the monotonicity assumption ($\mathcal{G}_{\texttt{{siso}}}$
is a monotone NEP); the setup is the same of Fig. \ref{fig1}. More
specifically, we plot the SUs rates versus the iteration index achieved
by Algorithm 2 and the Iterative Tikhonov Algorithm recently proposed
in \cite{Yin-Shanbhag-Mehta_TAC10} for solving monotone VIs. In the
latter algorithm we chose the variable step-size sequences $\gamma_{n}=n^{-0.4}$
and $\delta_{n}=n^{-0.49}$ so that (sufficient) conditions for the
convergence given in \cite[Prop. 15.1]{Konnov_VI_book} are satisfied
(we use the same notation as in \cite{Yin-Shanbhag-Mehta_TAC10};
see therein for the details). The figure clearly shows that our best-response-based
scheme converges in a very few iterations, whereas the gradient-response
algorithm needs much more iterations (two orders of magnitude more)
to reach comparable performance. The same convergence properties as
in Fig. \ref{fig:fig_intro} has been experienced for all the channel
realizations we simulated. 
\begin{figure}[H]
\vspace{-0.5cm}\center\includegraphics[height=7.5cm]{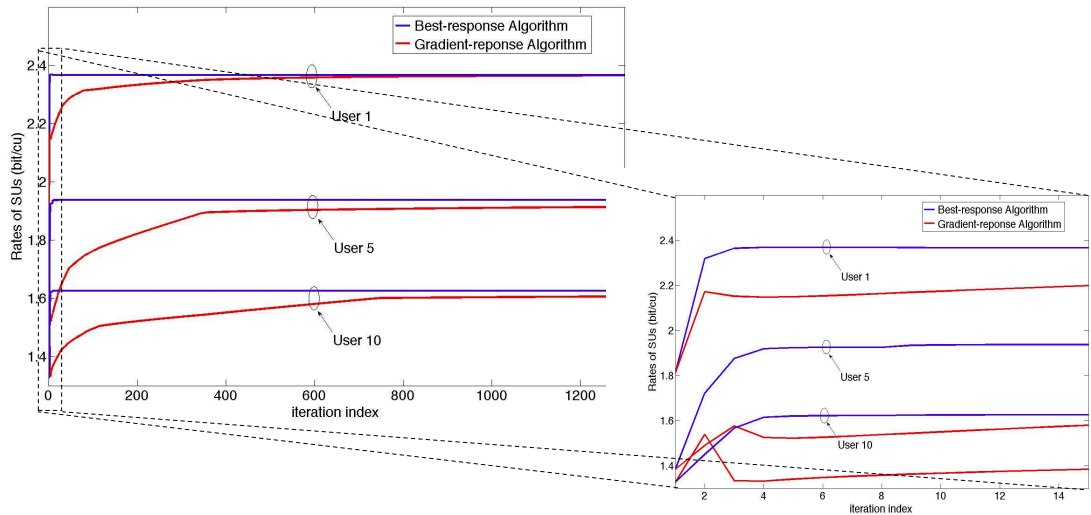}\vspace{-0.4cm}
\caption{{\small{Typical behavior of gradient-response versus best-response
algorithms solving the monotone NEP $\mathcal{G}_{\texttt{{siso}}}$:
rate of three out 25 users versus the iterations achievable by the
gradient-response algorithm \cite{Yin-Shanbhag-Mehta_TAC10} (red-line
curves) and the simultaneous best-response algorithm described in
Algorithm 1 (blue-line curves).}}\label{fig:fig_intro} }
\end{figure}

\noindent \textbf{Example \#3}\textbf{\emph{ }}\textbf{(NE selection
vs. stationary solutions: the MIMO case)}\emph{.} We compare here
some of the proposed algorithms in the MIMO setting. We consider the
same scenario as in Fig. \ref{fig1a}, with the only difference that
now all the transceivers are equipped with three antennas and there
are $I=5$ active SUs. The channels are MIMO frequency-selective (the
order of the channels is $L=10$ and the number of subcarriers is
$N=128$) and are generated in order to guarantee that the matrix
$\boldsymbol{{\Upsilon}}_{\mathbf{F}^{\mathbb{C}}}$ in (\ref{eq:Upsilon_F_Q})
is positive semidefinite, resulting thus in a monotone NEP $\mathcal{G}^{\texttt{{mimo}}}$.
Soft average power shaping interference constraints are imposed to
each SU along the direction of the primary transmitters; all the interference
threshold are assumed to be equal and set to $I_{pi}^{\limfunc{ave}}=10^{-3}$.
The best-response of each user cannot be computed in closed form (unless
the proximal regularization is not included in the objective function),
but can be efficiently computed using any nonlinear programming solvers
(each player's optimization problem is strongly convex). In Fig. \ref{fig5_GMIMO}
we plot the sum-rate versus the inner iteration index achieved by
the Jacobi version of Algorithm 2 (green-line curve), Algorithm \ref{algo2}
based on the merit function $\phi(\mathbf{Q})$ defined in (\ref{outer_merit_function_MIMO})
(blue line curve); and iii) the Gauss-Seidel based algorithm proposed
in \cite{KimGiannakisIT11} to compute stationary solutions of the
sum-rate problem (we adapted the algorithm in \cite{KimGiannakisIT11}
including the interference constraints in the feasible set of the
optimization problem). Fig. \ref{fig5_GMIMO} shows the trade-off
between performance and signaling that is achievable by the three
algorithms: Algorithm 7 implementing a NE selection leads to better
sum-rates than Algorithm 2 at the cost of almost the same signaling
among the SUs of classical MIMO IWFAs (a constant price depending
on the cross-channel matrices needs to be computing by each SU before
running the algorithms); higher sum-rates can be achieved using algorithm
in \cite{KimGiannakisIT11} but at the cost of more signaling among
the SUs (note that in the MIMO case, the scheme \cite{KimGiannakisIT11}
requires the SUs to exchange matrix informations at each iteration).
\begin{figure}[h]
\vspace{-0.3cm}\center \includegraphics[height=6.5cm]{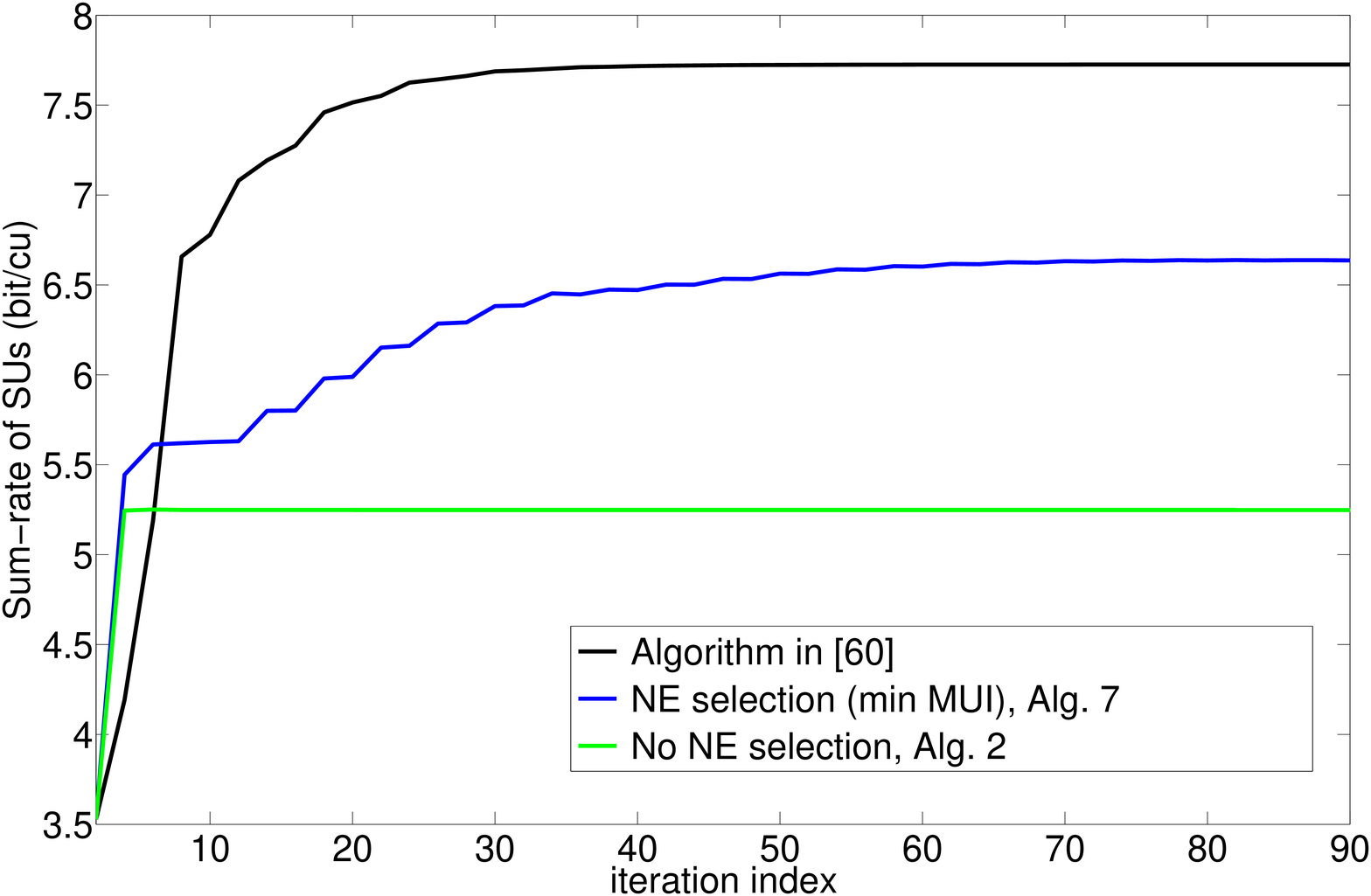}\vspace{-0.6cm}
 \caption{{\small{Comparison of distributed algorithms solving }}$\mathcal{G}_{\texttt{{mimo}}}${\small{:
Sum-rate of the SUs versus inner iterations.}}}

{\small{\label{fig5_GMIMO} }} \vspace{-0.6cm}
\end{figure}

\section{Conclusions\label{sec:Conclusions}\vspace{-0.2cm}}

In this paper, we have proposed a novel method based on VIs suitable
to study and solve general real or complex player-convex NEPs, having
(possibly) multiple solutions. The proposed framework has many desirable
new features, such as: i) it can be applied to real and complex NEPs
having no specific structure and for which the best-response mapping
is not available in closed form or unique; ii) the algorithms proposed
for computing a NE converge under mild conditions that do not imply
the uniqueness of the equilibrium; and iii) in the presence of multiple
equilibria, one can control the quality of the computed solution by
guaranteeing convergence to the ``best'' NE (according to some prescribed
criterion), at the cost of some signaling among the players. These
features make the proposed algorithms applicable to a variety of scenarios
in different fields; the choice of one scheme with respect to the
other will depend on the trade-off between signaling and performance
that the users are willing to exchange/achieve. \textcolor{black}{The
analysis of algorithms for complex NEPs hinges on the definition of
the VI problem in the complex domain; this new class of VIs along
with their properties are introduced and studied for the fist time
in this paper. }

Finally, to have suitable case studies, we applied the proposed framework
to solve some novel NEPs modeling various resource allocation problems
in SISO/MIMO CR systems. The resulting distributed best-response algorithms
were shown to converge even when current schemes proposed in the literature
for related problems fail. Numerical results showed the superiority
of our (distributed) approach with respect to plain noncooperative
solutions as well as good performance with respect to centralized
solutions.\vspace{-0.3cm}

\appendix

\section*{Appendix\vspace{-0.3cm}}

\section{(Partitioned) Variational Inequalities\label{sec:A-Theory-of_partitioned_VI}\vspace{-0.3cm}}

The simplest way to see a VI is as a generalization of the minimum
principle for convex optimization problems, which is recalled next.
Consider a convex optimization problem (in the minimization form),
whose objective function $f:\mathcal{Q}\mapsto\mathbb{R}$ is convex
and continuously differentiable on the feasible set $\mathcal{Q}\subseteq\mathbb{R}^{n}$,%
\footnote{When we say that a (vector-valued) function is continuous or continuously
differentiable on a closed set we mean that the function is so on
an open set containing the closed set.%
} which is a convex and closed subset of $\mathbb{R}^{n}$. A point
$\mathbf{x}^{\star}\in\mathcal{Q}$ is an optimal solution of the
optimization problem if and only if\vspace{-0.2cm} 
\begin{equation}
\left(\mathbf{x}-\mathbf{x}^{\star}\right)^{T}\nabla f\left(\mathbf{x}^{\star}\right)\geq0,\qquad\forall\mathbf{x}\in\mathcal{Q}.\label{eq:minimum_principle}
\end{equation}

The VI problem is a generalization of the minimum principle (\ref{eq:minimum_principle})
where the gradient $\nabla f$ is substituted by a general real-valued
vector function $\mathbf{F}$. More formally, we have the following.
Let $\mathcal{Q}\subseteq\mathbb{R}^{n}$ be a nonempty, closed, and
convex set and let $ $$\mathbf{F}:$ $\mathcal{Q}\to$ $\mathbb{R}$$^{n}$
be a vector-valued real function. The VI ($\mathcal{Q},$ $\mathbb{\mathbf{F)}}$
is the problem of finding a feasible point $\mathbf{x}^{\star}\in\mathcal{Q}$
such that \cite[Def. 1.1.1]{Facchinei-Pang_FVI03} 
\begin{equation}
\left(\mathbf{x}-\mathbf{x}^{\star}\right)^{T}\mathbf{F}\left(\mathbf{x}^{\star}\right)\geq0,\qquad\forall\mathbf{x}\in\mathcal{Q}.\label{eq:VI_def-1}
\end{equation}

\noindent The set of solutions to this problem is denoted by $\limfunc{SOL}(\mathcal{Q},\mathbf{F})$. 

Several standard problems in nonlinear programming, game theory, and
nonlinear analysis can be naturally formulated as a VI problem; many
examples can be found in \cite[Ch. 1]{Facchinei-Pang_FVI03}, \cite{Scutari-Palomar-Facchinei-Pang_SPMag10},
and \cite{Konnov_VI_book}. Below we summarize some known facts and
definitions about VI.\vspace{-0.2cm}

\subsection{Solution analysis\label{sub:Solution-analysis}}

In order to present results about the existence and the structure
of the solution set of a VI, we introduce some function classes. \vspace{-0.3cm}

\begin{definition} \label{Def_monotonicity} \label{Def_P_properties}
A mapping $\mathbf{F}:\mathcal{Q}\mathcal{\rightarrow}\mathbb{R}^{n}$,
with $\mathcal{Q}$ closed and convex, is said to be \vspace{-0.2cm}
\begin{description}
\item [{{{\rm}(i)}}] \emph{monotone} on $\mathcal{Q}$ if for all
pairs $\mathbf{x}$ and $\mathbf{y}$ in \emph{$\mathcal{Q}$},\vspace{-0.2cm}
\begin{equation}
(\,\mathbf{x}-\mathbf{y}\,)^{T}(\,\mathbf{F}(\mathbf{x})-\mathbf{F}(\mathbf{y})\,)\,\geq0;\label{eq:def_monotonicity}
\end{equation}

\item [{{(ii)}}] \emph{strictly monotone} if for all pairs $\mathbf{x}\neq\mathbf{y}$
in \emph{$\mathcal{Q}$ the inequality in (\ref{eq:def_monotonicity})
is strict;} 
\item [{{{\rm}(iii)}}] \emph{strongly monotone} on $\mathcal{Q}$
if there exists a constant $c_{\limfunc{sm}}>0$ such that for all
pairs $\mathbf{x}$ and $\mathbf{y}$ in \emph{$\mathcal{Q}$}, 
\begin{equation}
{\displaystyle (\,\mathbf{x}-\mathbf{y}\,)^{T}(\,\mathbf{F}(\mathbf{x})-\mathbf{F}(\mathbf{y})\,)\,\geq\, c_{\limfunc{sm}}\,\|\,\mathbf{x}-\mathbf{y}\,\|^{2}.}\label{eq:strongly monotone function}
\end{equation}
The constant $c_{\limfunc{sm}}$ is called strong monotonicity constant. 
\end{description}
If we assume a Cartesian product structure, i.e. $\mathbf{F}=(\mathbf{F}_{i}(\mathbf{x}))_{i=1}^{I}$
and $\mathcal{Q}=\prod_{i}\mathcal{Q}_{i}$, then $\mathbf{F}$ is
said to be 
\begin{description}
\item [{{{\rm}(iv)}}] \textsl{\emph{a}}\textsl{ P}$_{0}$\textsl{
function} on $\mathcal{Q}$ if for all pairs of distinct tuples $\mathbf{x}=(\mathbf{x}_{i})_{i=1}^{I}$
and $\mathbf{y}=(\mathbf{y}_{i})_{i=1}^{I}$ in $\mathcal{Q}$, an
index $i$ exists such that $\mathbf{x}_{i}\neq\mathbf{y}_{i}$ and
\begin{equation}
(\,\mathbf{x}_{i}-\mathbf{y}_{i}\,)^{T}\left(\,\mathbf{F}_{i}(\mathbf{x})-\mathbf{F}_{i}(\mathbf{y})\,\right)\,\geq\,0;\label{eq:P_o_def}
\end{equation}

\item [{{{\rm}(v)}}] a \emph{P function }on $\mathcal{Q}$ if for
all pairs of distinct tuples $\mathbf{x}=(\mathbf{x}_{i})_{i=1}^{I}$
and $\mathbf{y}=(\mathbf{y}_{i})_{i=1}^{I}$ in $\mathcal{Q}$, \emph{the
inequality in (\ref{eq:P_o_def}) is strict;} 
\item [{{{\rm}(vi)}}] a \emph{uniformly P function} on $\mathcal{Q}$
if there exists a constant $c_{\limfunc{uP}}>0$ such that for all
pairs $\mathbf{x}=(\mathbf{x}_{i})_{i=1}^{I}$ and $\mathbf{y}=(\mathbf{y}_{i})_{i=1}^{I}$
in \emph{$\mathcal{Q}$}, 
\begin{equation}
{\displaystyle {\max_{1\leq\, i\,\leq Q}}\,(\,\mathbf{x}_{i}-\mathbf{\mathbf{y}}_{i}\,)^{T}(\,\mathbf{F}_{i}(\mathbf{x})-\mathbf{F}_{i}(\mathbf{y})\,)\,\geq\, c_{\limfunc{uP}}\,\|\,\mathbf{x}-\mathbf{y}\,\|^{2}.}\label{eq:Uniformly P function}
\end{equation}
The constant $c_{\limfunc{uP}}$ is called uniformly P constant. 
\end{description}
If a function $\mathbf{F}$ enjoys one of the properties above, we
will also say that the corresponding VI $(\mathcal{Q},\mathbf{F})$
enjoys the property (i.e., if $\mathbf{F}$ is monotone, we say that
the VI is monotone, etc...).

\end{definition}

Note that in the case of affine functions, $\mathbf{F}(\mathbf{x})=\mathbf{M}\mathbf{x}+\mathbf{b}$,
there is no difference between strict monotonicity and strongly monotonicity,
and the uniform P property coincides with the P property. Monotonicity
properties play in the VI realm the same role that convex functions
play in optimization. In fact, we recall that a differentiable function
$f$ is convex, strictly, strongly convex on a convex set $\mathcal{Q}$
if and only if its gradient is monotone, strictly, strongly monotone
on $\mathcal{Q}$. The P properties can be viewed as an extension
of the monotonicity properties tailored to the possible partitioned
structure of the VI; when the partitioned VI has only one block, i.e.,
$I=1$, the P properties collapse to the corresponding monotonicity
properties. In Fig. \ref{fig:Monotonicity} we summarize in a pictorial
way some well established relations between these various classes
along with some of their consequences. Theorem \ref{Theo_existence_uniqueness}
provides instead a formal statement of some existence and uniqueness
results that will be used throughout the paper. 
\begin{figure}[H]
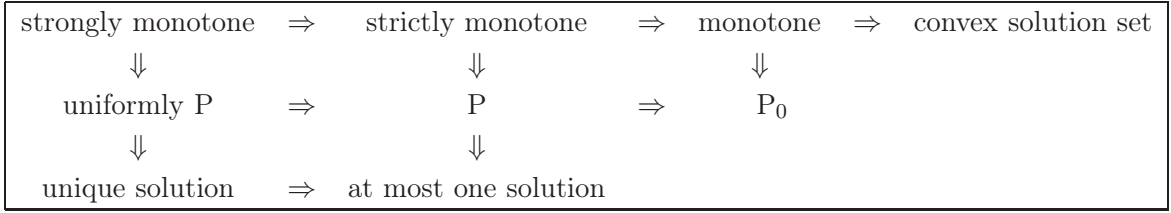

\centering{}%
\begin{tabular}{|ccccccc|}
\hline 
strongly monotone  & $\Rightarrow$  & strictly monotone  & $\Rightarrow$  & monotone  & $\Rightarrow$  & convex solution set \tabularnewline
$\Downarrow$  &  & $\Downarrow$  &  & $\Downarrow$  &  & \tabularnewline
uniformly P  & $\Rightarrow$  & P  & $\Rightarrow$  & \hspace{0.5pc} P$_{0}$  &  & \tabularnewline
$\Downarrow$  &  & $\Downarrow$  &  &  &  & \tabularnewline
unique solution  & $\Rightarrow$  & at most one solution  &  &  &  & \tabularnewline
\hline 
\end{tabular}\caption{Monotonicity and their consequences on VIs.\label{fig:Monotonicity}}
\end{figure}

\begin{theorem}\label{Theo_existence_uniqueness} Given the VI $(\mathcal{Q},\,\mathbf{F})$,
suppose that $\mathcal{Q}$ is closed and convex and $\mathbf{F}$
is continuous on $\mathcal{Q}$. The following statements hold: \vspace{-0.2cm}
\begin{description}
\item [{{(a)}}] The \emph{VI}$(\mathcal{Q},\mathbf{F})$\emph{ }has a
(possibly empty) closed solution set. If $\ensuremath{\mathcal{Q}}$
is bounded, the solution set is nonempty and thus compact \cite[Cor. 2.2.5]{Facchinei-Pang_FVI03};
\vspace{-0.2cm}
\item [{{(b)}}] If $\ensuremath{\mathbf{F}}$ is monotone on $\mathcal{Q}$,
then the \emph{VI}$(\mathcal{Q},\mathbf{F})$ has a (possibly empty)
convex solution set \cite[Th. 2.3.5]{Facchinei-Pang_FVI03}; \vspace{-0.2cm}
\item [{{(c)}}] If $\ensuremath{\mathbf{F}}$ is strictly monotone on
$\mathcal{Q}$, then the\emph{ VI}$(\mathcal{Q},\mathbf{F})$\emph{
}has at most one solution \cite[Th. 2.3.3(a)]{Facchinei-Pang_FVI03};
the same conclusion holds if the\emph{ VI}$(\mathcal{Q},\mathbf{F})$
is partitioned and $\ensuremath{\mathbf{F}}$ is a P function on $\mathcal{Q}$
\cite[Prop. 3.5.10(a)]{Facchinei-Pang_FVI03}; \vspace{-0.2cm}
\item [{{(d)}}] If $\mathbf{F}$ is strongly monotone on $\mathcal{Q}$,
then the VI$(\mathcal{Q},\mathbf{F})$ has a unique solution \cite[Th. 2.3.3(b)]{Facchinei-Pang_FVI03};
the same conclusion holds if the\emph{ VI}$(\mathcal{Q},\mathbf{F})$
is partitioned and $\ensuremath{\mathbf{F}}$ is a uniformly P function
on $\mathcal{Q}$ \cite[Prop. 3.5.10(b)]{Facchinei-Pang_FVI03}. 
\end{description}
\end{theorem}

\section{Proof of Proposition \ref{monotonicity}\label{sec:Proof-of-Proposition_monotonicity}}

Because of space limitation, we prove only (e); the proof of (a)-(d)
follows similar ideas.

Given $\mathbf{x}\triangleq(\mathbf{x}_{i})_{i=1}^{I},\mathbf{y}\triangleq(\mathbf{y}_{i})_{i=1}^{I}\in\mathcal{Q}$,
with $\mathbf{x}\neq\mathbf{y}$, let define the univariate continuously
differentiable function $\Phi_{i}\,:\,[0,1]\ni t\,\mapsto\mathbb{R}$
as $\Phi_{i}(t)\triangleq\left(\mathbf{x}_{i}-\mathbf{y}_{i}\right)^{T}\left(\mathbf{F}_{i}\left(t\,\mathbf{x}+(1-t)\,\mathbf{y}\right)\right).$
Then, by the mean-value theorem it follows that there exists some
$\bar{t}{}_{i}\in(0,1)$ such that, denoting by $\mathbf{z}_{\bar{{t}}_{i}}=\bar{{t}}_{i}\mathbf{x}+\left(1-\bar{{t}}_{i}\right)\mathbf{y}$,
we have
\[
\hspace{-0.5cm}\left(\mathbf{x}_{i}-\mathbf{y}_{i}\right)^{T}\left(\mathbf{F}_{i}(\mathbf{x})-\mathbf{F}_{i}(\mathbf{y})\right)=\left.\frac{d\Phi_{i}(t)}{dt}\right|_{t=\bar{{t}}_{i}}=\left(\mathbf{x}_{i}-\mathbf{y}_{i}\right)^{T}\,\sum\limits _{j=1}^{I}\mbox{J}_{i}\mathbf{F}_{j}(\mathbf{z}_{\bar{t}_{i}})\,\left(\mathbf{x}_{j}-\mathbf{y}_{j}\right)
\]
 
\begin{eqnarray}
 & \geq & \left(\mathbf{x}_{i}-\mathbf{y}_{i}\right)^{T}\,\mbox{J}_{i}\mathbf{F}_{i}(\mathbf{z}_{\bar{t}_{i}})\,\left(\mathbf{x}_{i}-\mathbf{y}_{i}\right)-\left|\left(\mathbf{x}_{i}-\mathbf{y}_{i}\right)^{T}\,\sum\limits _{j\neq i}\mbox{J}_{i}\mathbf{F}_{j}(\mathbf{z}_{\bar{t}_{i}})\,\left(\mathbf{x}_{j}-\mathbf{y}_{j}\right)\right|\nonumber \\
 & \geq & \left\Vert \mathbf{C}_{i}^{-T}\left(\mathbf{x}_{i}-\mathbf{y}_{i}\right)\right\Vert ^{2}\,\alpha_{i}^{\min}\,-\left\Vert \mathbf{C}_{i}^{-T}\left(\mathbf{x}_{i}-\mathbf{y}_{i}\right)\right\Vert \sum\limits _{j\neq i}\beta_{ij}^{\max}\left\Vert \mathbf{C}_{j}^{-T}\left(\mathbf{x}_{j}-\mathbf{y}_{j}\right)\right\Vert ,\label{eq:Unif_P_property_Jac_last}
\end{eqnarray}
where last inequality follows from the definition of $\alpha_{i}^{\min}$
and $\beta_{ij}^{\max}$ as given in (\ref{eq:def_alpha_and_beta_Jac}),
and $\mathbf{C}_{i}$'s are the set of nonsingular matrices coming
from (\ref{eq:def_alpha_and_beta_Jac}). Introducing $\mathbf{e}=(e_{i})_{i=1}^{I}$,
with each $e_{i}\triangleq\left\Vert \mathbf{C}_{i}^{-1}(\mathbf{x}_{i}-\mathbf{y}_{i})\right\Vert $,
and using the definition of $\mathbf{\boldsymbol{\Upsilon}}_{\mathbf{F}}$,
(\ref{eq:Unif_P_property_Jac_last}) can be written as 
\[
\left(\mathbf{x}_{i}-\mathbf{y}_{i}\right)^{T}\left(\mathbf{F}_{i}(\mathbf{x})-\mathbf{F}_{i}(\mathbf{y})\right)\geq e_{i}\left(\mathbf{\boldsymbol{\Upsilon}}_{\mathbf{F}}\mathbf{e}\right)_{i}.
\]
 Since $\mathbf{\boldsymbol{\Upsilon}}_{\mathbf{F}}$ is a P-matrix,
it follows from \cite[Th. 3.3.4(b)]{Cottle-Pang-Stone_bookLCP92}
that $c(\mathbf{\boldsymbol{\Upsilon}}_{\mathbf{F}})\triangleq\min_{\left\Vert \mathbf{y}\triangleq(y_{i})_{i=1}^{I}\right\Vert ^{2}=1}\max_{i}\, y_{i}(\mathbf{\boldsymbol{\Upsilon}}_{\mathbf{F}}\,\mathbf{y})_{i}>0$.
Therefore, we have 
\begin{equation}
\max_{i=1,\ldots,I}\left(\mathbf{x}_{i}-\mathbf{y}_{i}\right)^{T}\left(\mathbf{F}_{i}(\mathbf{x})-\mathbf{F}_{i}(\mathbf{y})\right)\geq\max_{i=1,\ldots,I}\, e_{i}(\mathbf{\boldsymbol{\Upsilon}}_{\mathbf{F}}\mathbf{e})_{i}\geq\dfrac{c(\mathbf{\boldsymbol{\Upsilon}}_{\mathbf{F}})}{\max_{i=1,\ldots,I}\lambda_{\max}(\mathbf{C}_{i}^{T}\mathbf{C}_{i})}\,{\left\Vert \mathbf{x}-\mathbf{y}\right\Vert ^{2}}.\label{eq:P-constant_Gamma_low}
\end{equation}
Finally the lower bound in (\ref{eq:mod_F_unifP}) can be proved using
\cite[Ex. 5.11.19]{Cottle-Pang-Stone_bookLCP92}, from which one can
readily obtain $c(\mathbf{\boldsymbol{\Upsilon}}_{\mathbf{F}})\geq{\displaystyle \delta(\boldsymbol{\Upsilon}_{\mathbf{F}})/(I\cdot(1+\zeta(\boldsymbol{\Upsilon}_{\mathbf{F}})/\delta(\boldsymbol{\Upsilon}_{\mathbf{F}}))^{2(I-1)})}$.
This completes the proof of (e). \hfill $\square$.\vspace{-0.1cm}

\section{Proof of Theorem \ref{Theo-async_best-response_NEP}\label{Appendix:Proof-of-Theorem-convergence-best_response}\vspace{-0.1cm}}

It is sufficient to show that, under the assumptions of the theorem,
the best-response mapping is a block-contraction (i.e., a contraction
with respect to some block-maximum norm); the latter property indeed
guarantees that conditions of the asynchronous convergence theorem
in \cite[Prop. 2.1]{Bertsekas_Book-Parallel-Comp} are satisfied by
the asynchronous best-response algorithm described in Algorithm \ref{async_best-response_algo}.

We introduce first the following norms: Given the set of nonsingular
matrices $\mathbf{C}_{i}$'s coming from (\ref{eq:def_alpha_and_beta_Jac}),
the block-maximum norm on $\mathbb{\mathbb{R}}^{n},$ with $n=n_{1}+\ldots+n_{I},$
is defined as 
\begin{equation}
\left\Vert \mathbf{x}\right\Vert _{\text{block}}^{\mathbf{w}}\triangleq\max_{i=1,\ldots,I}\frac{\left\Vert \mathbf{C}_{i}^{-1}\mathbf{x}_{i}\right\Vert }{c_{i}},\quad\mbox{for}\quad\mathbf{x}=(\mathbf{x}_{i})_{i=1}^{I}\in\mathbb{R}^{n},\label{block_max_weight_norm}
\end{equation}
where $\left\Vert \mathbf{\bullet}\right\Vert $ is the Euclidean
norm on $\mathbb{R}^{n_{i}}$ and $\mathbf{c}\triangleq[c_{1},\ldots,c_{I}]^{T}>\mathbf{0}$
is any positive weight vector; the (weighted) maximum norm on $\mathbb{\mathbb{R}}^{I}$
is defined as (see, e.g., \cite{Bertsekas_Book-Parallel-Comp}) 
\begin{equation}
\left\Vert \mathbf{x}\right\Vert _{\infty,\text{vec}}^{\mathbf{c}}\triangleq\max_{i=1,\ldots,I}\frac{\left\vert x_{i}\right\vert }{c_{i}},\quad\mbox{for}\quad\mathbf{x\in\mathbb{R}}^{I};\label{weighted_infinity_vector_norm}
\end{equation}
and the\emph{ }matrix norm $\left\Vert \mathbf{\cdot}\right\Vert _{\infty,\text{mat}}^{\mathbf{c}}$
on $\mathbb{R}^{I\times I}$ induced by $\left\Vert \cdot\right\Vert _{\infty,\text{vec}}^{\mathbf{c}}$
is given by 
\begin{equation}
\left\Vert \mathbf{A}\right\Vert _{\infty,\text{mat}}^{\mathbf{c}}\triangleq\max_{i}\frac{1}{c_{i}}\dsum\limits _{j=1}^{I}\left\vert [\mathbf{A}]_{ij}\right\vert c_{j},\quad\mbox{for}\quad\mathbf{A\in\mathbb{R}}^{I\times I}.\label{H_max_weight_norm}
\end{equation}
Under the under the P property of $\boldsymbol{{\Upsilon}}_{\mathbf{F}}$,
let us introduce the best-response mapping 
\begin{equation}
\mathcal{B}(\mathbf{x})\triangleq(\mathcal{B}_{i}(\mathbf{x}_{-i}))_{i=1}^{I}:\mathcal{Q}\ni\mathbf{x}\mapsto\mathcal{Q},\qquad\mbox{with}\qquad\mathcal{B}_{i}(\mathbf{x}_{-i})\triangleq\mbox{argmin}_{\mathbf{x}_{i}\in\mathcal{Q}_{i}}f_{i}(\mathbf{x}_{i},\,\mathbf{x}_{-i}).\label{eq:best-response_mapping}
\end{equation}
The contraction properties of $\mathcal{B}(\mathbf{x})$ are given
in the following, where $\mathbf{\boldsymbol{\Gamma}_{\mathbf{F}}}$
is defined in (\ref{eq:Upsilon_matrix}) {[}note that the P property
of $\boldsymbol{{\Upsilon}}_{\mathbf{F}}$, implies the strong convexity
of each $f_{i}(\cdot,\mathbf{x}_{-i})$ for any $\mathbf{x}_{-i}\in\mathcal{Q}_{-i}$,
and thus $\alpha_{i}^{\min}>0$ for all $i$; hence, $\boldsymbol{\Gamma}_{\mathbf{F}}$
is well-defined{]}.

\begin{proposition} \label{Proposition_contraction_best-response}
Suppose that $\boldsymbol{{\Upsilon}}_{\mathbf{F}}$ in (\ref{eq:Upsilon_matrix})
is a P matrix. Then, the best-response mapping $\mathcal{B}(\mathbf{x})$
is a block-contraction, i.e., there there exists some $\mathbf{c}>\mathbf{0}$
such that 
\begin{equation}
\left\Vert \mathcal{B}(\mathbf{x})-\mathcal{B}(\mathbf{y})\right\Vert _{\text{\emph{block}}}^{\mathbf{c}}\leq\alpha_{\mathbf{c}}\left\Vert \mathbf{x}-\mathbf{y}\right\Vert _{\text{\emph{block}}}^{\mathbf{c}}\qquad\forall\mathbf{x},\mathbf{y}\in\mathcal{Q}\label{eq:contraction_best_response}
\end{equation}
with $\alpha_{\mathbf{c}}\triangleq\left\Vert \mathbf{\boldsymbol{\Gamma}_{\mathbf{F}}}\right\Vert _{\infty,\text{\emph{mat}}}^{\mathbf{c}}<1$.

\end{proposition}\begin{proof}For any two vectors $\mathbf{x},\mathbf{y}\in\mathcal{Q}$
we have by the minimum principle 
\begin{equation}
\begin{array}{c}
\left(\mathbf{z}_{i}-\mathcal{B}_{i}(\mathbf{x})\right)^{T}\nabla_{\mathbf{x}_{i}}f_{i}\left(\mathcal{B}_{i}(\mathbf{x}),\,\mathbf{x}_{-i}\right)\geq0\qquad\forall\mathbf{z}_{i}\in\mathcal{Q}_{i},\quad i=1,\ldots,I,\\
\left(\mathbf{z}_{i}-\mathcal{B}_{i}(\mathbf{y})\right)^{T}\nabla_{\mathbf{x}_{i}}f_{i}\left(\mathcal{B}_{i}(\mathbf{y}),\,\mathbf{y}_{-i}\right)\geq0\qquad\forall\mathbf{z}_{i}\in\mathcal{Q}_{i},\quad i=1,\ldots,I.
\end{array}\label{eq:variationa_principle_in_the_best-response}
\end{equation}
Substituting $\mathbf{z}_{i}=\mathcal{B}_{i}(\mathbf{y})$ into the
former inequality and $\mathbf{z}_{i}=\mathcal{B}_{i}(\mathbf{x})$
into the latter, adding the two resulting inequalities we obtain with
$\hat{\mathbf{z}}{}_{i}\triangleq t_{i}(\mathcal{B}_{i}(\mathbf{y}),\mathbf{y}_{-i})+(1-t_{i})(\mathcal{B}_{i}(\mathbf{x}),\mathbf{x}_{-i})$
and some $t_{i}\in(0,1)$:\vspace{-0.2cm}

\begin{eqnarray}
0 & \leq & \left(\mathcal{B}_{i}(\mathbf{x})-\mathcal{B}_{i}(\mathbf{y})\right)^{T}\left(\nabla_{\mathbf{x}_{i}}f_{i}\left(\mathcal{B}_{i}(\mathbf{y}),\,\mathbf{y}_{-i}\right)-\nabla_{\mathbf{x}_{i}}f_{i}\left(\mathcal{B}_{i}(\mathbf{x}),\,\mathbf{x}_{-i}\right)\right)\nonumber \\
 & = & \left(\mathcal{B}_{i}(\mathbf{x})-\mathcal{B}_{i}(\mathbf{y})\right)^{T}\nabla_{\mathbf{x}_{i}\mathbf{x}_{i}}^{2}f_{i}\left(\hat{\mathbf{z}}\right)\left(\mathcal{B}_{i}(\mathbf{y})-\mathcal{B}_{i}(\mathbf{x})\right)+\left(\mathcal{B}_{i}(\mathbf{x})-\mathcal{B}_{i}(\mathbf{y})\right)^{T}\sum_{j\neq i}\nabla_{\mathbf{x}_{i}\mathbf{x}_{j}}^{2}f_{i}\left(\hat{\mathbf{z}}\right)\left(\mathbf{y}_{j}-\mathbf{x}_{j}\right)\label{eq:bet-response_ineq_row2}
\end{eqnarray}
where the equality follows from the application of the main-value
theorem to the univariate, differentiable, scalar function 
\begin{equation}
[0,\,1]\ni t_{i}\mapsto\left(\mathcal{B}_{i}(\mathbf{x})-\mathcal{B}_{i}(\mathbf{y})\right)^{T}\nabla_{\mathbf{x}_{i}}f_{i}\left(t_{i}(\mathcal{B}_{i}(\mathbf{y}),\mathbf{y}_{-i})+(1-t_{i})(\mathcal{B}_{i}(\mathbf{x}),\mathbf{x}_{-i})\right).\label{eq:scalar_diff_func_for_mean_value_theo}
\end{equation}
Using the definition of $\alpha_{i}^{\min}$ and $\beta_{ij}^{\max}$
as given in (\ref{eq:def_alpha_and_beta_Jac}), we deduce from the
inequality in (\ref{eq:bet-response_ineq_row2}) 
\begin{equation}
\left\Vert \mathcal{B}_{i}(\mathbf{x})-\mathcal{B}_{i}(\mathbf{y})\right\Vert _{i}\alpha_{i}^{\min}\leq\sum_{j\neq i}\beta_{ij}^{\max}\left\Vert \mathbf{x}_{j}-\mathbf{y}_{j}\right\Vert _{j},\qquad i=1,\ldots,I,\label{eq:bet-response_ineq_row3}
\end{equation}
(the inequality in (\ref{eq:bet-response_ineq_row2}) is trivially
satisfied if $\left\Vert \mathcal{B}_{i}(\mathbf{x})-\mathcal{B}_{i}(\mathbf{y})\right\Vert _{i}=0$).
Introducing the vectors $\mathbf{e}_{\mathcal{B}}\triangleq(e_{\mathcal{B}_{i}})_{i=1}^{I}$
and $\mathbf{e}\triangleq(e_{i})_{i=1}^{I}$ with $e_{\mathcal{B}_{i}}\triangleq\left\Vert \mathcal{B}_{i}(\mathbf{x})-\mathcal{B}_{i}(\mathbf{y})\right\Vert _{i}$
and $e_{i}=\left\Vert \mathbf{x}_{i}-\mathbf{y}_{i}\right\Vert _{i}$,
using the definition of $\boldsymbol{\Gamma}_{\mathbf{F}}$ in (\ref{eq:Upsilon_matrix}),
and the fact that $\alpha_{i}^{\min}>0$ for all $i$, (\ref{eq:bet-response_ineq_row3})
can be written in vectorial form as 
\begin{equation}
\mathbf{e}_{\mathcal{B}}\leq\boldsymbol{{\Gamma}}_{\mathbf{F}}\,\mathbf{e},\qquad\forall\mathbf{x},\mathbf{y\in\mathcal{Q}}.\label{eq:bet-response_ineq_row4}
\end{equation}
It follows from (\ref{eq:bet-response_ineq_row4}) that, for any given
$\mathbf{c}>0$, we have 
\begin{equation}
\left\Vert \mathcal{B}(\mathbf{x})-\mathcal{B}(\mathbf{y})\right\Vert _{\text{{block}}}^{\mathbf{c}}=\left\Vert \mathbf{e}_{\mathcal{B}}\right\Vert _{\infty,\text{{vec}}}^{\mathbf{c}}\leq\left\Vert \boldsymbol{{\Gamma}}_{\mathbf{F}}\right\Vert _{\infty,\text{{mat}}}^{\mathbf{c}}\left\Vert \mathbf{e}\right\Vert _{\infty,\text{{vec}}}^{\mathbf{c}}=\left\Vert \boldsymbol{{\Gamma}}_{\mathbf{F}}\right\Vert _{\infty,\text{{mat}}}^{\mathbf{c}}\left\Vert \mathbf{x}-\mathbf{y}\right\Vert _{\text{{block}}}^{\mathbf{c}}\label{eq:bet-response_ineq_row5}
\end{equation}
which proves the inequality in (\ref{eq:contraction_best_response}).
To complete the proof we need to show that $\left\Vert \boldsymbol{{\Gamma}}_{\mathbf{F}}\right\Vert _{\infty,\mbox{mat}}^{\mathbf{c}}<1$
for some $\mathbf{c}>0$. Invoking \cite[Lemma 13.14]{Cottle-Pang-Stone_bookLCP92}
and \cite[Cor. 6.1]{Bertsekas_Book-Parallel-Comp}, we obtain the
desired result: 
\begin{equation}
\boldsymbol{{\Upsilon}}_{\mathbf{F}}\mbox{ is a P-matrix}\qquad\Leftrightarrow\qquad\exists\,\bar{{\mathbf{c}}}>0\quad\mbox{such that }\left\Vert \mathbf{\boldsymbol{{\Gamma}}_{\mathbf{F}}}\right\Vert _{\infty,\text{mat}}^{\bar{{\mathbf{c}}}}<1.\label{eq:equivalence_P-property_contraction}
\end{equation}

\end{proof}\vspace{-0.3cm}

\section{Proof of Proposition \ref{Proposition_G_and_G_tau_x}\label{Proof-of-Proposition_G_and_G_tau_x}}

Since $\mathcal{G}$ is a monotone NEP, the VI associated with the
NEP $\mathcal{G}_{\tau,\mathbf{y}}-$the VI$(\mathcal{Q},\mathbf{F}+\tau(\mathbf{I}-\mathbf{y}))$
with $\mathbf{F}=(\nabla_{\mathbf{x}_{i}}f_{i})_{i=1}^{I}-$is strongly
monotone on $\mathcal{Q}$; it follows by Theorem \ref{Theo_existence_uniquenessNE}(d)
that $\mathcal{G}_{\tau,\mathbf{y}}$ has a unique NE for any given
$\tau>0$ and $\mathbf{y}\in\mathbb{R}^{n}$. Let us denote such a
unique NE by $\mathbf{S}_{\tau}(\mathbf{y})\triangleq\sol(\mathcal{Q},\mathbf{F}+\tau(\mathbf{I}-\mathbf{y}))$.

\noindent \emph{Necessity}: Let $\mathbf{x}^{\star}\in\mathcal{Q}$
be a NE of the monotone NEP $\mathcal{G}$. By Proposition \ref{VI_reformulation},
$\mathbf{x}^{\star}\in\sol(\mathcal{Q},\mathbf{F})$; then the following
hold: $\forall\mathbf{x}\in\mathcal{Q}$, 
\begin{align*}
\left(\mathbf{x}-\mathbf{x}^{\star}\right)^{T}\mathbf{F}(\mathbf{x}^{\star}) & \geq0\,\,\,\Leftrightarrow\,\,\,\left(\mathbf{x}-\mathbf{x}^{\star}\right)^{T}\left(\mathbf{F}(\mathbf{x}^{\star})+\tau(\mathbf{x}^{\star}-\mathbf{x}^{\star})\right)\geq0\,\,\,\Rightarrow\,\,\,\mathbf{x}^{\star}=\sol\left(\mathcal{Q},\mathbf{F}+\tau(\mathbf{I}-\mathbf{x}^{\star})\right)=\mathbf{S}_{\tau}(\mathbf{x}^{\star}),
\end{align*}
implying that $\mathbf{x}^{\star}$ is the unique NE of $\mathcal{G}_{\tau,\mathbf{x}^{\star}}$.

\noindent \emph{Sufficiency}: Let $\mathbf{x}^{\star}$ be a NE of
$\mathcal{G}_{\tau,\mathbf{x}^{\star}}$. Then, we have $\mathbf{x}^{\star}=\mathbf{S}_{\tau}(\mathbf{x}^{\star})$,
which leads to the desired result: $\forall\mathbf{x}\in\mathcal{Q}$,
\begin{align*}
\left(\mathbf{x}-\mathbf{S}_{\tau}(\mathbf{x}^{\star})\right)^{T}\left(\mathbf{F}\left(\mathbf{S}_{\tau}(\mathbf{x}^{\star})\right)+\tau(\mathbf{S}_{\tau}(\mathbf{x}^{\star})-\mathbf{x}^{\star})\right)\geq0\,\,\,\Leftrightarrow\,\,\,\left(\mathbf{x}-\mathbf{x}^{\star}\right)^{T}\mathbf{F}(\mathbf{x}^{\star}) & \geq0.
\end{align*}

\section{Proof of Theorem \ref{the:mod}\label{sec:Proof-of-Theorem_NE_Selection}}

To prove the theorem we hinge on the theory of VIs. We preliminary
observe that the game $\mathcal{G}$ is equivalent to the VI$(\mathcal{Q},\mathbf{F})$,
with $\mathbf{F}=(\nabla_{\mathbf{x}_{i}}f_{i})_{i=1}^{I}$ (Proposition
\ref{VI_reformulation}); $\sol(\mathcal{Q},\mathbf{f})$ is thus
also the solution set of the VI, i.e., $\sol(\mathcal{Q},\mathbf{F})=\sol(\mathcal{Q},\mathbf{f})$.
Moreover, still invoking Proposition \ref{VI_reformulation}, we have
that the game $\mathcal{G}_{\tau,\varepsilon^{(n)},\mathbf{x}^{(n)}}$
in Step 2 of Algorithm \ref{algo2} is equivalent to the VI$(\mathcal{Q},\mathbf{F}^{(n)})$,
where 
\begin{equation}
\mathbf{F}^{(n)}(\mathbf{x})\triangleq\mathbf{F}(\mathbf{x})+\varepsilon^{(n)}\nabla\phi(\mathbf{x})+\tau\,(\mathbf{x}-\mathbf{x}^{(n)}).\label{F_n}
\end{equation}
Observe that, under the assumptions of the theorem, $\mathbf{F}^{(n)}$
is strongly monotone {[}Definition \ref{Def_monotonicity}(iii){]},
implying that the VI$(\mathcal{Q},\mathbf{F}^{(n)})$ has a unique
solution {[}Theorem \ref{Theo_existence_uniqueness}(c){]} and thus
$\mathbf{x}^{(n+1)}$ in Step 2 of Algorithm \ref{algo2} is well
defined at each iteration. Moreover, denoting by $\mathcal{S}$ the
solution set of \eqref{eq:VI-C min}, assumptions of the theorem,
ensure that $\mathcal{S}$ is nonempty, bounded, and convex. Let us
introduce for each $n$, 
\[
\delta^{(n)}\,\triangleq\,\dfrac{1}{2}\,\text{{dist}}(\mathbf{x}^{(n)},\mathcal{S})\,=\dfrac{1}{2}\,\|\mathbf{x}^{(n)}-{P}_{\mathcal{S}}(\mathbf{x}^{(n)})\|^{2},
\]
where $P_{\mathcal{S}}(\mathbf{y})\triangleq\text{argmin}_{\mathbf{x}\in\mathcal{S}}\|\mathbf{x}-\mathbf{y}\|$
denotes the Euclidean projection on the nonempty, closed, and convex
set $\mathcal{S}$. Then, to prove the theorem it suffices to show
that the sequence $\{\delta^{(n)}\}$ converges to zero. Observe first
that, since $\mathbf{x}^{(n+1)}$ at Step 2 is the solution of the
game $\mathcal{G}_{\tau,\varepsilon^{(n)},\mathbf{x}^{(n)}}$$-$the
VI$(\mathcal{K},\mathbf{F}^{(n)})$$-$we get, for any $\mathbf{y}\in\mathcal{Q}$,\vspace{-0.2cm}
\begin{equation}
\left[\,\mathbf{F}(\mathbf{x}^{(n+1)})+\varepsilon^{(n)}\nabla\phi(\mathbf{x}^{(n+1)})\,\right]^{T}(\,\mathbf{y}-\mathbf{x}^{(n+1)}\,)\geq\,\tau\,(\,\mathbf{x}^{(n)}-\mathbf{x}^{(n+1)}\,)^{T}(\,\mathbf{y}-\mathbf{x}^{(n+1)}\,).\label{eq:pass1}
\end{equation}

We can write 
\begin{align}
\begin{array}{lll}
\delta^{(n+1)}-\delta^{(n)} & = & \dfrac{1}{2}\,\|\,\mathbf{x}^{(n+1)}-P_{\mathcal{S}}(\mathbf{x}^{(n+1)})\,\|^{2}-\dfrac{1}{2}\,\|\,\mathbf{x}^{(n)}-P_{\mathcal{S}}(\mathbf{x}^{(n)})\,\|^{2}\smallskip\\
 & \overset{(a)}{\leq} & \dfrac{1}{2}\,\|\mathbf{x}^{(n+1)}-P_{\mathcal{S}}(\mathbf{x}^{(n)})\|^{2}-\dfrac{1}{2}\,\|\,\mathbf{x}^{(n)}-P_{\mathcal{S}}(\mathbf{x}^{(n)})\,\|^{2}\smallskip\\
 & = & -\dfrac{1}{2}\|\mathbf{x}^{(n+1)}-\mathbf{x}^{(n)}\|^{2}+\,(\mathbf{x}^{(n)}-\mathbf{x}^{(n+1)})^{T}(P_{\mathcal{S}}(\mathbf{x}^{(n)})-\mathbf{x}^{(n+1)})\smallskip\\
 & \overset{(b)}{\leq} & -\,\dfrac{1}{2}\,\|\mathbf{x}^{(n+1)}-\mathbf{x}^{(n)}\|^{2}+{\displaystyle {\dfrac{1}{\tau}}\,\mathbf{F}(\mathbf{x}^{(n+1)})^{T}(P_{\mathcal{S}}(\mathbf{x}^{(n)})-\mathbf{x}^{(n+1)})}+\,{\dfrac{\varepsilon^{(n)}}{\tau}}\,\nabla\phi(\mathbf{x}^{(n+1)})^{T}(P_{\mathcal{S}}(\mathbf{x}^{(n)})-\mathbf{x}^{(n+1)})\,\smallskip\\
 & \overset{(c)}{\leq} & -\dfrac{1}{2}\,\|\mathbf{x}^{(n+1)}-\mathbf{x}^{(n)}\|^{2}+\,{\displaystyle {\dfrac{\varepsilon^{(n)}}{\tau}}\,\underbrace{\nabla\phi(\mathbf{x}^{(n+1)})^{T}\,(P_{\mathcal{S}}(\mathbf{x}^{(n)})-\mathbf{x}^{(n+1)})}_{V^{(n+1)}}}
\end{array}\label{eq:modd 1}
\end{align}
where: (a) follows readily from the definition of projection; (b)
comes from \eqref{eq:pass1} evaluated at $\mathbf{y}=P_{\mathcal{S}}(\mathbf{x}^{(n)})\in\mathcal{Q}$;
and (c) can be obtained by observing that since $P_{\mathcal{S}}(\mathbf{x}^{(n)})\in\mathcal{S}\subseteq\sol(\mathcal{Q},\mathbf{F})$
and $\mathbf{x}^{(n+1)}\in\mathcal{Q}$, we have $\mathbf{F}(P_{\mathcal{S}}(\mathbf{x}^{(n)}))^{T}\,(\mathbf{x}^{(n+1)}-P_{\mathcal{S}}(\mathbf{x}^{(n)}))\geq0$,
which yields in turn, by the monotonicity of $\mathbf{F}$ {[}Definition
\ref{Def_monotonicity}(i){]}, $\mathbf{F}(\mathbf{x}^{(n+1)})^{T}\,(P_{\mathcal{S}}(\mathbf{x}^{(n)})-\mathbf{x}^{(n+1)})\leq0$.
We now distinguish three cases.

\medskip{}

\noindent \emph{Case 1}: Eventually, $V^{(n+1)}\leq0$.

\smallskip{}

\noindent In this case the nonnegative sequence $\{\delta^{(n)}\}$
is (eventually) non-increasing and therefore convergent. Let us denote
by $n_{0}$ the index from which all $V^{(n)}$ are non positive,
and let us consider $n\geq n_{0}$ without loss of generality. Since
$\mathcal{S}$ is bounded, this implies that also $\{\mathbf{x}^{(n)}\}$
is bounded. Furthermore, it follows from (\ref{eq:modd 1}) that $\{\delta^{(n+1)}-\delta^{(n)}\}$
converges to zero and $\delta^{(n+1)}-\delta^{(n)}\,\leq\,-({1}/{2})\,\|\mathbf{x}^{(n+1)}-\mathbf{x}^{(n)}\|^{2},$
which shows that 
\begin{equation}
\lim_{n\to\infty}\,\|\mathbf{x}^{(n+1)}-\mathbf{x}^{(n)}\|\,=\,0.\label{eq:modd 2}
\end{equation}
Summing (\ref{eq:modd 1}) from $n_{0}$ to $n-1$, we get 
\[
\delta^{(n)}-\delta^{(n_{0})}\,\leq\,\dfrac{1}{\tau}\sum_{i=n_{0}}^{n-1}\varepsilon^{(i)}V^{(i+1)}.
\]
Since $\{\delta^{(n)}\}$ converges and $V^{(n)}\leq0$ and $\sum_{n}\varepsilon^{(n)}=\infty$
in the theorem implies that ${\displaystyle {\limsup_{n\to\infty}}\, V^{(n)}=0}$.
Then, there exists a subsequence $J$ such that 
\begin{equation}
{\displaystyle {\lim_{{{n\in J}\atop {n\to\infty}}}}\, V^{(n)}\,=\,0.}\label{eq:modd 5}
\end{equation}
Since $\{\mathbf{x}^{(n)}\}$ is bounded we may assume, without loss
of generality, that $\lim_{{{n\in J}\atop {n\to\infty}}}\mathbf{x}^{(n)}\,=\,\wt{\mathbf{x}}.$
Note that, since $\mathcal{Q}$ is closed, $\wt{\mathbf{x}}\in\mathcal{Q}$.
We show that actually $\wt{\mathbf{x}}\in\sol(\mathcal{Q},\mathbf{F})$.
If this is not so, there exists a point $\mathbf{y}\in\mathcal{Q}$
such that $\mathbf{F}(\wt{\mathbf{x}})^{T}(\mathbf{y}-\wt{\mathbf{x}})<0$.
Since $\mathbf{x}^{(n)}$ is the solution of the VI$(\mathcal{K},\mathbf{F}^{(n)})$
in Step 2 of the algorithm, we can write, 
\begin{equation}
\left[\,\mathbf{F}(\mathbf{x}^{(n)})+\tau\,(\,\mathbf{x}^{(n)}-\mathbf{x}^{(n-1)}\,)\,\right]^{T}(\,\mathbf{y}-\mathbf{x}^{(n)}\,)+\,\varepsilon^{(n-1)}\,\nabla\phi\left(\mathbf{x}^{(n)}\right)^{T}(\,\mathbf{y}-\mathbf{x}^{(n)}\,)\,\geq\,0.\label{eq:modd 3}
\end{equation}
By continuity, the definition of $\mathbf{y}$, the boundedness of
$\{\mathbf{x}^{(n)}\}$, and (\ref{eq:modd 2}), we have, without
loss of generality (after a suitable renumeration), 
\[
\lim_{{{n\in J}\atop {n\to\infty}}}\mathbf{F}(\mathbf{x}^{(n)})^{T}(\mathbf{y}-\mathbf{x}^{(n)})\,<\,0,\quad\lim_{{{n\in J}\atop {n\to\infty}}}\tau(\mathbf{x}^{(n)}-\mathbf{x}^{(n-1)})^{T}(\mathbf{y}-\mathbf{x}^{(n)})\,=\,0,\quad\lim_{{{n\in J}\atop {n\to\infty}}}\varepsilon^{(n-1)}\nabla\phi(\mathbf{x}^{(n)})^{T}(\,\mathbf{y}-\mathbf{x}^{(n)}\,)\,=\,0,
\]
which contradicts (\ref{eq:modd 3}). Therefore $\wt{\mathbf{x}}\in\sol(\mathcal{Q},\mathbf{F})$.

Thanks to (\ref{eq:modd 2}) we have ${\displaystyle {\lim_{n\in J,n\to\infty}}\,\mathbf{x}^{(n-1)}=\wt{\mathbf{x}}}$.
Therefore, by (\ref{eq:modd 5}) and continuity, we get $\nabla\phi(\wt{\mathbf{x}})^{T}(P_{\mathcal{S}}(\wt{\mathbf{x}})-\wt{\mathbf{x}})=0$.
But the convexity of $\phi$ implies that $\phi(P_{\mathcal{S}}(\wt{\mathbf{x}}))\,\geq\,\phi(\wt{\mathbf{x}})+\nabla\phi(\wt{\mathbf{x}})^{T}(P_{\mathcal{S}}(\wt{\mathbf{x}})-\wt{\mathbf{x}})\,=\,\phi(\wt{\mathbf{x}}),$
thus showing that $\wt{\mathbf{x}}\in\mathcal{S}$. Therefore we get
$\lim_{{{n\in J}\atop {n\to\infty}}}\delta^{(n)}\,=\,0.$ But since
the whole sequence $\{\delta^{(n)}\}$ is convergent, this implies
that the entire sequence $\{\delta^{(n)}\}$ converges to 0, thus
concluding the analysis of Case~1.

\medskip{}

\noindent \emph{Case 2}: The two index sets $J$ and $\bar{J}$ are
both infinite, where $J\,\triangleq\,\left\{ \, n\,\mid\, V^{(n)}>0\,\right\} $
and 
\[
\bar{J}\,\triangleq\,\left\{ \, n\,\in\, J\,\mid\,-\dfrac{1}{2}\,\|\,\mathbf{x}^{(n)}-\mathbf{x}^{(n-1)}\,\|^{2}+{\displaystyle {\dfrac{\varepsilon^{(n-1)}}{\tau}}\, V^{(n)}\,>\,0\,}\right\} .
\]
By (\ref{eq:modd 1}), if $n\in\bar{J}$ it might happen that $\delta^{(n)}>\delta^{(n-1)}$,
while if $n\not\in\bar{J}$ then necessarily $\delta^{(n)}\leq\delta^{(n-1)}$.
Therefore, since $\bar{J}$ is infinite, to prove that $\{\delta^{(n)}\}$
goes to zero it is enough to show that the subsequence $\{\delta^{(n)}\}_{\bar{J}}$
converges to zero. To this end, first observe that for every $n\in\bar{J}$
it holds that 
\begin{equation}
\varepsilon^{(n-1)}V^{(n)}\,>\,\dfrac{\tau}{2}\,\|\mathbf{x}^{(n)}-\mathbf{x}^{(n-1)}\|^{2}.\label{eq:modd 6}
\end{equation}
The sequence $\{\mathbf{x}^{(n)}\}_{\bar{J}}$ is bounded since the
definition of $V^{(n)}$, \eqref{eq:modd 6} and convexity imply $\phi(P_{\mathcal{S}}(\mathbf{x}^{(n-1)}))\,\geq\,\phi(\mathbf{x}^{(n)}).$
But $\phi(P_{\mathcal{S}}(\mathbf{x}^{(n-1)}))$ is the optimal value
of \eqref{eq:VI-C min} and therefore is a number, say $\beta$, that
does not depend on the iteration $n$. Therefore, since $\mathbf{x}^{(n-1)}$
belongs to $\mathcal{Q}$, we have that $\{\mathbf{x}^{(n)}\}_{\bar{J}}$
is bounded. By continuity, also $\{V^{(n)}\}_{\bar{J}}$ is bounded.
Hence, since $\{\varepsilon^{(n)}\}$ converges to $0$, (\ref{eq:modd 6})
yields 
\begin{equation}
\lim_{{{n\in\bar{J}}\atop {n\to\infty}}}\,\|\mathbf{x}^{(n)}-\mathbf{x}^{(n-1)}\|\,=\,0.\label{eq:modd 7}
\end{equation}
Since $\{\mathbf{x}^{(n)}\}_{\bar{J}}$ is bounded, it has limit points.
Let $\tilde{J}\subseteq\bar{J}$ be a subsequence such that $\lim_{{{n\in\tilde{J}}\atop {n\to\infty}}}\,=\,\wt{\mathbf{x}}.$
Reasoning exactly as in Case 1, (the only difference is that instead
of (\ref{eq:modd 2}) we use (\ref{eq:modd 7})), we may deduce that
$\wt{\mathbf{x}}\in\sol(\mathcal{Q},\mathbf{F})$. By continuity,
$\nabla\phi(\wt{\mathbf{x}})^{T}(P_{\mathcal{S}}(\wt{\mathbf{x}})-\wt{\mathbf{x}})\geq0$.
Thus $\wt{\mathbf{x}}\in\mathcal{S}$; hence ${\lim_{{{n\in\tilde{J}}\atop {n\to\infty}}}}\delta^{(n)}=0$.
Since this reasoning can be repeated for every convergent subsequence
of $\{{\mathbf{x}}^{(n)}\}_{\bar{J}}$, we conclude that ${\lim_{{{n\in\bar{J}}\atop {n\to\infty}}}}\,\delta^{(n)}=0$,
thus concluding the analysis of this case.

\medskip{}

\noindent \emph{Case 3}: The index set $J$ is infinite while $\bar{J}$
is finite. In this case, the sequence $\{\delta^{(n)}\}$ is non-increasing
eventually. Therefore $\{\delta^{(n)}\}$ converges, implying that
$\{\mathbf{x}^{(n)}\}$ is bounded, $\{\delta^{(n+1)}-\delta^{(n)}\}$
converges to zero and therefore, by (\ref{eq:modd 1}), also (\ref{eq:modd 2})
holds. At this point, we can proceed exactly as in Case 1 and Case
2 
to prove that $\{\delta^{(n)}\}$ converges to zero, thus concluding
the proof of the theorem. \hfill{}$\square$\vspace{-0.2cm}

\section{Proof of Lemma \ref{Lemma_min_principle}\label{sec:Proof-of-Lemma_min_principle}}

To prove the lemma it is sufficient to show that the first-order Taylor
expansion as given in (\ref{eq:Taylor_coomplex_domain_r2}) holds
for the function $f$; the rest of the proof follows similar steps
as those used to prove the minimum principle for (real-valued) functions
of real variables and thus is omitted; see e.g., \cite{BertsekasNedicOzdaglar_book_convex03}.

Before proving the lemma, it is useful to introduce a real-coordinate
representation of real-valued functions of complex matrices and establish
the connection between standard derivatives of this representation
and the $\mathbb{R}$-derivatives of the original functions of complex
variables. 

The complex space $\mathbb{C}^{n\times m}$ of dimension $n\cdot m$
has a natural structure of a real space $\mathbb{R}^{2nm}$ of dimensions
$2n\cdot m$; this comes readily, e.g., from the following isomorphic
transformation: 
\begin{equation}
\mathbb{C}^{n\times m}\ni\mathbf{Z}\Longleftrightarrow\check{\mathbf{z}}\triangleq\left[\begin{array}{c}
\text{{vec}}\left(\text{{Re}}(\textbf{Z})\right)\\
\text{{vec}}\left(\text{{Im}}(\textbf{Z})\right)
\end{array}\right]\in\mathbb{R}^{2nm}.\label{eq:isomorphic_mapping}
\end{equation}
For the sake of simplicity, in the following, we will denote by $\mathcal{Z}\triangleq\mathbb{C}^{n\times m}$
the original complex space and by $\mathbf{Z}$ the elements of $\mathcal{Z}$;
$\mathcal{R}\triangleq\mathbb{R}^{2nm}$ will be the $2n\cdot m$-dimensional
space of real vectors in the form $\check{\mathbf{z}}$, i.e., 
\begin{equation}
\mathcal{R}\triangleq\left\{ \check{\mathbf{z}}\in\mathbb{R}^{2nm}\,:\,\check{\mathbf{z}}\triangleq\left[\begin{array}{c}
\check{\mathbf{z}}_{R}\\
\check{\mathbf{z}}_{I}
\end{array}\right]\triangleq\left[\begin{array}{c}
\text{{vec}}\left(\text{{Re}}(\textbf{Z})\right)\\
\text{{vec}}\left(\text{{Im}}(\textbf{Z})\right)
\end{array}\right],\,\mbox{for some }\mathbf{Z}\in\mathcal{Z}\right\} ;\label{eq:R_set}
\end{equation}
elements of $\mathcal{R}$ will be denoted by $\check{\mathbf{z}}$,
and partitioned as in (\ref{eq:R_set}). 

Given a real-valued function of complex matrices $f:\mathcal{Z}\rightarrow\mathbb{R}$,
the representation of $f(\mathbf{Z})$ under the isomorphic transformation
(\ref{eq:isomorphic_mapping}) is denoted by $\check{f}(\check{\mathbf{z}})=f(\mathbf{Z})$.
Note that if $f(\mathbf{Z})$ is $\mathbb{R}$-(continuously) differentiable
on $\mathcal{Z}$ then $\check{f}(\check{\mathbf{z}})$ is (continuously)
differentiable on $\mathcal{R}$. Moreover, we can easily establish
the connection between the Jacobian of $\check{f}(\check{\mathbf{z}})$
and the Jacobian and conjugate Jacobian of $f(\mathbf{Z})$ {[}cf.
(\ref{eq:Jacobian_scalar_f}){]}, as shown next. By definition, for
any $\check{\mathbf{z}}\in\mathcal{R}$, the Jacobian of $\check{f}(\check{\mathbf{z}})$
is\vspace{-0.2cm} 
\begin{equation}
D_{\check{\mathbf{z}}}\check{f}(\check{\mathbf{z}})\triangleq\left(\nabla_{\check{\mathbf{z}}}\check{f}(\check{\mathbf{z}})\right)^{T}=\left[\dfrac{{\partial}\check{f}(\check{\mathbf{z}})}{{\partial}\check{\mathbf{z}}_{R}^{T}},\,\dfrac{{\partial}\check{f}(\check{\mathbf{z}})}{{\partial}\check{\mathbf{z}}_{I}^{T}}\right]\triangleq\left[D_{\check{\mathbf{z}}_{R}}\check{f}(\check{\mathbf{z}}),\, D_{\check{\mathbf{z}}_{I}}\check{f}(\check{\mathbf{z}})\right].\label{eq:Jacobian_real}
\end{equation}
Using (\ref{eq:complex_derivatives-1}) and (\ref{eq:Jacobian_scalar_f}),
it is not difficult to see that, for any given $\mathcal{Z}\ni\mathbf{Z}\Longleftrightarrow\check{\mathbf{z}}\in\mathcal{R}$,
the following hold 
\begin{align}
D{}_{\mathbf{Z}}f\left(\mathbf{Z}\right) & \triangleq\dfrac{{\partial}f\left(\mathbf{Z}\right)}{{\partial}\text{{vec}}\left(\mathbf{Z}\right)^{T}}=\dfrac{{1}}{2}\,\left[D_{\check{\mathbf{z}}_{R}}\check{f}(\check{\mathbf{z}})-j\cdot D_{\check{\mathbf{z}}_{I}}\check{f}(\check{\mathbf{z}})\right],\label{eq:Jacobians_f_row_1}\\
D{}_{\mathbf{Z}^{\ast}}f\left(\mathbf{Z}\right) & \triangleq\dfrac{{\partial}f\left(\mathbf{Z}\right)}{{\partial}\text{{vec}}\left(\mathbf{Z}^{\ast}\right)^{T}}=\dfrac{{1}}{2}\,\left[D_{\check{\mathbf{z}}_{R}}\check{f}(\check{\mathbf{z}})+j\cdot D_{\check{\mathbf{z}}_{I}}\check{f}(\check{\mathbf{z}})\right],\label{eq:Jacobians_f_row_2}
\end{align}
which provides the desired relationship between $D_{\check{\mathbf{z}}}\check{f}(\check{\mathbf{z}})$
and $D{}_{\mathbf{Z}}f\left(\mathbf{Z}\right)$ and $D{}_{\mathbf{Z}^{\ast}}f\left(\mathbf{Z}\right)$. 

Exploring the above equivalences, we can now readily prove Lemma \ref{Lemma_min_principle}
leveraging on standard  real calculus results. Given a real-valued
convex and continuously $\mathbb{R}$-differentiable function $f:\mathcal{K}\rightarrow\mathbb{R}$
on $\mathcal{K}$, the first-order Taylor expansion of $f(\mathbf{Z})=\check{f}(\check{\mathbf{z}})$
at $\mathcal{K}\ni\mathbf{Z}_{0}(\Leftrightarrow\check{\mathbf{z}}_{0})$
exists and it is given by: 
\begin{align}
f(\mathbf{\mathbf{Z}_{0}}+\Delta\mathbf{Z})-f(\mathbf{Z}_{0})\,= & \,\check{f}(\check{\mathbf{z}}_{0}+\Delta\check{\mathbf{z}})-\check{f}(\check{\mathbf{z}}_{0})\nonumber \\
\simeq & \, D_{\check{\mathbf{z}}}\check{f}(\check{\mathbf{z}}_{0})\cdot\text{\ensuremath{\Delta\check{\mathbf{z}}}}\label{eq:first_order_eq_r1}\\
= & \,2\,\text{{Re}}\left\{ D{}_{\mathbf{Z}}f\left(\mathbf{Z}_{0}\right)\text{{vec}}\left(\Delta\mathbf{Z}\right)\right\} \label{eq:first_order_eq_r3}\\
= & \,2\,\text{{Re}}\left\{ \text{{tr}}\left(\left(\nabla_{\mathbf{Z}}f(\mathbf{Z}_{0})\right)^{T}\Delta\mathbf{Z}\right)\right\} \label{eq:first_order_eq_r4}\\
= & \,\,2\,\left\langle \Delta\mathbf{Z},\,\nabla_{\mathbf{Z}^{\ast}}f(\mathbf{Z}_{0})\right\rangle ,\label{eq:first_order_eq_r5}
\end{align}
where (\ref{eq:first_order_eq_r1}) follows from the first-order Taylor
expansion of real-valued functions of real vectors (see, e.g., \cite{BertsekasNedicOzdaglar_book_convex03});
(\ref{eq:first_order_eq_r3}) follows from (\ref{eq:Jacobians_f_row_1});
(\ref{eq:first_order_eq_r4}) is due to $D{}_{\mathbf{Z}}f\left(\mathbf{Z}_{0}\right)=\text{{vec}}\left(\nabla_{\mathbf{Z}}f(\mathbf{Z}_{o})\right)^{T}$
and the property $\text{{vec}}\left(\mathbf{A}\right)^{T}\text{{vec}}\left(\mathbf{B}\right)=\text{{tr}}\left(\mathbf{A}^{T}\mathbf{B}\right)$
for any $\mathbf{A},\mathbf{B}\in\mathbb{C}^{n\times m}$; and in
(\ref{eq:first_order_eq_r5}) we used the fact that $f$ is real and
thus $\left(\nabla_{\mathbf{Z}}f(\mathbf{Z}_{o})\right)^{\ast}=\nabla_{\mathbf{Z}^{\ast}}f(\mathbf{Z}_{o})$.
This completes the proof. \hfill{}$\square$\vspace{-0.3cm}

\section{Complex Matrix Derivatives in Example \ref{Example_MIMO_complex_min_princ}\label{App: Example 24-Intermediate}}

We derive here the expressions of the (conjugate) derivatives used
in the Example \ref{Example_MIMO_complex_min_princ} and Example \ref{Example_MIMO_complex_min_princ}
revisited.

In order to obtain the expression of the augmented Hessian $\mathcal{H}_{\mathbf{Z}\mathbf{Z}^{\ast}}\tilde{{f}}(\mathbf{Z})$,
we need to compute $\nabla_{\mathbf{Z}\mathbf{Z}^{\ast}}^{2}\tilde{{f}}({\mathbf{Z}})$
$\triangleq D_{\mathbf{Z}}\left(\nabla_{\mathbf{Z}^{\ast}}\tilde{{f}}({\mathbf{Z}})\right)$
and $\nabla_{\mathbf{Z^{\ast}}\mathbf{Z}^{\ast}}^{2}\tilde{{f}}({\mathbf{Z}})\triangleq D_{\mathbf{Z}^{\ast}}\left(\nabla_{\mathbf{Z}^{\ast}}\tilde{{f}}({\mathbf{Z}})\right)$.
We preliminary compute $\nabla_{\mathbf{\mathbf{Z}}}{f}(\mathbf{Z})$
and $\nabla_{\mathbf{\mathbf{Z}^{\ast}}}{f}(\mathbf{Z})$.

Recalling that \cite[Prop. 3.12]{Are_book_MatrixDiff} $d\ln\det\left(\mathbf{Z}\right)=\textrm{Tr}\left(\mathbf{Z}^{-1}d\mathbf{Z}\right)$
for all\textcolor{black}{{} $\mathbf{Z}$ such that $\det\mathbf{Z}\neq0$,
}with $d\ln\det\left(\mathbf{Z}\right)$ being the (complex) differential
of $\ln\det\left(\mathbf{Z}\right)$, we have (up to a constant positive
factor)\vspace{-0.2cm}
\begin{eqnarray}
df(\mathbf{Z}) & = & \textrm{Tr}\left(\left(\mathbf{R}_{n}+\mathbf{HZH}^{H}\right)^{-1}\mathbf{H}d\mathbf{Z}\mathbf{H}^{H}\right)=\textrm{vec}^{T}\left(\left(\mathbf{H}^{H}\left(\mathbf{R}_{n}+\mathbf{HZH}^{H}\right)^{-1}\mathbf{H}\right)^{T}\right)\text{{vec}}(d\mathbf{Z})\vspace{-0.2cm}\label{eq:3-1}
\end{eqnarray}
which, using the identification rule as given in \cite[Table 3.3]{Are_book_MatrixDiff},
leads to the following Jacobian matrices of $f$:\vspace{-0.2cm}
\begin{eqnarray}
D_{\mathbf{Z}}f(\mathbf{Z}) & = & \textrm{vec}^{T}\left(\left(\mathbf{H}^{H}\left(\mathbf{R}_{n}+\mathbf{HZH}^{H}\right)^{-1}\mathbf{H}\right)^{T}\right)\quad\mbox{and}\quad D_{\mathbf{Z^{\ast}}}f(\mathbf{Z})=\mathbf{0},\label{eq:4}
\end{eqnarray}
 and thus {[}cf. (\ref{eq:Jacobian_scalar_f}){]}\vspace{-0.2cm}
\begin{equation}
\nabla_{\mathbf{Z}}f(\mathbf{Z})=\left(\mathbf{H}^{H}\left(\mathbf{R}_{n}+\mathbf{HZH}^{H}\right)^{-1}\mathbf{H}\right)^{T}\quad\mbox{and}\quad\nabla_{\mathbf{Z^{\ast}}}f(\mathbf{Z})=\mathbf{0}.\label{eq:5}
\end{equation}

It follows from (\ref{eq:5}) that\vspace{-0.2cm}
\begin{equation}
\nabla_{\mathbf{Z^{\ast}}}\tilde{{f}}({\mathbf{Z}})=\nabla_{\mathbf{\mathbf{Z}^{\ast}}}{f}(\mathbf{Z})+\nabla_{\mathbf{\mathbf{Z}^{\ast}}}{f}(\mathbf{Z})^{\ast}=\left(\nabla_{\mathbf{\mathbf{Z}}}{f}(\mathbf{Z})\right)^{\ast}=\mathbf{H}^{H}\left(\mathbf{R}_{n}+\mathbf{H}\mathbf{Z}^{H}\mathbf{H}^{H}\right)^{-1}\mathbf{H}.\label{eq:navla_f_Z}
\end{equation}

Given (\ref{eq:navla_f_Z}), we can now compute $\nabla_{\mathbf{Z}\mathbf{Z}^{\ast}}^{2}\tilde{{f}}({\mathbf{Z}})$
and $\nabla_{\mathbf{Z^{\ast}}\mathbf{Z}^{\ast}}^{2}\tilde{{f}}({\mathbf{Z}})$.
The differential of $\nabla_{\mathbf{Z^{\ast}}}\tilde{{f}}(\mathbf{Z})$
is:\vspace{-0.5cm}

\begin{eqnarray}
\text{{vec}}\left[d\left(\nabla_{\mathbf{Z^{\ast}}}\tilde{{f}}(\mathbf{Z}^{\ast})\right)\right] & = & \textrm{vec}\left[\mathbf{H}^{H}d\left(\mathbf{R}_{n}+\mathbf{H}\mathbf{Z}^{H}\mathbf{H}^{H}\right)^{-1}\mathbf{H}\right]\label{eq:8-1}\\
 & = & -\underset{\mathbf{G}(\mathbf{Z})\triangleq\mathbf{H}^{H}\left(\mathbf{R}_{n}+\mathbf{H}\mathbf{Z}^{H}\mathbf{H}^{H}\right)^{-1}\mathbf{H}}{\underbrace{\textrm{vec}\left[\mathbf{G}(\mathbf{Z})(d\mathbf{Z})^{H}\mathbf{G}(\mathbf{Z})\right]}}\label{eq:8-4}\\
 & = & -\left[\mathbf{G}(\mathbf{Z})^{T}\otimes\mathbf{G}(\mathbf{Z})\right]\textrm{vec}\left[(d\mathbf{Z}^{\ast})^{T}\right]\label{eq:8-5}\\
 & = & -\left[\mathbf{G}(\mathbf{Z})^{T}\otimes\mathbf{G}(\mathbf{Z})\right]\mathbf{K}_{n^{2}n^{2}}\textrm{vec}\left[(d\mathbf{Z}^{\ast})\right]\label{eq:Jacobian_nabla}
\end{eqnarray}
where in (\ref{eq:8-4}) we used the rule $d\mathbf{Z}^{-1}=-\mathbf{Z}^{-1}(d\mathbf{Z})\mathbf{Z}^{-1}$
\cite[Prop. 3.8]{Are_book_MatrixDiff}; (\ref{eq:8-5}) follows from
the property $\text{{vec}}(\mathbf{A}\mathbf{B}\mathbf{C})=(\mathbf{C}^{T}\otimes\mathbf{A})\,\text{{vec}}(\mathbf{B})$
\cite[Lemma 2.11]{Are_book_MatrixDiff}; and in the last equality
we introduced the commutation matrix $\mathbf{K}_{n^{2}n^{2}}$, which
is the $n^{2}\times n^{2}$ permutation matrix such that  $\text{{vec}}(\mathbf{A}^{T})=\mathbf{K}_{n^{2}n^{2}}\text{{vec}}(\mathbf{A})$
\cite[Def. 1.8]{Are_book_MatrixDiff}. 

It follows from (\ref{eq:Jacobian_nabla}) and the identification
rule \cite[Table 3.3]{Are_book_MatrixDiff} that 
\begin{eqnarray}
\nabla_{\mathbf{Z}^{*}\mathbf{Z}^{\ast}}^{2}\widetilde{f}(\mathbf{Z}) & = & \nabla_{\mathbf{Z}\mathbf{Z}^{*}}^{2}\widetilde{f}(\mathbf{Z})-\left[\mathbf{G}(\mathbf{Z})^{T}\otimes\mathbf{G}(\mathbf{Z})\right]\mathbf{K}_{n^{2}n^{2}}\quad\mbox{and}\quad\nabla_{\mathbf{Z}\mathbf{Z}^{*}}^{2}\widetilde{f}(\mathbf{Z})=\mathbf{0},\label{eq:9-1}
\end{eqnarray}
 which leads to the expression of the augmented Hessian $\mathcal{H}_{\mathbf{Z}\mathbf{Z}^{\ast}}\tilde{{f}}(\mathbf{Z})$
as given in (\ref{eq:augmented_Hessian_f_tilde}).

\section{Proof of Propositions \ref{monotonicity_complexVI} and \ref{VI_monotonicity_closed_sets}\label{sub:Proof-of-Proposition_monotonicity_complexVI}}

It is sufficient to prove only Proposition \ref{VI_monotonicity_closed_sets};
Proposition \ref{monotonicity_complexVI} is just a special case.
To do that, we need the following intermediate result. \vspace{-0.3cm}

\subsection{Mean-value theorem for functions of complex variables}

We provide here a version of the mean-value theorem that is suitable
for real-valued functions of complex matrices. We focus directly on
the specific function that we need to prove Proposition \ref{VI_monotonicity_closed_sets}.

Given a continuously $\mathbb{R}$-differentiable matrix function
$\mathbf{F}^{\mathbb{C}}:\mathcal{K}\rightarrow\mathbb{C}^{n\times m}$
on the convex and closed set $\mathcal{K}\subseteq\mathbb{C}^{n\times m}$
and a point $\Delta\mathbf{Y}\in\mathbb{C}^{n\times m}$, let us consider
the real-valued function of complex matrix variables 
\begin{equation}
g(\mathbf{Z})\triangleq\left\langle \Delta\mathbf{Y},\,\mathbf{\mathbf{F}^{\mathbb{C}}}(\mathbf{Z})\right\rangle .\label{eq:G_function}
\end{equation}
For every two points $\mathbf{Z}_{1},\mathbf{Z}_{2}\in\mathcal{K}$,
with $\Delta\mathbf{Z}\triangleq\mathbf{Z}_{2}-\mathbf{Z}_{1}$, let
$h(t):[0,1]\rightarrow\mathbb{R}$ be the real-valued scalar function,
defined as $[0,1]\ni t\mapsto h(t)\triangleq\left\langle \Delta\mathbf{Y},\,\mathbf{\mathbf{F}^{\mathbb{C}}}\left(\mathbf{Z}(t)\right)\right\rangle ,$
with $\mathbf{Z}(t)\triangleq\mathbf{Z}_{1}+t\,\Delta\mathbf{Z}$.
For some $\bar{{t}}\in(0,1)$, we have 
\begin{align}
g(\mathbf{Z}_{2})-g(\mathbf{Z}_{1}) & =h(1)-h(0)=h^{'}(\bar{{t}})\label{eq:mean_value_theorem}
\end{align}
where $h^{'}({t})$ is the first order derivative of $h(t)$ {[}note
that $h$ is continuously differentiable on $(0,1)${]}, and the last
equality in (\ref{eq:mean_value_theorem}) follows from the mean-value
theorem applied to the function $h(t)$. To compute $h^{'}(t)$ we
use the chain rule for complex matrix derivatives \cite{Are_book_MatrixDiff}
as shown next. Rewriting $h(t)$ as 
\begin{equation}
h(t)=\dfrac{{1}}{2}\,\text{{tr}}\left(\Delta\textbf{Y}^{H}\,\mathbf{\mathbf{F}^{\mathbb{C}}}\left(\mathbf{Z}(t)\right)\right)+\dfrac{{1}}{2}\,\text{{tr}}\left(\Delta\textbf{Y}^{T}\,\mathbf{\mathbf{F}^{\mathbb{C}}}\left(\mathbf{Z}(t)\right)^{\ast}\right)\vspace{-0.2cm}\label{eq:h_function_v2}
\end{equation}
and using\vspace{-0.2cm}
\[
\begin{array}{ll}
D{}_{\mathbf{Z}}\,\text{{tr}}\left(\Delta\textbf{Y}^{H}\,\mathbf{\mathbf{F}^{\mathbb{C}}}\left(\mathbf{Z}\right)\right)=\text{{vec}}\left(\Delta\textbf{Y}^{\ast}\right)^{T}D{}_{\mathbf{Z}}\mathbf{\mathbf{F}^{\mathbb{C}}}\left(\mathbf{Z}\right)\qquad & D{}_{\mathbf{Z^{\ast}}}\,\text{{tr}}\left(\Delta\textbf{Y}^{H}\,\mathbf{\mathbf{F}^{\mathbb{C}}}\left(\mathbf{Z}\right)\right)=\text{{vec}}\left(\Delta\textbf{Y}^{\ast}\right)^{T}D{}_{\mathbf{Z}^{\ast}}\mathbf{F}^{\mathbb{C}}\left(\mathbf{Z}\right)\bigskip\\
D{}_{\mathbf{Z}}\,\text{{tr}}\left(\Delta\textbf{Y}^{T}\,\mathbf{\mathbf{F}^{\mathbb{C}}}\left(\mathbf{Z}\right)^{\ast}\right)=\left(D{}_{\mathbf{Z^{\ast}}}\,\text{{tr}}\left(\Delta\textbf{Y}^{H}\,\mathbf{\mathbf{F}^{\mathbb{C}}}\left(\mathbf{Z}\right)\right)\right)^{\ast} & D{}_{\mathbf{Z}^{\ast}}\,\text{{tr}}\left(\Delta\textbf{Y}^{T}\,\mathbf{\mathbf{F}^{\mathbb{C}}}\left(\mathbf{Z}\right)^{\ast}\right)=\left(D{}_{\mathbf{Z}}\,\text{{tr}}\left(\Delta\textbf{Y}^{H}\,\mathbf{\mathbf{F}^{\mathbb{C}}}\left(\mathbf{Z}\right)\right)\right)^{\ast}
\end{array}
\]
 we have\vspace{-0.2cm} 
\begin{align}
h^{'}(t) & =D{}_{\mathbf{Z}(t)}h(t)\, D{}_{t}\mathbf{Z}(t)+D{}_{\mathbf{Z}(t)^{\ast}}h(t)\, D{}_{t}\mathbf{Z}(t)^{\ast}\medskip\label{eq:h_prime_row1}\\
 & =\dfrac{{1}}{2}\,\text{{vec}}\left(\Delta\textbf{Y}^{\ast}\right)^{T}D{}_{\mathbf{Z}}\mathbf{F}^{\mathbb{C}}\left(\mathbf{Z}(t)\right)\text{{vec}}\left(\Delta\mathbf{Z}\right)+\dfrac{{1}}{2}\,\text{{vec}}\left(\Delta\textbf{Y}\right)^{T}D{}_{\mathbf{Z}}\mathbf{\mathbf{F}^{\mathbb{C}}}\left(\mathbf{Z}(t)\right)^{\ast}\text{{vec}}\left(\Delta\mathbf{Z}\right)\nonumber \\
 & \quad+\dfrac{{1}}{2}\,\text{{vec}}\left(\Delta\textbf{Y}^{\ast}\right)^{T}D{}_{\mathbf{Z}^{\ast}}\mathbf{\mathbf{F}^{\mathbb{C}}}\left(\mathbf{Z}(t)\right)\text{{vec}}\left(\Delta\mathbf{Z}^{\ast}\right)+\dfrac{{1}}{2}\,\text{{vec}}\left(\Delta\textbf{Y}\right)^{T}D{}_{\mathbf{Z^{\ast}}}\mathbf{\mathbf{F}^{\mathbb{C}}}\left(\mathbf{Z}(t)\right)^{\ast}\text{{vec}}\left(\Delta\mathbf{Z}^{\ast}\right)\nonumber \\
 & =\dfrac{{1}}{2}\,\text{{vec}}\left(\left[\Delta\textbf{Y},\,\Delta\textbf{Y}^{\ast}\right]\right)^{H}\left[\begin{array}{ll}
D_{\mathbf{Z}}\mathbf{F}^{\mathbb{C}}(\mathbf{Z}(t)) & D_{\mathbf{Z}^{\ast}}\mathbf{F}^{\mathbb{C}}(\mathbf{Z}(t))\\
D_{\mathbf{Z}}\mathbf{F}^{\mathbb{C}}(\mathbf{Z}(t))^{\ast} & D_{\mathbf{Z}^{\ast}}\mathbf{F}^{\mathbb{C}}(\mathbf{Z}(t))^{\ast}
\end{array}\right]\text{{vec}}\left(\left[\Delta\textbf{Z},\,\Delta\textbf{Z}^{\ast}\right]\right),\label{eq:h_prime}
\end{align}
where in (\ref{eq:h_prime_row1}) we used the chain rule \cite[Th. 3.1]{Are_book_MatrixDiff}.
Using (\ref{eq:h_prime}) and the augmented Jacobian matrix $\mathbf{JF}^{\mathbb{C}}(\mathbf{Z})$
introduced in (\ref{eq:Augmented_Jacobian}), we can rewrite (\ref{eq:mean_value_theorem})
in a compact form as 
\begin{equation}
g(\mathbf{Z}_{2})-g(\mathbf{Z}_{1})=\,\text{{vec}}\left(\left[\Delta\textbf{Y},\,\Delta\textbf{Y}^{\ast}\right]\right)^{H}\mathbf{JF}^{\mathbb{C}}(\mathbf{Z}(\bar{{t}}))\,\text{{vec}}\left(\left[\Delta\textbf{Z},\,\Delta\textbf{Z}^{\ast}\right]\right).\label{eq:Delta_GZ_r5}
\end{equation}
which is the desired result.

\subsection{Proof of Proposition \ref{VI_monotonicity_closed_sets}}

We prove only (a)-(c); the proof of (d)-(e) follows similar steps.
\smallskip{}

\noindent \emph{Sufficiency} \emph{part}. For (a)-(c), it is enough
to prove only (c). Given two points $\mathbf{Z}_{1},\mathbf{Z}_{2}\in\mathcal{K}$,
let us define $\Delta\textbf{Z}\triangleq\mathbf{Z}_{2}-\mathbf{Z}_{1}$;
we have, for some $\bar{{t}}\in(0,1)$, 
\begin{align}
\left\langle \mathbf{Z}_{2}-\mathbf{Z}_{1},\,\mathbf{\mathbf{F}^{\mathbb{C}}}\left(\mathbf{Z}_{2}\right)-\mathbf{\mathbf{F}^{\mathbb{C}}}\left(\mathbf{Z}_{1}\right)\right\rangle  & =\text{{vec}}\left(\left[\Delta\textbf{Z},\,\Delta\textbf{Z}^{\ast}\right]\right)^{H}\mathbf{JF}^{\mathbb{C}}(\mathbf{Z}(\bar{{t}}))\,\text{{vec}}\left(\left[\Delta\textbf{Z},\,\Delta\textbf{Z}^{\ast}\right]\right),\label{eq:monotonocity_proof_1}
\end{align}
where the equality follows from (\ref{eq:Delta_GZ_r5}). Since $\mathbf{Z}_{1},\mathbf{Z}_{2}\in\mathcal{K}$,
we have that $\Delta\mathbf{Z}\in S_{\mathcal{K}}$; moreover $\mathbf{Z}(\bar{{t}})\in\mathcal{K}$
(due to the convexity of $\mathcal{K}$). It follows from (\ref{eq:monotonocity_proof_1})
that if there exists a constant $c_{\text{{sm}}}$ such that $\text{{vec}}\left([\mathbf{Y},\mathbf{Y}^{\ast}]\right)^{H}\mathbf{JF}^{\mathbb{C}}(\mathbf{Z})$
$\text{\text{{vec}}}\left([\mathbf{Y},\mathbf{Y}^{\ast}]\right)\geq$\textbf{\emph{$\,(c_{\text{{sm}}}/2)\,\left\Vert \mathbf{Y}\right\Vert _{F}^{2}$}}
for all $\mathbf{Y}\in\mathcal{S}_{\mathcal{K}}$ and $\mathbf{Z}\in\mathcal{K}$,
then $ $ $\mathbf{\mathbf{F}^{\mathbb{C}}}\left(\mathbf{Z}\right)$
is strongly monotone on $\mathcal{K}$. \smallskip{}

\noindent \emph{Necessity} \emph{part}. Let us focus on the strongly
monotonicity property only; monotonicity is obtained in a similar
way. Suppose that $\mathbf{\mathbf{F}^{\mathbb{C}}}\left(\mathbf{Z}\right)$
is strongly monotone on $\mathcal{K}$ with constant $c_{\text{{sm}}}>0$.
Let us show first that $\text{{vec}}\left([\mathbf{Y},\mathbf{Y}^{\ast}]\right)^{H}\mathbf{JF}^{\mathbb{C}}(\mathbf{Z})\text{{vec}}\left([\mathbf{Y},\mathbf{Y}^{\ast}]\right)\geq(c_{\text{{sm}}}/2)\,\left\Vert \mathbf{Y}\right\Vert _{F}^{2}$
for all $\mathbf{Y}\in S_{\mathcal{K}}$ and $\mathbf{Z}\in\text{{ri}}(\mathcal{K})$,
where $\text{{ri}}(\mathcal{K})$ denotes the relative interior of
$\mathcal{K}$ (see \cite[Ch. 1.4]{BertsekasNedicOzdaglar_book_convex03}
for the definition of $\text{{ri}}(\mathcal{K})$ and its main properties).
Then, we have 
\begin{align}
\text{{vec}}\left([\mathbf{Y},\mathbf{Y}^{\ast}]\right)^{H}\mathbf{JF}^{\mathbb{C}}(\mathbf{Z})\text{{vec}}\left([\mathbf{Y},\mathbf{Y}^{\ast}]\right) & =\dfrac{{1}}{2}\,\text{{vec}}\left([\mathbf{Y},\mathbf{Y}^{\ast}]\right)^{H}\lim_{t\downarrow0}\dfrac{{1}}{t}\left[\begin{array}{c}
\text{{vec}}\left(\mathbf{F}^{\mathbb{C}}(\mathbf{Z}+t\mathbf{Y})-\mathbf{F}^{\mathbb{C}}(\mathbf{Z})\right)\\
\text{{vec}}\left(\mathbf{F}^{\mathbb{C}}(\mathbf{Z}+t\mathbf{Y})-\mathbf{F}^{\mathbb{C}}(\mathbf{Z})\right)^{\ast}
\end{array}\right]\label{eq:strongly_monotonicity_F_complex_0}\\
 & =\dfrac{{1}}{2}\,\lim_{t\downarrow0}\dfrac{{1}}{t^{2}}\,\left\langle t\mathbf{Y},\mathbf{F}^{\mathbb{C}}(\mathbf{Z}+t\mathbf{Y})-\mathbf{F}^{\mathbb{C}}(\mathbf{Z})\right\rangle \label{eq:strongly_monotonicity_F_complex_1}\\
 & \geq\dfrac{c_{\text{{sm}}}}{2}\,\lim_{t\downarrow0}\dfrac{{1}}{t^{2}}\,\left\Vert t\mathbf{Y}\right\Vert _{F}^{2}=\dfrac{c_{\text{{sm}}}}{2}\,\left\Vert \mathbf{Y}\right\Vert _{F}^{2},\quad\forall\mathbf{Y}\in S_{\mathcal{K}}\,\mbox{and}\,\mathbf{Z}\in\text{{ri}}(\mathcal{K}),\label{eq:strongly_monotonicity_F_complex_2}
\end{align}
where the equality in (\ref{eq:strongly_monotonicity_F_complex_0})
follows from the ($\mathbb{R}$-)differentiability of $ $$\mathbf{\mathbf{F}^{\mathbb{C}}}\left(\mathbf{Z}\right)$
on $\mathcal{K}$ {[}(\ref{eq:strongly_monotonicity_F_complex_0})
can be proved using the same approach as in the proof of Lemma \ref{Lemma_min_principle}
but applied to $\text{{vec}}\left(\mathbf{F}^{\mathbb{C}}(\mathbf{Z})\right)${]};
in (\ref{eq:strongly_monotonicity_F_complex_1}) we used the definition
of inner product (\ref{eq:inner_product_matrix}); and in (\ref{eq:strongly_monotonicity_F_complex_2})
we used i) the fact that for sufficiently small $t>0$, $\mathbf{Z}+t\mathbf{Y}\in\mathcal{K}$
{[}since $\mathbf{Z}\in\text{{ri}}(\mathcal{K})${]}, and ii) the
strongly monotonicity of $\mathbf{F}^{\mathbb{C}}(\mathbf{Z})$ on
$\mathcal{K}$. 

Next, let $\mathbf{Z}\in\mathcal{K}$ but $\mathbf{Z}\notin\text{{ri}}(\mathcal{K})$;
by \cite[Proposition 1.4.1(a)]{BertsekasNedicOzdaglar_book_convex03}
there exists a sequence $\{\mathbf{Z}_{k}\}\subset\text{{ri}}(\mathcal{K})$
such that $\mathbf{Z}_{k}\rightarrow\mathbf{Z}$. By (\ref{eq:strongly_monotonicity_F_complex_2})
evaluated in each $\mathbf{Z}_{k}\in\text{{ri}}(\mathcal{K})$ and
the continuity of $\mathbf{JF}^{\mathbb{C}}(\mathbf{Z})$, it is follows
that 
\[
\text{{vec}}\left([\mathbf{Y},\mathbf{Y}^{\ast}]\right)^{H}\mathbf{JF}^{\mathbb{C}}(\mathbf{Z})\text{{vec}}\left([\mathbf{Y},\mathbf{Y}^{\ast}]\right)=\lim_{k\rightarrow\infty}\text{{vec}}\left([\mathbf{Y},\mathbf{Y}^{\ast}]\right)^{H}\mathbf{JF}^{\mathbb{C}}(\mathbf{Z}_{k})\text{{vec}}\left([\mathbf{Y},\mathbf{Y}^{\ast}]\right)\geq\dfrac{c_{\text{{sm}}}}{2}\,\left\Vert \mathbf{Y}\right\Vert _{F}^{2},
\]
for all $\mathbf{Y}\in S_{\mathcal{K}}.$ This completes the proof
of the necessity part. 

Note that Proposition \ref{VI_monotonicity_closed_sets} reduces to
Proposition \ref{monotonicity_complexVI} if the set $\mathcal{K}$
has nonempy interior. Indeed, when this happens, $\text{{Aff}}(\mathcal{K})=\mathbb{C}^{n\times m}$
and thus $S_{\mathcal{K}}=\mathbb{C}^{n\times m}$.

\section{Proof of Proposition \ref{Prop:MIMO_Game}\label{sec:Proof-of-Proposition_MIMOGame}}

Statement (a) follows from Proposition \ref{Lemma_NEP_complex_VI};
for (b) and (c), we prove only (b).

According to Proposition \ref{monotonicity_complexVI}(a), we need
to show that, under the assumption that $\boldsymbol{{\Upsilon}}_{\mathbf{F}^{\mathbb{C}}}^{\text{\texttt{{mimo}}}}$
in (\ref{eq:Upsilon_F_Q}) is positive semidefinite, the augmented
Jacobian matrix $\mathbf{JF}^{\mathbb{C}}(\mathbf{Q})$ associated
to $\mathbf{F}^{\mathbb{C}}(\mathbf{Q})$ in (\ref{eq:F_MIMO_G})
is augmented positive semidefinite on $\mathcal{P}^{\texttt{{mimo}}}$.
Given (\ref{eq:F_MIMO_G}), $D_{\mathbf{Q}^{\star}}\mathbf{F}^{\mathbb{C}}(\mathbf{Q})=\mathbf{0}$,
implying that $\mathbf{JF}^{\mathbb{C}}(\mathbf{Q})$ is a block diagonal
matrix: 
\begin{equation}
\mathbf{JF}^{\mathbb{C}}(\mathbf{Q})=\,\dfrac{{1}}{2}\,\left[\begin{array}{cc}
D_{\mathbf{Q}}\mathbf{F}^{\mathbb{C}}(\mathbf{Q}) & \mathbf{0}\\
\mathbf{0} & \left(D_{\mathbf{Q}}\mathbf{F}^{\mathbb{C}}(\mathbf{Q})\right)^{\ast}
\end{array}\right]\label{eq:JF_Q}
\end{equation}
with $D_{\mathbf{Q}}\mathbf{F}^{\mathbb{C}}(\mathbf{Q})$ given by
(see Appendix \ref{App: Example 24-Intermediate} for a similar computation):
\begin{align*}
D_{\mathbf{Q}}\mathbf{F}^{\mathbb{C}}(\mathbf{Q}) & =\left[\begin{array}{ccc}
D_{\mathbf{Q}_{1}}\mathbf{F}_{1}^{\mathbb{C}}(\mathbf{Q},) & \cdots & D_{\mathbf{Q}_{I}}\mathbf{F}_{1}^{\mathbb{C}}(\mathbf{Q})\\
\vdots & \ddots & \vdots\\
D_{\mathbf{Q}_{1}}\mathbf{F}_{I}^{\mathbb{C}}(\mathbf{Q}) & \cdots & D_{\mathbf{Q}_{Q}}\mathbf{F}_{I}^{\mathbb{C}}(\mathbf{Q})
\end{array}\right]=\left[\begin{array}{ccc}
\boldsymbol{{\Psi}}_{11}(\mathbf{Q})^{\ast}\otimes\boldsymbol{{\Psi}}_{11}(\mathbf{Q}) & \cdots & \boldsymbol{{\Psi}}_{1I}(\mathbf{Q})^{\ast}\otimes\boldsymbol{{\Psi}}_{1I}(\mathbf{Q})\\
\vdots & \ddots & \vdots\\
\boldsymbol{{\Psi}}_{I1}(\mathbf{Q})^{\ast}\otimes\boldsymbol{{\Psi}}_{I1}(\mathbf{Q}) & \cdots & \boldsymbol{{\Psi}}_{II}(\mathbf{Q})^{\ast}\otimes\boldsymbol{{\Psi}}_{II}(\mathbf{Q})
\end{array}\right]
\end{align*}
with 
\begin{equation}
\boldsymbol{{\Psi}}_{ij}(\mathbf{Q})\triangleq\mathbf{H}_{ii}^{H}\underset{\triangleq\mathbf{S}_{i}(\mathbf{Q})}{\underbrace{\left(\mathbf{R}_{n_{i}}+\sum\limits _{j=1}^{Q}\mathbf{H}_{ij}\mathbf{Q}_{j}\mathbf{H}_{ij}^{H}\right)^{^{-1}}}}\mathbf{H}_{ij}=\mathbf{H}_{ii}^{H}\mathbf{S}_{i}(\mathbf{Q})\mathbf{H}_{ij}.\label{eq:Psi_Q_ij}
\end{equation}
We will denote by $\boldsymbol{{\Psi}}_{ii}^{1/2}(\mathbf{Q})$ the
square root of the positive definite matrix $\boldsymbol{{\Psi}}_{ii}(\mathbf{Q})$
(recall that the channel matrices $\mathbf{H}_{ii}$ are assumed to
be full-column rank), i.e., $\boldsymbol{{\Psi}}_{ii}(\mathbf{Q})=\boldsymbol{{\Psi}}_{ii}^{H/2}(\mathbf{Q})\,\boldsymbol{{\Psi}}_{ii}^{1/2}(\mathbf{Q})$.

Therefore, $\mathbf{JF}^{\mathbb{C}}(\mathbf{Q})$ is augmented positive
semidefinite on $\mathcal{P}^{\texttt{{mimo}}}$ if and only if $D_{\mathbf{Q}}\mathbf{F}^{\mathbb{C}}(\mathbf{Q})$
is positive semidefinite on $\mathcal{P}^{\texttt{{mimo}}}$, or equivalently
the following matrix is so:\vspace{-0.2cm} 
\begin{equation}
\left[\begin{array}{ccc}
\mathbf{I} & \cdots & \left(\boldsymbol{{\Psi}}_{11}^{-H/2}\boldsymbol{{\Psi}}_{1I}\boldsymbol{{\Psi}}_{11}^{-1/2}\right)^{\ast}\otimes\left(\boldsymbol{{\Psi}}_{11}^{-H/2}\boldsymbol{{\Psi}}_{1I}\boldsymbol{{\Psi}}_{11}^{-1/2}\right)\\
\vdots & \ddots & \vdots\\
\left(\boldsymbol{{\Psi}}_{II}^{-H/2}\boldsymbol{{\Psi}}_{I1}\boldsymbol{{\Psi}}_{II}^{-1/2}\right)^{\ast}\otimes\left(\boldsymbol{{\Psi}}_{II}^{-H/2}\boldsymbol{{\Psi}}_{I1}\boldsymbol{{\Psi}}_{II}^{-1/2}\right) & \cdots & \mathbf{I}
\end{array}\right]\label{eq:Jacobian_MIMO_2}
\end{equation}
where for notational simplicity we omitted the dependence on $\mathbf{Q}$
and write $\boldsymbol{{\Psi}}_{ij}$, instead of $\boldsymbol{{\Psi}}_{ij}(\mathbf{Q})$.
The condensed matrix associated to (\ref{eq:Jacobian_MIMO_2}) is
the following $I\times I$ matrix\vspace{-0.2cm} 
\begin{equation}
\left[\begin{array}{ccc}
1 & \cdots & \left\Vert \left(\boldsymbol{{\Psi}}_{11}^{-H/2}\boldsymbol{{\Psi}}_{1Q}\boldsymbol{{\Psi}}_{11}^{-1/2}\right)^{\ast}\otimes\left(\boldsymbol{{\Psi}}_{11}^{-H/2}\boldsymbol{{\Psi}}_{1Q}\boldsymbol{{\Psi}}_{11}^{-1/2}\right)\right\Vert _{2}\\
\vdots & \ddots & \vdots\\
\left\Vert \left(\boldsymbol{{\Psi}}_{II}^{-H/2}\boldsymbol{{\Psi}}_{I1}\boldsymbol{{\Psi}}_{II}^{-1/2}\right)^{\ast}\otimes\left(\boldsymbol{{\Psi}}_{II}^{-H/2}\boldsymbol{{\Psi}}_{I1}\boldsymbol{{\Psi}}_{II}^{-1/2}\right)\right\Vert _{2} & \cdots & 1
\end{array}\right],\label{eq:Jacobian_MIMO_3}
\end{equation}
where $\left\Vert \mathbf{A}\right\Vert _{2}\triangleq\sqrt{{\rho(\mathbf{A}^{H}\mathbf{A})}}$
is the spectral norm of $\mathbf{A}$. Note that we can rewrite each
of the off-diagonal terms of (\ref{eq:Jacobian_MIMO_3}) as: with
$\widetilde{\boldsymbol{{\Psi}}}_{ij}\triangleq\boldsymbol{{\Psi}}_{ii}^{-H/2}\boldsymbol{{\Psi}}_{ij}\boldsymbol{{\Psi}}_{ii}^{-1/2}$,
\begin{align}
\left\Vert \left(\boldsymbol{{\Psi}}_{ii}^{-H/2}\boldsymbol{{\Psi}}_{ij}\boldsymbol{{\Psi}}_{ii}^{-1/2}\right)^{\ast}\otimes\left(\boldsymbol{{\Psi}}_{ii}^{-H/2}\boldsymbol{{\Psi}}_{ij}\boldsymbol{{\Psi}}_{ii}^{-1/2}\right)\right\Vert _{2} & =\left\Vert \widetilde{\boldsymbol{{\Psi}}}_{ij}^{\ast}\otimes\widetilde{\boldsymbol{{\Psi}}}_{ij}\right\Vert _{2}=\left[\rho\left(\widetilde{\boldsymbol{{\Psi}}}_{ij}^{T}\,\widetilde{\boldsymbol{{\Psi}}}_{ij}^{\ast}\otimes\widetilde{\boldsymbol{{\Psi}}}_{ij}^{H}\,\widetilde{\boldsymbol{{\Psi}}}_{ij}\right)\right]^{1/2}=\rho\left(\widetilde{\boldsymbol{{\Psi}}}_{ij}^{H}\,\widetilde{\boldsymbol{{\Psi}}}_{ij}\right),\label{eq:matrix_rho_equality}
\end{align}
where in the last equality we used the property $\rho\left(\mathbf{A}^{T}\mathbf{A}^{\ast}\otimes\mathbf{A}^{H}\mathbf{A}\right)=\rho\left(\mathbf{A}^{T}\mathbf{A}^{\ast}\right)\rho\left(\mathbf{A}^{H}\mathbf{A}\right)$
and the fact that the eigenvalues of $\mathbf{A}^{T}\mathbf{A}^{\ast}$
coincide with those of $\mathbf{A}^{H}\mathbf{A}$. Using (\ref{eq:matrix_rho_equality}),
we can now introduce the so-called comparison matrix $\boldsymbol{{\Upsilon}}_{\mathbf{F}^{\mathbb{C}}}^{\text{\texttt{{mimo}}}}(\mathbf{Q})$
associated to (\ref{eq:Jacobian_MIMO_3}) and defined as $ $ 
\[
\left[\boldsymbol{{\Upsilon}}_{\mathbf{F}^{\mathbb{C}}}^{\text{\texttt{{mimo}}}}(\mathbf{Q})\right]_{ij}\triangleq\left\{ \begin{array}{ll}
1, & \mbox{if }i=j\\
-\rho\left(\widetilde{\boldsymbol{{\Psi}}}_{ij}^{H}(\mathbf{Q})\,\widetilde{\boldsymbol{{\Psi}}}_{ij}(\mathbf{Q})\right) & \mbox{otherwise}.
\end{array}\right.
\]
It is indeed not difficult to see that if $\boldsymbol{{\Upsilon}}_{\mathbf{F}^{\mathbb{C}}}^{\text{\texttt{{mimo}}}}(\mathbf{Q})$
is positive semidefinite on $\mathcal{P}^{\texttt{{mimo}}}$ then
so is the matrix (\ref{eq:Jacobian_MIMO_3}) and thus also $D_{\mathbf{Q}}\mathbf{F}^{\mathbb{C}}(\mathbf{Q})$.
To conclude the proof, it is enough to show that $\boldsymbol{{\Upsilon}}_{\mathbf{F}^{\mathbb{C}}}^{\text{\texttt{{mimo}}}}(\mathbf{Q})\geq\boldsymbol{{\Upsilon}}_{\mathbf{F}^{\mathbb{C}}}^{\text{\texttt{{mimo}}}}$
for all $\mathbf{Q}\in\mathcal{P}^{\texttt{{mimo}}}$, where $\boldsymbol{{\Upsilon}}_{\mathbf{F}^{\mathbb{C}}}^{\text{\texttt{{mimo}}}}$
is defined in (\ref{eq:Upsilon_F_Q}) and the inequality has to be
intended component-wise. The latter properties indeed implies that
if $\boldsymbol{{\Upsilon}}_{\mathbf{F}^{\mathbb{C}}}^{\text{\texttt{{mimo}}}}$
is positive semidefinite then so is $\boldsymbol{{\Upsilon}}_{\mathbf{F}^{\mathbb{C}}}^{\text{\texttt{{mimo}}}}(\mathbf{Q})$
on $\mathcal{P}^{\texttt{{mimo}}}$. To this end, we focus next on
the off-diagonal terms $\rho(\widetilde{\boldsymbol{{\Psi}}}_{ij}^{H}(\mathbf{Q})\,\widetilde{\boldsymbol{{\Psi}}}_{ij}(\mathbf{Q}))$
of $\boldsymbol{{\Upsilon}}_{\mathbf{F}^{\mathbb{C}}}^{\text{\texttt{{mimo}}}}(\mathbf{Q})$
and prove that $|[\boldsymbol{{\Upsilon}}_{\mathbf{F}^{\mathbb{C}}}^{\text{\texttt{{mimo}}}}(\mathbf{Q})]_{ij}|\leq|[\boldsymbol{{\Upsilon}}_{\mathbf{F}^{\mathbb{C}}}^{\text{\texttt{{mimo}}}}]_{ij}|$
for all $\mathbf{Q}\in\mathcal{P}^{\texttt{{mimo}}}$ and $i\neq j$.
Denoting by $\mathbf{S}_{i}^{1/2}(\mathbf{Q})$ the square root of
the positive definite matrix $\mathbf{S}_{i}(\mathbf{Q})$ defined
in (\ref{eq:Psi_Q_ij}) {[}i.e., $\mathbf{S}_{i}(\mathbf{Q})=\boldsymbol{\mathbf{S}}_{i}^{H/2}(\mathbf{Q})\,\mathbf{S}_{i}^{1/2}(\mathbf{Q})${]},
and using $\widetilde{\boldsymbol{{\Psi}}}_{ij}(\mathbf{Q})\triangleq\boldsymbol{{\Psi}}_{ii}^{-H/2}(\mathbf{Q})\boldsymbol{{\Psi}}_{ij}(\mathbf{Q})\boldsymbol{{\Psi}}_{ii}^{-1/2}(\mathbf{Q})$,
we have the following chain of equalities/inequalities: for all $\mathbf{Q}\in\mathcal{P}^{\texttt{{mimo}}}$
and $i\neq j$, 
\begin{align}
\hspace{-0.2cm}\left|\left[\boldsymbol{{\Upsilon}}_{\mathbf{F}^{\mathbb{C}}}^{\text{\texttt{{mimo}}}}(\mathbf{Q})\right]_{ij}\right| & =\rho\left(\widetilde{\boldsymbol{{\Psi}}}_{ij}^{H}(\mathbf{Q})\,\widetilde{\boldsymbol{{\Psi}}}_{ij}(\mathbf{Q})\right)=\rho\left(\boldsymbol{{\Psi}}_{ij}^{H}(\mathbf{Q})\boldsymbol{{\Psi}}_{ii}^{-1}(\mathbf{Q})\boldsymbol{{\Psi}}_{ij}(\mathbf{Q})\boldsymbol{{\Psi}}_{ii}^{-1}(\mathbf{Q})\right)\nonumber \\
 & =\rho\left(\boldsymbol{{\Psi}}_{ii}^{-H/2}(\mathbf{Q})\mathbf{H}_{ij}^{H}\mathbf{S}_{i}(\mathbf{Q})\mathbf{H}_{ii}\left(\mathbf{H}_{ii}^{H}\mathbf{S}_{i}(\mathbf{Q})\mathbf{H}_{ii}\right)^{-1}\mathbf{H}_{ii}^{H}\mathbf{S}_{i}(\mathbf{Q})\mathbf{H}_{ij}\boldsymbol{{\Psi}}_{ii}^{-1/2}(\mathbf{Q})\right)\nonumber \\
 & =\rho\left(\boldsymbol{{\Psi}}_{ii}^{-H/2}(\mathbf{Q})\mathbf{H}_{ij}^{H}\mathbf{S}_{i}^{H/2}(\mathbf{Q})\,\right.\underset{\triangleq\mathbf{P}_{\mathcal{R}(\mathbf{H}_{ii})}\preceq\,\mathbf{I}}{\underbrace{\mathbf{S}_{i}^{1/2}(\mathbf{Q})\mathbf{H}_{ii}\left(\mathbf{H}_{ii}^{H}\mathbf{S}_{i}(\mathbf{Q})\mathbf{H}_{ii}\right)^{-1}\mathbf{H}_{ii}^{H}\mathbf{S}_{i}^{H/2}(\mathbf{Q})}}\,\left.\mathbf{S}_{i}^{1/2}(\mathbf{Q})\mathbf{H}_{ij}\boldsymbol{{\Psi}}_{ii}^{-1/2}(\mathbf{Q})\right)\label{eq:bounds_on_rho_1}\\
 & \leq\rho\left(\mathbf{S}_{i}^{1/2}(\mathbf{Q})\mathbf{H}_{ij}\boldsymbol{{\Psi}}_{ii}^{-1}(\mathbf{Q})\mathbf{H}_{ij}^{H}\mathbf{S}_{i}^{H/2}(\mathbf{Q})\right)=\rho\left(\mathbf{S}_{i}^{1/2}(\mathbf{Q})\mathbf{H}_{ij}\left(\mathbf{H}_{ii}^{H}\mathbf{S}_{i}(\mathbf{Q})\mathbf{H}_{ii}\right)^{-1}\mathbf{H}_{ij}^{H}\mathbf{S}_{i}^{H/2}(\mathbf{Q})\right)\label{eq:bounds_on_rho_2}\\
 & \leq\rho\left(\mathbf{S}_{i}^{1/2}(\mathbf{Q})\mathbf{H}_{ij}\mathbf{H}_{ii}^{\dagger}\mathbf{S}_{i}^{-1}(\mathbf{Q})\mathbf{H}_{ii}^{\dagger H}\mathbf{H}_{ij}^{H}\mathbf{S}_{i}^{H/2}(\mathbf{Q})\right)\label{eq:bounds_on_rho_3}\\
 & \leq\rho\left(\mathbf{S}_{i}^{-1}(\mathbf{Q})\right)\cdot\rho\left(\mathbf{S}_{i}(\mathbf{Q})\right)\cdot\rho\left(\mathbf{H}_{ii}^{\dagger H}\mathbf{H}_{ij}^{H}\mathbf{H}_{ij}\mathbf{H}_{ii}^{\dagger}\right)\label{eq:bounds_on_rho_4}\\
 & \leq\texttt{\text{{INNR}}}_{ij}\cdot\rho\left(\mathbf{H}_{ii}^{\dagger H}\mathbf{H}_{ij}^{H}\mathbf{H}_{ij}\mathbf{H}_{ii}^{\dagger}\right)=\left|\left[\boldsymbol{{\Upsilon}}_{\mathbf{F}^{\mathbb{C}}}^{\text{\texttt{{mimo}}}}\right]_{ij}\right|\label{eq:bounds_on_rho_5}
\end{align}
where in (\ref{eq:bounds_on_rho_1}), $\mathbf{P}_{\mathcal{R}(\mathbf{H}_{ii})}\triangleq\mathbf{S}_{i}^{1/2}(\mathbf{Q})\mathbf{H}_{ii}\left(\mathbf{H}_{ii}^{H}\mathbf{S}_{i}(\mathbf{Q})\mathbf{H}_{ii}\right)^{-1}\mathbf{H}_{ii}^{H}\mathbf{S}_{i}^{H/2}(\mathbf{Q})$
is the orthogonal projection onto the range space of $\mathbf{H}_{ii}$;
in (\ref{eq:bounds_on_rho_2}) we used the property of the projection
$\mathbf{P}_{\mathcal{R}(\mathbf{H}_{ii})}\preceq\,\mathbf{I}$ and
the spectral radius inequality $\rho(\mathbf{A}^{H}\mathbf{B}\mathbf{A})\leq\rho(\mathbf{A}^{H}\mathbf{C}\mathbf{A})$
for all $\mathbf{0}\preceq\mathbf{B}\preceq\mathbf{C}$; (\ref{eq:bounds_on_rho_3})
follows from the property $\left(\mathbf{X}^{H}\mathbf{A}\mathbf{X}\right)^{-1}\preceq\mathbf{X}^{\dagger}\mathbf{A}^{-1}\mathbf{X}^{\dagger H}$,
valid for all positive definite $n\times n$ matrices $\mathbf{A}$
and $n\times k$ full-column rank matrices $\mathbf{X}$; in (\ref{eq:bounds_on_rho_4})
we used $\mathbf{A}\preceq\rho\left(\mathbf{A}\right)\cdot\mathbf{I}$
and the spectral radius inequality as in (\ref{eq:bounds_on_rho_2});
and finally (\ref{eq:bounds_on_rho_5}) follows from $\rho\left(\mathbf{S}_{i}^{-1}(\mathbf{Q})\right)\cdot\rho\left(\mathbf{S}_{i}(\mathbf{Q})\right)\leq\texttt{\text{{INNR}}}_{ij}$,
with $\texttt{\text{{INNR}}}_{ij}$ defined in (\ref{eq:INNR_MIMO}).

The above chain of inequalities proves the desired relationship between
$\boldsymbol{{\Upsilon}}_{\mathbf{F}^{\mathbb{C}}}^{\text{\texttt{{mimo}}}}(\mathbf{Q})$
and $\boldsymbol{{\Upsilon}}_{\mathbf{F}^{\mathbb{C}}}^{\text{\texttt{{mimo}}}}$,
which completes the proof. \hfill{}$\square$

\global\long\def\baselinestretch{1}

{\small{ \bibliographystyle{IEEEtran}
\bibliography{scutari_refs}
 }} 
\end{document}